\documentclass[a4paper,USenglish,numberwithinsect]{lipics-preprint}
\usepackage{etex}
\usepackage{chngcntr}
\usepackage{apptools}
\AtAppendix{\counterwithin{theorem}{section}}
\newif\ifoutputappendix % when set to true, the appendix is output
\outputappendixtrue
\newif\ifignore % when set to true, additional text appears containing
\ignorefalse
\newcommand{\auxproof}[1]{
\ifignore\mbox{}\newline
\textbf{BEGIN: AUX-PROOF} \dotfill\newline
{#1}\mbox{}\newline
\textbf{END: AUX-PROOF}\dotfill\newline
\fi}

\usepackage{tikz}
\usepackage{tikz-qtree}
\usetikzlibrary{trees}

\usepackage[all]{xy}
\SelectTips{cm}{}  % Tips (of arrows) are in accordance with Computer Modern
\xyoption{v2}
\xyoption{curve}
\xyoption{2cell}
\SilentMatrices
\def\labelstyle{\scriptstyle}

\newdir{ >}{{}*!/-8pt/@{>}}
\newcommand{\kar}{\ar|-*\dir{|}}

\usepackage{mathtools}
\usepackage{stmaryrd}
\usepackage{wrapfig}

\theoremstyle{plain}
\newtheorem{mytheorem}[theorem]{Theorem}
\newtheorem{mylemma}[theorem]{Lemma}

\newtheorem{mysublemma}[theorem]{Sublemma}
\newtheorem{mycorollary}[theorem]{Corollary}

\theoremstyle{definition}
\newtheorem{mydefinition}[theorem]{Definition}

\theoremstyle{remark}
\newtheorem{myremark}[theorem]{Remark}
\newtheorem{myexample}[theorem]{Example}

\newtheorem{myassumptions}[theorem]{Assumptions}

\newcommand{\seq}[2]{\begingroup%
    \renewcommand{\i}{1}#1%
    ,\dotsc,%
    \renewcommand{\i}{#2}#1%
\endgroup}
\newcommand{\seqby}[3]{\begingroup%
    \renewcommand{\i}{#2}#1%
    ,\dotsc,%
    \renewcommand{\i}{#3}#1%
\endgroup}

\newcommand{\tuple}[1]{\langle #1 \rangle}
\newcommand{\iso}{\mathrel{\stackrel{
           \raisebox{.5ex}{$\scriptstyle\cong\,$}}{
           \raisebox{0ex}[0ex][0ex]{$\rightarrow$}}}}
\newcommand{\place}{\underline{\phantom{n}}\,} % place holder
\newcommand{\nat}{\mathbb{N}}

\newcommand{\Sets}{\mathbf{Sets}}
\newcommand{\Meas}{\mathbf{Meas}}
\newcommand{\C}{\mathbb{C}}
\newcommand{\id}{\mathrm{id}}
\newcommand{\bang}{\mbox{$!$}}
\newcommand{\pow}{\mathcal{P}}
\newcommand{\dist}{\mathcal{D}}
\newcommand{\giry}{\mathcal{G}}
\newcommand{\Kl}{\mathcal{K}\hspace{-.1em}\ell}
\newcommand{\oF}{\overline{F}}
\newcommand{\op}{\mathrm{op}}

\newcommand{\co}{\mathbin{\circ}}

\newcommand{\kco}{\mathbin{\odot}}

\newcommand{\relarrow}[1]{\mathrel{\ooalign{$#1$\crcr\hss\raisebox{.1ex}[0ex][0ex]{$\shortmid$}\hss}}}
\newcommand{\relto}{\relarrow{\rightarrow}}
\newcommand{\kto}{\relto}
\newcommand{\longrelto}{\relarrow{\longrightarrow}}
\newcommand{\longkto}{\longrelto}

\newcommand{\qnt}[2]{#1#2\mathpunct{.}}
\newcommand{\qnta}{\qnt{\forall}}

\newcommand{\lfp}{\operatorname{lfp}}
\newcommand{\sol}{\mathrm{sol}}

\newcommand{\X}{\mathcal{X}}
\newcommand{\Lang}{\mathrm{Lang}}
\newcommand{\sigalg}{\mathfrak{F}}
\newcommand{\FSigma}{F_{\Sigma}}

\newcommand{\tr}{\operatorname{\mathsf{tr}}}
\newcommand{\trinf}{\tr^{\infty}}
\newcommand{\trp}{\tr^{\mathrm{p}}}

\newcommand{\myTree}{\mathrm{Tree}}
\newcommand{\Run}{\mathrm{Run}}
\newcommand{\AccRun}{\mathrm{AccRun}}

\newcommand{\Dom}{\mathrm{Dom}}
\newcommand{\Branch}{\mathrm{Branch}}
\newcommand{\DelSt}{\mathrm{DelSt}}
\newcommand{\myroot}{\mathrm{rt}}

\newcommand{\AccProb}{\mathrm{AccProb}}
\newcommand{\Cyl}{\mathrm{Cyl}}
\newcommand{\NDL}{\mathrm{NoDiv}}
\newcommand{\DL}{\mathrm{Div}}

\newcommand{\hd}{\mathrm{hd}}
\newcommand{\tl}{\mathrm{tl}}
\newcommand{\accstate}{\raisebox{.65ex}{\xymatrix{*+<.95em>[o][F=]{}}}}
\newcommand{\nonaccstate}{\raisebox{.65ex}{\xymatrix{*+<.95em>[o][F-]{}}}}
\newcommand{\dlstate}{\spadesuit}
\newcommand{\dlletter}{\mathsf{o}}

\title{Coalgebraic Trace Semantics for B\"{u}chi and Parity Automata}
\author[1,2]{Natsuki Urabe}
\author[1]{Shunsuke Shimizu}
\author[1]{Ichiro Hasuo}

\affil[1]{Department of Computer Science, The University of Tokyo, Japan
}
\affil[2]{JSPS Research Fellow}

\authorrunning{N. Urabe, S. Shimizu and I. Hasuo} %mandatory. First: Use abbreviated first/middle names. Second (only in severe cases): Use first author plus 'et. al.'

\Copyright{Natsuki Urabe, Shunsuke Shimizu and Ichiro Hasuo} %mandatory, please use full first names. LIPIcs license is "CC-BY";  http://creativecommons.org/licenses/by/3.0/

% mandatory: Please choose ACM 1998 classifications from http://www.acm.org/about/class/ccs98-html . E.g., cite as "F.1.1 Models of Computation". 

 % mandatory: Please provide 1-5 keywords
\begin{document}

\maketitle

\begin{abstract}
Despite its success in producing numerous general results on state-based
 dynamics, the theory of \emph{coalgebra} has struggled to accommodate
 the \emph{B\"{u}chi acceptance condition}---a basic notion in the
 theory of automata for infinite words or trees. In this paper we present a clean answer to the question that builds on the ``maximality'' characterization of infinite traces (by Jacobs and C\^{\i}rstea): the accepted language of a B\"{u}chi automaton is characterized by two commuting diagrams, one for a \emph{least} homomorphism and the other for a \emph{greatest}, much like in a system of (least and greatest) fixed-point equations. This characterization works uniformly for the nondeterministic branching and the probabilistic one; and for words and trees alike. We present our results in terms of the \emph{parity} acceptance condition that generalizes B\"{u}chi's.
\end{abstract}

%\end{frontmatter}

\section{Introduction}
\label{sec:intro}
%\textbf{B\"{u}chi Automata}\quad
\subparagraph*{B\"{u}chi Automata}
\emph{Automata} are central to
   theoretical computer science. Besides their significance in formal
   language theory and as models of computation,  many
   \emph{formal verification} techniques rely on them, exploiting their balance between expressivity and 
   tractable complexity of operations on them. See
   e.g.~\cite{Vardi95aaa,GradelTW02ala}. Many current problems in
   verification are about
%   \emph{reactive} and  hence
   \emph{nonterminating} systems (like servers); for their analyses, naturally,
   automata that classify \emph{infinite} objects---such as infinite
   words and infinite trees---are employed. 

The \emph{B\"{u}chi acceptance condition} is the  simplest  nontrivial
    acceptance condition for automata for infinite objects. Instead
    of requiring finally reaching an accepting state $\accstate$---which
    makes little sense for infinite words/trees---it
    requires  accepting states  visited \emph{infinitely often}. 
   This simple condition, too, has proved  both expressive and 
   computationally tractable:  for the word case
    the B\"{u}chi condition can express any $\omega$-regular properties;
%    (but not for trees);
   and the emptiness problem for B\"{u}chi automata can be solved
    efficiently by searching for a \emph{lasso} computation.

%\vspace*{.3em}
%\noindent\textbf{Coalgebras} \quad
\subparagraph*{Coalgebras}
Studies of automata and state-based transition systems in general have
been shed a fresh \emph{categorical} light in 1990's, by
the theory of \emph{coalgebra}. 
% Following the studies of
% non-well-founded sets~\cite{Aczel88} and the principle of coinduction at
% its heart~\cite{BarwiseM96}, 
Its
%deceptively
simple modeling of
state-based dynamics---as a \emph{coalgebra}, i.e.\ an arrow $c\colon
X\to FX$ in a category $\C$---has produced numerous results
that capture mathematical essences and provide general techniques. Among its
basic results are: \emph{behavior-preserving maps} as 
homomorphisms; a \emph{final coalgebra} as a fully abstract domain of
behaviors; \emph{coinduction} (by finality) as definition and proof
principles; a general span-based definition of \emph{bisimulation}; etc. See
e.g.~\cite{Jacobs12itc,Rutten00uca}. More advanced 
results are on: coalgebraic modal logic (see
e.g.~\cite{CirsteaKPSV11mla}); process algebras and congruence formats (see
e.g.~\cite{Klin09bma}); generalization of Kleene's theorem
%  and various
% constructions on automata for finite words
(see e.g.~\cite{Silva15asi}); etc.

\auxproof{
Coalgebras are good at \emph{generalizing} existing results: a known
result, once expressed using a coalgebra $c\colon X\to FX$ in a
category $\C$, can apply to a variety of systems. This
happens by changing the endofunctor $F\colon\C\to\C$---that describes
the ``type of behaviors'' of the systems in question---and by changing the
base category $\C$.  For example, the same span diagram---with different
choices of $F$---instantiates to the common notion of bisimulation
between nondeterministic systems and to the one for \emph{probabilistic}
systems as well (see e.g.~\cite{Sokolova05cao}). Another example is the
modeling of \emph{name-passing processes}, where the base category $\C$
is changed from usual $\Sets$ to suitable presheaf category (or the
category of nominal sets). 
}

%\vspace*{.3em}
%\noindent\textbf{B\"{u}chi Automata, Coalgebraically} \quad
\subparagraph*{B\"{u}chi Automata, Coalgebraically}
In the coalgebra community, however, two  important phenomena in
automata and/or concurrency have been known to be hard
to model---many previous attempts have seen only limited success.
One is \emph{internal ($\tau$-)transitions} and \emph{weak (bi)similarity};
see e.g.\ recent~\cite{GoncharovP14cwb}.
The other one
% ,
% that is the current paper's topic,
is the  B\"{u}chi acceptance condition.

\begin{wrapfigure}[3]{r}{0pt}
\hspace{-1em}\small
\raisebox{-.2cm}[0pt][0cm]{
\begin{math}
 \vcenter{\xymatrix@C-1.4em@R=.7em{
  {FX}
     \ar[r]^-{Ff}
  &
  {FY}
  \\
  {X}
    \ar[r]_{f}
    \ar[u]^{c}
 &
  {Y}
    \ar[u]_{d}
}}
\end{math}
}
\end{wrapfigure}
Here is a (sketchy) explanation why these two
phenomena should be hard to model coalgebraically. The theory of coalgebra is centered around
\emph{homomorphisms}
as behavior-preserving maps; see the diagram on the right.  
Deep rooted in it is the idea of \emph{local matching}
between one-step transitions in  $c$ and those in
$d$. This is what fails in the two phenomena: in weak bisimilarity 
a one-step transition in $c$ is matched by a possibly multi-step
transition in $d$; and the  B\"{u}chi acceptance
condition---stipulating
that accepting states are visited \emph{infinitely often}, in the long
run---is utterly \emph{nonlocal}. 

There have been some works that study B\"{u}chi acceptance conditions
(or more general \emph{parity} or \emph{Muller}
conditions) in coalgebraic settings. One
is~\cite{CianciaV12saa}, where they rely on the lasso characterization of
nonemptiness and use  $\Sets^{2}$ as a  base category. Another line is on \emph{coalgebra automata} (see
e.g.~\cite{Venema06aaf}), where however  B\"{u}chi/parity/Muller
acceptance conditions reside outside the realm of
coalgebras.\footnote{More precisely: a coalgebra automaton is an
automaton (with B\"{u}chi/parity/Muller
acceptance conditions) that \emph{classifies} coalgebras (as
generalization of words and trees). A coalgebra automaton itself is
\emph{not} described as a coalgebra; nor is its acceptance condition.}
Inspired by these works, and also by our  work~\cite{HasuoSC16lpm}
on alternating fixed points and coalgebraic model checking, the current paper introduces
a coalgebraic modeling of B\"uchi and parity automata based on \emph{systems of
fixed-point equations}. 

%\vspace*{.3em}
%\noindent\textbf{Contributions} \quad
\subparagraph*{Contributions}
% \emph{Coalgebraic Modeling of B\"{u}chi Automata and Their Languages} 
% \quad
We present a clean answer to the  question of ``B\"{u}chi automata,
coalgebraically,'' relying on the previous work on coalgebraic
infinitary trace semantics~\cite{Jacobs04tsf,Cirstea10git} and 
%in terms of maximal morphisms
%lattice-theoretic characterization of solutions of
fixed-point equations~\cite{HasuoSC16lpm}. Our modeling, hinted
in~(\ref{eq:fromMaximalTractToBuechiTrace}), features: 1) accepting
states as a \emph{partition} of a state space;
%rather than as a label;
and 2) explicit use of $\mu$ and $\nu$---for least/greatest fixed
points---in
%commutative
 diagrams.
 We state our results for the \emph{parity} condition (that
generalizes
the B\"uchi one).
% , for characterization of accepted
% languages. 
\begin{equation}\label{eq:fromMaximalTractToBuechiTrace}
\vspace{-1.4em}
  \begin{tabular}{|c|}
  \hline
    \def\labelstyle{\textstyle}
    $\vcenter{
      \begin{xy}\footnotesize
        \xymatrix@R.8em@C+0em{
          {\overline{F}X} \kar@{->}[r]
          \ar@{}[dr]|{\color{red} =_{\nu}}
          &
          {\overline{F}Z}
          \\
          {X} \kar[u]_{c} \kar@{->}_{\trinf(c)}[r]
          &
          {Z} \kar[u]_{J\zeta}^{\cong}
        }
      \end{xy}
    }$
    \begin{tabular}{l}
      in a Kleisli\\ category $\Kl(T)$
    \end{tabular}
  \\
    \begin{minipage}[t]{.42\textwidth}
      \centering\footnotesize
      Characterization of languages under no (i.e.\ the trivial)
      acceptance condition~\cite{Jacobs04tsf,Cirstea10git}
    \end{minipage}
  \\\hline
  \end{tabular}
  \Longrightarrow
  \begin{tabular}{|c|}
  \hline
    \def\labelstyle{\textstyle}
    \begin{xy}\footnotesize
      \xymatrix@R=.8em@C+0em{
        {\overline{F}X} \kar[r]
        \ar@{}[rd]|{\color{blue}=_{\mu}}
        &
        {\overline{F}Z}
        \\
        {X_{1}} \kar[u]^{c_{1}} \kar[r]_{\trp(c_{1})}
        &
        {Z\mathrlap{\enspace}} \kar[u]_{J\zeta}^{\cong}
      }
    \end{xy}
    \;
    \def\labelstyle{\textstyle}
    \begin{xy}\footnotesize
      \xymatrix@R=.8em@C+0em{
        {\overline{F}X} \kar[r]
        \ar@{}[rd]|{\color{red}=_{\nu}}
        &
        {\overline{F}Z}
        \\
        {X_{2}} \kar[u]^{c_{2}} \kar[r]_{\trp(c_{2})}
        &
        {Z\mathrlap{\enspace}} \kar[u]_{J\zeta}^{\cong}
      }
    \end{xy}
  \\
    \begin{minipage}[t]{.40\textwidth}
      \centering\footnotesize
      Under the B\"{u}chi acceptance condition, with 
      $X_{1}=\{\nonaccstate\text{'s}\}$ and $X_{2}=\{\accstate\text{'s}\}$
    \end{minipage}
  \\\hline
  \end{tabular}
\end{equation}
Our framework is generic: its
leading examples are \emph{nondeterministic} and (generative) \emph{probabilistic} 
tree automata, with the B\"{u}chi/parity acceptance condition.

Our contributions are: 1) coalgebraic modeling of automata with the
B\"uchi/parity conditions; 2) characterizing their
accepted languages by diagrams with $\mu$'s and $\nu$'s ($\trp$
in~(\ref{eq:fromMaximalTractToBuechiTrace})); and 3) proving that
the characterization indeed captures the
conventional definitions.
% We establish that this categorical modeling indeed coincides
% with conventional definitions.
The last ``sanity-check'' proves to be
intricate in the probabilistic case, and our proof---relying on previous~\cite{Cirstea10git,Schubert09tcf}---identifies
the role of \emph{final sequences}~\cite{Worrell05otf} in
%the understanding of
probabilistic processes.

With explicit $\mu$'s and $\nu$'s---that specify in \emph{which}
homomorphism, among \emph{many} that exist, we are interested---we
depart from the powerful reasoning principle of \emph{finality}
(existence of a unique homomorphism). We believe this is a necessary
step forward, for the theory of coalgebra to take up long-standing
challenges like the B\"{u}chi condition and weak bisimilarity.
 Our characterization~(\ref{eq:fromMaximalTractToBuechiTrace})---although
it is not so simple as the uniqueness argument by finality---seems
useful, too: we have obtained some results on \emph{fair
simulation} notions between B\"{u}chi
automata~\cite{UrabeSH16buechiSimulationArXiv}, following the current work.

%\vspace*{.3em}
%\noindent
%\textbf{Organization of the Paper} \quad\
\subparagraph*{Organization of the Paper}
In~\S{}\ref{sec:prelim} we provide  backgrounds on: the
coalgebraic theory
of trace  in a \emph{Kleisli category}~\cite{Jacobs04tsf,Cirstea10git}
(where we explain the diagram on the left
in~(\ref{eq:fromMaximalTractToBuechiTrace})); and systems of fixed-point
equations. In~\S{}\ref{sec:coalgebraicModeling} we present a coalgebraic
modeling of B\"uchi/parity automata and their languages. Coincidence
with the conventional definitions is shown
in~\S{}\ref{sec:nondetParitySys}
for the  nondeterministic setting, and
in~\S{}\ref{sec:probParitySys} for the probabilistic one.

Most proofs are deferred to the appendix.

%\vspace*{.3em}
%\noindent\textbf{Future Work} \quad
\subparagraph*{Future Work}
Here we are based on the coalgebraic theory of trace and
simulation~\cite{PowerT97ecs,Jacobs04tsf,HasuoJS07gts,UrabeH15cit}; 
it has been developed under the \emph{trivial} acceptance condition
(any run that does not diverge, i.e.\ that does not come to a deadend, is
 accepted). The current paper is about accommodating the B\"uchi/parity
 conditions in the \emph{trace} part of the theory; for the
 \emph{simulation} part 
 we also have exploited the current results
 to obtain sound \emph{fair simulation} notions for nondeterministic
 B\"uchi tree automata and probabilistic B\"uchi word automata~\cite{UrabeSH16buechiSimulationArXiv}. 

On the practical side our future work mainly consists of proof methods
for \emph{trace/language inclusion}, a problem omnipresent in formal
verification. \emph{Simulations}---as one-step, local witnesses for
trace inclusion---have been often used as a sound (but not necessarily
complete) proof method that is computationally more tractable; with the
observations in~\cite{UrabeSH16buechiSimulationArXiv} we are
naturally interested in them. Possible directions are: synthesis of
\emph{simulation matrices} between finite systems by linear programming,
like in~\cite{UrabeH16qsb}; synthesis of simulations
by other optimization techniques for program verification (where problem
instances are infinite due to the integer type); and simulations as a
proof method in interactive theorem proving.

% On the theoretical side we aim at a unifying \emph{fibrational} picture
% of the coalgebraic theory of trace~\cite{Jacobs04tsf,HasuoJS07gts,UrabeH15cit}. In~(\ref{eq:fromMaximalTractToBuechiTrace})
% we see two fixed-points: the greatest (left) for 
% trivial acceptance; and another one (right) for B\"uchi
% acceptance. Accommodating the \emph{coalgebraic finite trace semantics}
% in~\cite{HasuoJS07gts}---where we use an initial algebra
% $\alpha\colon FA\iso A$
% in place of a final coalgebra $\zeta\colon Z\iso FZ$ and we  get
% proper finality---in this spectrum of fixed points is our first concrete
% goal. We plan to build on~\cite{HasuoCKJ13cpa} and investigate the use
% of fibrational \emph{comprehension}.

\section{Preliminaries}
\label{sec:prelim}
\subsection{Coalgebras in a Kleisli Category}
\label{subsec:coalgebra}
We assume some basic category theory, most of which is covered in~\cite{Jacobs12itc}.

The conventional coalgebraic modeling of
systems---as a function $X\to FX$---is known to capture
\emph{branching-time} semantics (such as
bisimilarity)~\cite{Jacobs12itc,Rutten00uca}. In contrast
accepted languages of B\"{u}chi automata  (with nondeterministic or
probabilistic branching)  constitute \emph{linear-time} semantics;
see~\cite{vanGlabbeek01tlt} for the so-called \emph{linear time-branching time spectrum}.

For the coalgebraic modeling of such linear-time semantics we 
follow the ``Kleisli modeling'' tradition~\cite{PowerT97ecs,Jacobs04tsf,HasuoJS07gts}.
Here a system is parametrized by a monad $T$ and an endofunctor $F$ on $\Sets$:
the former represents the \emph{branching type} while
the latter represents the \emph{(linear-time) transition type}; and 
a system is modeled as a function of the type $X\to TFX$.\footnote{
Another
 eminent approach to coalgebraic linear-time semantics is the
 \emph{Eilenberg-Moore} one (see e.g.~\cite{JacobsSS15tsv,AdamekBHKMS12acp}):
 notably
% in the former a system is expressed as $X\to TFX$ while 
in the latter
 a system is expressed as $X\to FTX$. 
The Eilenberg-Moore approach can be seen as a categorical
 generalization of \emph{determinization} or the \emph{powerset
 construction}. It is however not clear how determinization 
 serves our current goal (namely a coalgebraic modeling of the
 B\"{u}chi/parity
 acceptance conditions). 
% This however makes the approach hard to apply to
%  \emph{infinite} words or trees, since already for B\"{u}chi word
%  automata, it is known that deterministic ones are less expressive than
%  general, nondeterministic ones.
}

 A function $X\to TFX$ is nothing but an
\emph{$\oF$-coalgebra} $X\kto \oF X$ in the \emph{Kleisli category}
$\Kl(T)$---where $\oF$ is a suitable lifting of $F$.
%This allows us to 
This means we can
apply the standard coalgebraic machinery
 to linear-time behaviors,
by changing the base category from $\Sets$ to $\Kl(T)$.

%\begin{mydefinition}[Kleisli category $\Kl(T)$]
%\label{def:KleisliCategory}
A monad $T=(T,\eta,\mu)$ on a category $\mathbb{C}$ induces the \emph{Kleisli category} $\Kl(T)$.
The objects of $\Kl(T)$ are the same as $\mathbb{C}$'s; and for each pair
 $X,Y$  of objects,
the homset $\Kl(T)(X,Y)$ is given by  $\mathbb{C}(X,TY)$.
An arrow $f\in\Kl(T)(X,Y)$---that is  $X\to TY$ in
 $\mathbb{C}$---is called a \emph{Kleisli arrow} and is denoted by
 $f:X\kto Y$ for distinction.
Given two successive Kleisli arrows $f:X\kto Y$ and $g:Y\kto Z$, 
their \emph{Kleisli composition} is given by
 $\mu_Z\circ Tg\circ f:X\kto Z$ (where $\circ$ is composition in
 $\mathbb{C}$). This composition in $\Kl(T)$ is denoted by $g\odot f$ for distinction. 
The \emph{Kleisli inclusion} $J:\mathbb{C}\to\Kl(T)$ is defined by 
$J(X)=X$
%for $X\in\mathbb{C}$,
and $J(f)=\eta_{Y}\circ f:X\kto Y$.
 %for $f:X\to Y$.
%Here composition is denoted by $\odot$ instead of $\circ$ for the sake of distinction.
%Kleisli arrow $\kto$, Kleisli composition $\odot$ 
%\end{mydefinition}
% The above data indeed form a category. 
% For example, for each $X\in \Kl(T)$, the identity arrow $\id_{X}:X\kto X$ is given by $\eta_X:X\to TX$.

In this paper we mainly use two combinations of $T$ and $F$.
The first  is 
%a category $\Sets$ of sets, 
the \emph{powerset monad} $\pow$
and a polynomial functor on $\Sets$;
the second  is 
%a category $\Meas$ of measurable spaces and measurable function,
the \emph{(sub-)Giry monad}~\cite{Giry82cao} $\giry$ and
a polynomial functor on $\Meas$, 
the category of measurable spaces and measurable functions.
The \emph{Giry monad}~\cite{Giry82cao} is commonly used
for modeling (not necessarily discrete) probabilistic processes. We
shall use its ``sub'' variant; a \emph{subprobability measure} over
$(X,\sigalg_X)$ is a measure $\mu$ such that $0\le \mu(X)\le 1$ (we do
not require $\mu(X)= 1$).

%A coalgebraic system is parametrized by a category $\mathbb{C}$

%In this paper, we mainly use two combinations of 
%%$\mathbb{C}$, 
%$T$ and $F$. 
%The first combination is 
%%a category $\Sets$ of sets, 
%the \emph{powerset monad} $\pow$
%and a polynomial functor on $\Sets$.
%The second combination is 
%%a category $\Meas$ of measurable spaces and measurable function,
%the \emph{(sub-)Giry monad}~\cite{Giry82cao} $\giry$ and a polynomial functor on $\Meas$,
%a category of measurable spaces and measurable functions.

\begin{mydefinition}[$\pow,\giry$]\label{def:powersetMonadAndSubGiryMonad}
The \emph{powerset monad} $\pow$ on $\Sets$ is:
$\pow X=\{A\subseteq X\}$; $(\pow f)(A)=\{f(x) \mid x\in A\}$;
its unit is $\eta^{\pow}_X(x)=\{x\}$; and its multiplication is $\mu^{\pow}_{X}(M)=\bigcup_{A\in M}A$.

The \emph{sub-Giry monad} is a monad
$\giry=(\giry,\eta^{\giry},\mu^{\giry})$ on $\Meas$ such that
$\giry(X,\sigalg_X)=(\giry X, \sigalg_{\giry X})$, where $\giry X$ is
the set of all \emph{subprobability measures} on $(X,\sigalg_X)$, and
$\sigalg_{\giry X}$ is the smallest $\sigma$-algebra such that, for each
$S\in\sigalg_X$, the function $\text{ev}_S:\giry X\to[0,1]$ defined by
$\text{ev}_S(P)=P(S)$ is measurable.  Moreover,
$\eta^{\giry}_{(X,\sigalg_X)}(x)(S)$ is $1$ if $x\in S$ and $0$
otherwise (the \emph{Dirac} distribution), and
$\mu^{\giry}_{(X,\sigalg_X)}(\Psi)(S)=\int_{\giry (X,\sigalg_X)}
\text{ev}_S \,d\Psi$.
\end{mydefinition}

\begin{mydefinition}[polynomial functors on $\Sets$ and $\Meas$]\label{def:polynFunc}
A  \emph{polynomial functor} $F$ on $\Sets$  is defined by the BNF notation
$F\Coloneqq \id\mid A\mid F_1\times F_2\mid \coprod_{i\in I} F_i$.
Here $A\in\Sets$.

A \emph{(standard Borel) polynomial functor} $F$ on $\Meas$ is defined by the BNF notation
$F\Coloneqq \id\mid (A,\sigalg_A) \mid F_1\times F_2 \mid \coprod_{i\in I} F_i$. %\end{equation*} 
Here $I$ is countable, and
we require each constant $(A,\sigalg_A)\in\Meas$ be a \emph{standard Borel space}
(see e.g.~\cite{Doberkat09scl}).
The $\sigma$-algebra $\sigalg_{FX}$ associated to $FX$ is defined as usual,
with (co)product $\sigma$-algebras, etc.
$F$'s action on arrows is obvious.

A standard Borel polynomial functor shall often be called
simply a \emph{polynomial functor}.
\end{mydefinition}
The technical requirement of being standard Borel---meaning that it
arises from a \emph{Polish space}~\cite{Doberkat09scl}---will be used in the probabilistic
setting of~\S{}\ref{sec:probParitySys}; we
follow~\cite{Cirstea10git,Schubert09tcf} in its use.

There is a well-known correspondence between a polynomial functor and a
\emph{ranked alphabet}---a set $\Sigma$ with an \emph{arity map}
$|\place|\colon \Sigma\to \nat$. In this paper a functor $F$ (for the
linear-time behavior type) is restricted to be polynomial; this
essentially means that we are dealing with systems that generate
\emph{trees} over some ranked alphabet (with additional $T$-branching).

\begin{definition}[$\myTree_{\Sigma}$]\label{def:sigmaTree}
An \emph{(infinitary) $\Sigma$-tree}, as in the standard definition, is
a possibly infinite tree whose nodes are labeled with the ranked alphabet $\Sigma$ and
whose branching degrees are consistent with the arity of labels.
The set of $\Sigma$-trees is denoted by $\myTree_{\Sigma}$.
\end{definition}

\begin{mylemma}\label{lem:treeAndFinalCoalg}
Let $\Sigma$ be a ranked alphabet, and $F_{\Sigma}=\coprod_{\sigma\in\Sigma}(\place)^{|\sigma|}$ be
the corresponding polynomial functor on $\Sets$.
The set $\myTree_{\Sigma}$ of (infinitary) $\Sigma$-trees
carries a final $F_{\Sigma}$-coalgebra.
The same holds in $\Meas$, for countable $\Sigma$ and the corresponding
polynomial functor $F_{\Sigma}$.
\qed
\end{mylemma}
\noindent
We collect some standard notions and notations
for such trees in 
Appendix~\ref{sec:detailsOfTreeRunAcc}.

\begin{wrapfigure}[3]{r}[0pt]{0pt}
\hspace{-.4cm}
\raisebox{.3em}[0pt][0cm]{
 \begin{xy}\small
  \xymatrix@R=.6em@C-1.7em{
   {\Kl(T)}
     \ar[r]^{\oF}
  &
   {\Kl(T)}
  \\
   {\C}
     \ar[r]^{F}
     \ar[u]^{J}
  &
     {\C}
     \ar[u]_{J}
 }
 \end{xy}
 \hspace{-.6cm}\text{\begin{minipage}{1cm}\begin{equation}\label{eq:FAndOF}\end{equation}\end{minipage}}
}
\end{wrapfigure}
It is known~\cite{HasuoJS07gts,UrabeH15cit} that for
$(\mathbb{C},T)\in\{(\Sets,\pow),(\Meas,\giry)\}$ and polynomial
$F$ on $\mathbb{C}$, there is a canonical
\emph{distributive law}~\cite{Mulry94ltf}
$\lambda\colon FT\Rightarrow TF$---a natural transformation
compatible with $T$'s monad structure.
Such $\lambda$ induces a functor $\overline{F}\colon \Kl(T)\to\Kl(T)$
that makes the diagram~(\ref{eq:FAndOF}) commute.

Using this \emph{lifting} $\overline{F}$ of $F$ from $\C$ to $\Kl(T)$,
an arrow $c:X\to TFX$ in $\mathbb{C}$---that is
how we model an automaton---can be regarded as
an \emph{$\overline{F}$-coalgebra} $c\colon X\kto \overline{F}X$ in $\Kl(T)$.

Then the 
%state-based 
dynamics of $\mathcal{A}$---ignoring its initial and accepting states---is modeled as an $\overline{F}$-coalgebra $c:X\kto \overline{F}X$ in $\Kl(\pow)$
where: $F=\{a,b\}\times(\place)$, $X=\{x_1,x_2\}$ and $c:X\to \pow FX$
is the function  $c(x_1)=c(x_2)=\{(a,x_1),(b,x_2)\}$.
% condition, as a $TF$-coalgebra
The information on initial and accepting states is redeemed
%in our modeling
 later in
 % Note that the initial state and the accepting states of $\mathcal{A}$ is not modeled yet.
 % An approach to model them is given in
 \S{}\ref{subsec:coalgModeling}.
%\end{myexample}
%\end{minipage}
%\begin{minipage}{0.2\hsize}
%\end{minipage}
\vspace{1mm}

%\noindent\begin{minipage}{0.75\hsize}
%\renewcommand\windowpagestuff{
%\vspace{-2cm}
%\smash{
\begin{wrapfigure}[5]{r}[0pt]{0pt}
\hspace{-.4cm}
\raisebox{-1.3cm}[0pt][0cm]{
\tiny
\begin{xy}
  (0,0)*+[Fo]{x_1} = "x1",
  (0,10)*+[Fo]{x_2} = "x2",
  \ar (0,-5);"x1"*+++{}
  \ar @(ur,dr)^{a,\frac{1}{2}} "x1";"x1"*++{}
  \ar @/^2mm/^{b,\frac{1}{2}} "x1";"x2"*++{}
  \ar @(dr,ur)_{b,\frac{1}{2}} "x2";"x2"*++{}
  \ar @/^2mm/^{a,\frac{1}{2}} "x2";"x1"*++{}
\end{xy}
}
\end{wrapfigure}
%\opencutright
\noindent\begin{minipage}{0.175\hsize}\begin{myexample}\label{ex:NAAsTFCoalg}\end{myexample}\end{minipage}%
Let $\mathcal{M}$ be the Markov chain on the right.
 The dynamics of $\mathcal{M}$ is modeled as an $\overline{F}$-coalgebra $c:X\kto \overline{F}X$ in $\Kl(\giry)$ 
where: 
%$F=\left(\{a,b\},\pow(\{a,b\})\right)\times(\place)$, 
$F=\{a,b\}\times(\place)$,
%is the same as in Example~\ref{ex:NAAsTFCoalg},
$X=\{x_1,x_2\}$ with the discrete measurable structure,
%$X=\left(\{x_1,x_2\},\pow(\{x_1,x_2\})\right)$ 
and $c\colon X\to \giry FX$ is the (measurable) 
function defined by
$c(x)\bigl\{(a,x_{1})\bigr\}
=c(x)\bigl\{(b,x_{2})\bigr\}
=1/2
$, and
$c(x)\bigl\{(d,x')\bigr\}=0
$ for the other $(d,x')\in \{a,b\}\times X$.
% \begin{displaymath}
% c(x)\bigl(\{(d,x')\}\bigr)=\begin{cases}
% \frac{1}{2} & (\text{if }(d,x')\in\{(a,x_1),(b,x_2)\}) \\
% 0 & (\text{otherwise}) 
% \end{cases}
% \end{displaymath}
% for each $x\in\{x_1,x_2\}$ and $d\in\{a,b\}$.
% A Markov chain as a $TF$-coalgebra 
% 
% ...
%

Later we will equip Markov chains with accepting states and obtain
\emph{(generative) probabilistic B\"{u}chi automata}.
Their probabilistic accepted languages will be our subject of study.
% \end{cutout}
%\end{myexample}
%\end{minipage}
%\begin{minipage}{0.25\hsize}
%\quad
%\begin{xy}
%(0,0)*+[Fo]{x_1} = "x1",
%(0,12)*+[Fo]{x_2} = "x2",
%\ar (0,-5);"x1"*+++{}
%\ar @(ur,dr)^{a,\frac{1}{2}} "x1";"x1"*++{}
%\ar @/^3mm/^{b,\frac{1}{2}} "x1";"x2"*++{}
%\ar @(dr,ur)_{b,\frac{1}{2}} "x2";"x2"*++{}
%\ar @/^3mm/^{a,\frac{1}{2}} "x2";"x1"*++{}
%\end{xy}
%\end{minipage}

\begin{myremark}\label{rem:deadendProb}
 Due to the use of the \emph{sub-}Giry monad is that, in
 $f\colon X\kto Y$ in $\Kl(\giry)$, the probability
 $f(x)(Y)$ can be  smaller than $1$. The missing  $1-f(x)(Y)$ 
 is understood as that for \emph{divergence}.  In the
 nondeterministic case $f\colon X\kto Y$ in $\Kl(\pow)$ \emph{diverges} at $x$
 if $f(x)=\emptyset$. 

 This is in contrast with a system coming to halt generating a $0$-ary 
 symbol (such as $\checkmark$ in~(\ref{eq:infTraaa}) later); this is deemed as \emph{successful
 termination}. 
\end{myremark}

\subsection{Coalgebraic Theory of Trace}
\label{subsec:coalgTrSim}
The above ``Kleisli'' coalgebraic modeling  has produced 
some general results on: linear-time process semantics (called 
\emph{trace semantics}); and \emph{simulations} as 
witnesses of trace inclusion, generalizing the  theory
in~\cite{LynchV95fab}.
Here we review the former; it underpins our
developments later. A rough summary is in Table~\ref{table:traceAndSim}:
%They are very roughly summarized as in Table~\ref{table:traceAndSim};
typically the results apply to $T\in \{\pow,\dist,\giry\}$---where
$\dist$ is the \emph{subdistribution monad} on $\Sets$, a
discrete variant of $\giry$---and
polynomial $F$.
In what follows we present these previous results in precise terms,
sometimes strengthening the assumptions for the sake of presentation.
% This will set an appropriate technical
% context for the current work;
The current paper's goal is to incorporate the
B\"{u}chi acceptance condition in
% (the latter row of)
(the right column of)
Table~\ref{table:traceAndSim}. 

\begin{table}[tbp]\setlength{\tabcolsep}{0pt}
  \begin{tabular}{l||c|c}
    Semantics
    &
    Finite trace
    &
    Infinitary trace
  \\\hline
    \multirow{2}{*}{%
      \vbox{\hbox{Coalgebraic\hspace{.5em}}\hbox{modeling}}%
    }%
    &
    \hspace{1em}
    \begin{xy}
      \scriptsize\def\labelstyle{\textstyle}
      \xymatrix@R=1.4em@C=1.2em{
        {\overline{F}X} \ar@{}[drr]|{=} \kar@{-->}[rr]^{\overline{F}(\tr(c))}
        & &
        {\overline{F}A}
        \\
        {X} \kar[u]^{c} \kar@{-->}_{\tr(c)}[rr]
        & &
        {A} \kar[u]_{J\alpha^{-1}}_(.8){\hspace{1em}\text{final}}^{\cong}
      }
    \end{xy}
    \begin{minipage}{2em}
      \begin{equation}\label{eq:finiteTrace}\end{equation}
    \end{minipage}
    &
    \hspace{1em}
    \begin{xy}
      \scriptsize\def\labelstyle{\textstyle}
      \xymatrix@R=1.4em@C=1.2em{
        {\overline{F}X} \ar@{}[drr]|{\color{red}=_{\nu}} \kar@{->}[rr]^{\overline{F}(\trinf(c))}
        & &
        {\overline{F}Z}
        \\
        {X} \kar[u]^{c} \kar@{->}_{\trinf(c)}[rr]
        & &
        {Z} \kar[u]_{J\zeta}_(.8){\hspace{1em}\substack{\text{weakly}\\\text{final}}}^{\cong}
      }
    \end{xy}
    \begin{minipage}{2em}
      \begin{equation}\label{eq:infiniteTrace}\end{equation}
    \end{minipage}
  \\
    &
    \hspace{.5em}Finality in $\Kl(T)$ (Thm.~\ref{thm:finalityinKlT})\hspace{.5em}
    &
    \hspace{.5em}(Weak finality + maximality) in $\Kl(T)$ (Thm.~\ref{thm:infinitaryTrace})
  \end{tabular}
 \caption{Overview of existing results on coalgebraic trace semantics.
 %The diagrams are in the Kleisli category $\Kl(T)$.
 }
\label{table:traceAndSim}
\vspace{-1.7em}
\end{table}

Firstly, \emph{finite trace semantics}---linear-time behaviors that
eventually terminate, such as the accepted languages
of \emph{finite} words for NFAs---is captured by finality in $\Kl(T)$.

%\noindent\begin{minipage}{0.7\hsize}
\begin{mytheorem}[{\cite{HasuoJS07gts}}]\label{thm:finalityinKlT}
Let $T\in \{\pow,\dist\}$ and $F$ be a polynomial functor on $\Sets$.
An initial $F$-algebra $\alpha\colon FA\iso A$ in $\Sets$
yields a final $\oF$-coalgebra in $\Kl(T)$, as in~(\ref{eq:finiteTrace})
 in Table~\ref{table:traceAndSim}. 
\qed
\end{mytheorem}
% \end{minipage}
% \begin{minipage}{0.3\hsize}
%  \begin{equation}\label{eq:finiteTrace}
%                \def\labelstyle{\textstyle}
%   \begin{xy}
%   \scriptsize
%   \xymatrix@R=1.3em@C=1.8em{
%   {\overline{F} X} \ar@{}[drr]|{=} \kar@{-->}[rr]^{\overline{F}(\tr(c))} & & {\overline{F}A}  \\
%   {X} \kar[u]_{c}  \kar@{-->}_{\tr(c)}[rr]  & & {A} \kar[u]_{J\alpha^{-1}}_(.8){\hspace{1em}\text{final}}^{\cong}
%   }
% \end{xy}
% \end{equation}
% \end{minipage}
%\vspace{1mm}
 The carrier $A$ of an initial $F$-algebra in $\Sets$ is given by finite
words/trees (over the alphabet that corresponds to $F$). The
significance of Thm.~\ref{thm:finalityinKlT} is that: for many examples,
the
unique homomorphism
$\tr(c)$ induced by finality~(\ref{eq:finiteTrace}) captures the finite
trace semantics of the system $c$. Here the word ``finite'' means that
we collect only behaviors that eventually terminate. 

What if we are also interested in \emph{nonterminating} behaviors, like
the infinite word $b^{\omega}=bbb\ldots$ accepted by the automaton in
Example~\ref{ex:NAAsTFCoalg}? 
There is a categorical characterization of such \emph{infinitary trace
semantics} too, although proper finality is now lost. 

\begin{mytheorem}[{\cite{Jacobs04tsf,Cirstea10git,UrabeH15cit}}]\label{thm:infinitaryTrace}
%\end{minipage}%
% Assume that $T=\pow$ and $F$ is a polynomial functor on $\Sets$, or 
% $T=\giry$ and $F$ 
   Let $(\mathbb{C},T)\in \{(\Sets,\pow),(\Meas,\giry)\}$ and $F$ be a polynomial functor on $\mathbb{C}$.  
%    where 
%   $\mathbb{C}$ is $\Sets$ if $T=\pow$ and $\Meas$ if $T=\giry$.
%   the base category of $T$ ($\Sets$ or $\Meas$).
   %$\Sets$. 
 A final $F$-coalgebra $\zeta\colon Z\iso FZ$ in $\mathbb{C}$
 gives rise to a weakly final $\oF$-coalgebra in $\Kl(T)$,
 as in~(\ref{eq:infiniteTrace}) in
 Table~\ref{table:traceAndSim}. Moreover, the coalgebra $J\zeta$
	 additionally admits the \emph{greatest} homomorphism $\trinf(c)$ with respect to 
 the pointwise order $\sqsubseteq$ in the homsets of $\Kl(T)$
  (given by inclusion for $T=\pow$, and by 
  pointwise $\leq$ on subprobability measures
%  $\le$ on real numbers 
  for
 $T=\giry$). That is: for each homomorphism $f$ from $c$ to $J\zeta$ we
 have
 $f\sqsubseteq \trinf(c)$.
 \qed
 \end{mytheorem}

\begin{wrapfigure}[3]{r}[0pt]{0pt}
  \hspace{-.3cm}
  \raisebox{-.9cm}[0pt][0cm]{
  \begin{xy}
    (0,0)*{\circ} = "x1",
    (10,0)*+{\checkmark} = "x2",
    (5,8)*+{\text{\begin{minipage}{1cm}\begin{equation}\label{eq:infTraaa}\end{equation}\end{minipage}}}="",
    \ar (-5,0);"x1"
    \ar @(ur,ul)_{a} "x1";"x1"
    \ar "x1";"x2"
  \end{xy}
}
\end{wrapfigure}
% In the setting of Thm.~\ref{thm:infinitaryTrace} the largest
% homomorphism from $c$ to $J\zeta$ is denoted by $\trinf(c)$;
In many
examples the greatest homomorphism $\trinf(c)$ captures the infinitary trace semantics of the system $c$. 
(Here by \emph{infinitary}  we mean \emph{both finite
and infinite} behaviors.) For example, for the system 
(\ref{eq:infTraaa})
%illustrated
% on  the right
%\begin{displaymath}
% one state, with a loop, and checkmark, without double circles
%\end{displaymath}
where $\checkmark$ denotes successful termination, its finite trace
semantics is $\{\varepsilon, a,aa,\dotsc\}$ whereas its infinitary trace
semantics is $\{\varepsilon, a,aa,\dotsc\}\cup\{a^{\omega}\}$. The
latter is captured by the diagram~(\ref{eq:infiniteTrace}), with
$T=\pow$ and $F=\{\checkmark\}+\{a\}\times (\place)$.

% The ultimate goal of the current work is to incorporate \emph{B\"{u}chi
% acceptance conditions} in the picture of Table~\ref{table:traceAndSim}.
% The goal is not yet fully achieved: for infinitary trace
% semantics\footnote{Note that B\"{u}chi conditions do not make sense in
% finite trace semantics, much like they do not for finite words.} we do
% have a clean characterization that generalizes the above
% Thm.~\ref{thm:infinitaryTrace}, see~\S{}\ref{sec:traceSem}; for
% simulations however we have obtained only limited results. We believe
% that each of these partial results presented in~\S{}\ref{sec:fwdFairSimForBoth}---although not totally satisfactory
% from the coalgebraic viewpoint---is worth spelling out. They provide
% simulation-based proof techniques that are previously unknown; they do
% so for a variety of systems such as nondeterministic B\"{u}chi tree
% automata and finite-state probabilistic B\"{u}chi word automata. 
% These partial yet concrete results do exploit coalgebras' power of
% abstraction, as described in~\S{}\ref{sec:coalgDiscuss}. 

\subsection{Equational Systems for Alternating Fixed Points}
\label{subsec:eqsys} 
\emph{Nested, alternating} greatest and least fixed
points---as in a $\mu$-calculus formula $\nu u_{2}. \mu u_{1}.\, (p\land
u_{2})\lor \Box u_{1}$---are omnipresent in specification and
verification. For their relevance to the B\"uchi/parity acceptance condition
one can recall 
the well-known translation of LTL formulas to B\"{u}chi automata and
vice versa (see e.g.~\cite{Vardi95aaa}). 
To express such fixed points we follow~\cite{CleavelandKS92fmc,ArnoldN01rom} and
use
\emph{equational systems}---we prefer them to the textual $\mu$-calculus-like
presentations. 
\begin{mydefinition}[equational system]
\label{def:eqSys}
Let $L_{1},\dotsc,L_{n}$ be posets. An \emph{equational system} $E$
over $L_{1},\dotsc,L_{n}$ is an expression 
\begin{equation}\label{eq:sysOfEq}
 \begin{array}{c}
  u_{1}=_{\eta_{1}}f_{1}(u_{1},\dotsc, u_{n})\enspace,\quad
  \dotsc,\quad
  u_{n}=_{\eta_{n}} f_{n}(u_{1},\dotsc, u_{n})
 \end{array}
\end{equation}
where: $u_{1},\dotsc,u_{n}$ are \emph{variables},
 $\eta_{1},\dotsc,\eta_{n}\in\{\mu,\nu\}$, and 
$f_{i}\colon L_{1}\times\cdots\times L_{n}\to L_{i}$ is a monotone
 function.
% for each $i\in [1,m]$.
 A variable $u_{j}$ is  a \emph{$\mu$-variable} if
  $\eta_{j}=\mu$; it is a \emph{$\nu$-variable} if $\eta_{j}=\nu$.
 % We say $u_{i}$ has a \emph{bigger priority} than $u_{j}$ if $j<i$.

The \emph{solution} of the equational system $E$ is defined as follows,
under the assumption that $L_{i}$'s have enough supremums and infimums.
It proceeds as:
1) we solve the first equation to obtain an interim solution
$u_{1}=l^{(1)}_{1}(u_{2},\dotsc,u_{n})$;
2) it is used in the second equation to eliminate $u_{1}$ and yield a new equation
$u_{2}=_{\eta_{2}}f^{\ddagger}_{2}(u_{2},\dotsc,u_{n})$;
3) solving it again gives an interim solution
$u_{2}=l^{(2)}_{2}(u_{3},\dotsc,u_{n})$;
4) continuing this way from left to right eventually eliminates all variables and
leads to a closed solution $u_{n}=l^{(n)}_{n}\in L_{n}$; and
5) by propagating these closed solutions back from right to left,
we obtain closed solutions for all of $u_{1},\dotsc,u_{n}$.
A precise definition is found in Appendix~\ref{sec:appendixEqSys}.
\end{mydefinition}

\noindent It is important that the order of equations \emph{matters}: 
for $(u=_{\mu}v, v=_{\nu} u)$ the solution is $u=v=\top$ while 
for $(v=_{\nu} u, u=_{\mu}v)$ the solution is $u=v=\bot$.

Whether a solution is well-defined depends on how ``complete''
the posets $L_{1},\dotsc,L_{n}$ are.
It suffices if they are \emph{complete lattices},
in which case every monotone function $L_{i}\to L_{i}$ has
greatest/least fixed points (the \emph{Knaster-Tarski} theorem).
This is used in the nondeterministic setting:
note that $\pow Y$, hence the homset $\Kl(\pow)(X,Y)$, are complete lattices.
\begin{mylemma}\label{lem:eqSysSolvableCompLat}
The system $E$~(\ref{eq:sysOfEq}) has a solution if each $L_{i}$ is a complete lattice.
\qed
\end{mylemma}

\noindent
This does not work in the probabilistic case, since
%For the probabilistic case we need to argue differently: 
the homsets $\Kl(\giry)(X,Y)=\Meas(X,\giry Y)$ with the pointwise
order---on which we consider equational
systems---are not complete lattices. For example $\giry Y$ lacks the
greatest element in general; even if $Y=1$ (when $\giry
1\cong [0,1]$), the homset $\Kl(\giry)(X,1)$ can
fail to be a complete lattice. See Example~\ref{example:KlGiryX1IsNotACompleteLattice}.
Our strategy is: 1) to apply the 
following \emph{Kleene}-like result to the homset $\Kl(\giry)(X,1)$; and 2) to ``extend''
fixed points in $\Kl(\giry)(X,1)$ along a final $F$-sequence. See~\S{}\ref{subsec:fixedPtProbParitySys} later.

% It does not mean, however, $L_{i}$ which is not a complete lattice is ruled out.
% In fact, if $f_{i}$ is $\omega$-continuous and $\omega^{\op}$-continuous,
% $L_{i}$ is required to be only a pointed $\omega$/$\omega^{\op}$-cpo,
% as stated in the following lemma.
% Its proof is in Appendix~\ref{sec:omittedProofs}.
\begin{mylemma}\label{lem:eqSysSolvableOmegaCont}
The equational system $E$~(\ref{eq:sysOfEq})  has a solution
if: each $L_{i}$ is both a pointed $\omega$-cpo and a pointed $\omega^{\op}$-cpo;
and each $f_{i}$ is both $\omega$-continuous and $\omega^{\op}$-continuous.
\qed
\end{mylemma}

\auxproof{
\noindent
Here $L_{i}$ is a \emph{pointed $\omega$-cpo} if
% $L_{i}$ 
it has the least element $\bot$ and
the supremum $\bigsqcup_{j<\omega}l_{j}$ of any $\omega$-chain $l_{0} \sqsubseteq l_{1} \sqsubseteq \cdots$; and
$f_{i}\colon L_{1}\times\cdots L_{n}\to L_{i}$ is \emph{$\omega$-continuous} if
\begin{math}
  \textstyle
  \bigsqcup_{j<\omega}f(\seq{l_{\i,j}}{n})
  =
  f(\,\bigsqcup_{j<\omega}l_{1,j},\,\dotsc,\bigsqcup_{j<\omega}l_{n,j}\,)
  \,.
\end{math}
\emph{Pointed $\omega^{\op}$-cpo} and \emph{$\omega^{\op}$-continuity}
are defined 
analogously, using the opposite order.
}

In Appendix~\ref{sec:appendixEqSys} we have  additional lemmas on 
``homomorphisms'' of 
equational systems and preservation of solutions. 
They play important roles in the proofs of the later results.

\section{Coalgebraic Modeling of Parity Automata and Its Trace Semantics}
\label{sec:coalgebraicModeling}
Here we present our modeling of B\"uchi/parity automata.
We shall do so axiomatically with parameters $\C$, $T$ and $F$---much like
in~\S{}\ref{subsec:coalgebra}--\ref{subsec:coalgTrSim}.
Our examples cover: both nondeterministic and probabilistic branching; and
automata for trees (hence words as a special case).

\begin{myassumptions}\label{asm:assumptionsOnCTF}
In what follows a monad $T$ and an endofunctor $F$, both on $\mathbb{C}$, satisfy:
\begin{itemize}
  \item The base category $\mathbb{C}$ has a final object $1$ and finite coproducts.
  \item The functor $F$ has a final coalgebra $\zeta\colon Z\to FZ$ in $\mathbb{C}$.
  \item There is a \emph{distributive law} $\lambda\colon FT\Rightarrow TF$~\cite{Mulry94ltf},
        hence $F:\mathbb{C}\to\mathbb{C}$ is lifted to $\overline{F}\colon \Kl(T)\to\Kl(T)$. See~(\ref{eq:FAndOF}).
  \item For each $X,Y\in\Kl(T)$, the homset $\Kl(T)(X,Y)$ carries an order
        $\sqsubseteq_{X,Y}$ (or simply $\sqsubseteq$).
  \item Kleisli composition $\odot$ and  cotupling $[\place,\place]$ are
        monotone with respect to the order $\sqsubseteq$.
        The latter gives rise to an order isomorphism
        $\Kl(T)(X_{1}+X_{2},Y)\cong\Kl(T)(X_{1},Y)\times \Kl(T)(X_{2},Y)$, where
        $+$ is inherited along a left adjoint $J\colon \C\to\Kl(T)$.
  \item $\overline{F}:\Kl(T)\to\Kl(T)$ is \emph{locally monotone}:
        for $f,g\in\Kl(T)(X,Y)$, $f\sqsubseteq g$ implies $\overline{F}f\sqsubseteq\overline{F}g$.
\end{itemize}
\end{myassumptions}

\begin{myexample}\label{ex:nondet}
The category $\Sets$, the powerset monad $\pow$ (Def.~\ref{def:powersetMonadAndSubGiryMonad}) and
a polynomial functor $F$ on $\Sets$ (Def.~\ref{def:polynFunc}) satisfy Asm.~\ref{asm:assumptionsOnCTF}.
Here for $X,Y\in\Kl(\pow)$, an order $\sqsubseteq_{X,Y}$ is defined by:
$f\sqsubseteq g$ if $f(x)\subseteq g(x)$ for each $x\in X$.
\end{myexample}

\begin{myexample}\label{ex:probability}
The category $\Meas$, the sub-Giry monad $\giry$ (Def.~\ref{def:powersetMonadAndSubGiryMonad}) and
a polynomial functor $F$ on $\Meas$ (Def.~\ref{def:polynFunc}) satisfy Asm.~\ref{asm:assumptionsOnCTF}.
For $X,Y\in\Kl(\giry)$, a natural order $\sqsubseteq_{(X,\sigalg_X),(Y,\sigalg_Y)}$ is defined by:
$f\sqsubseteq g$ iff $f(x)(A)\le g(x)(A)$ (in $[0,1]$) for each $x\in X$ and $A\in\sigalg_Y$.
\end{myexample}

\subsection{Coalgebraic Modeling of B\"uchi/Parity Automata}
\label{subsec:coalgModeling}
The B\"uchi and parity acceptance conditions have been big challenges to
the coalgebra community, because of their \emph{nonlocal} and
\emph{asymptotic} nature (see~\S{}\ref{sec:intro}).
One possible modeling is to take the distinction between $\nonaccstate$ vs.\
$\accstate$---or different priorities in the parity case---as \emph{state labels}.
This is much like in the established coalgebraic modeling of deterministic automata
as $2\times (\place)^{\Sigma}$-coalgebras (see e.g.~\cite{Rutten00uca,Jacobs12itc}).
Here the set $2$ tells if a state is accepting or not.

A key to our current modeling,
however, is that accepting states should rather be specified by a
\emph{partition} $X=X_{1}+X_{2}$ of a state space, with
$X_{1}=\{\nonaccstate\text{'s}\}$ and $X_{2}=\{\accstate\text{'s}\}$.
 This idea  smoothly generalizes to  \emph{parity}  conditions, too, by $X_{i}=\{\text{states of priority $i$}\}$. Equipping such
partitions to coalgebras (with explicit initial states, as
in~\S{}\ref{subsec:coalgTrSim}) leads to the following.

Henceforth we state results for the parity condition, with B\"uchi
being a special case.

\begin{mydefinition}[parity $(T,F)$-system]\label{def:parityTFSys}
A \emph{parity $(T,F)$-system} is given by a triple
\begin{math}
  \X
  =
  \bigl(\,
    (X_{1},\dotsc,X_{n}),\,
    c\colon X\kto \overline{F}X,\,
    s\colon 1\kto X
  \,\bigr)
\end{math}
where  $n$ is a positive integer, and:
\begin{itemize}
 \item $(X_{1},\dotsc,X_{n})$ is an $n$-tuple of objects in $\C$
   for \emph{states} (with their \emph{priorities}),
   and we define $X=X_{1}+\cdots+ X_{n}$ (a coproduct in $\C$);
 \item $c\colon X \kto \overline{F}X$ is an arrow in $\Kl(T)$ for \emph{dynamics}; and
 \item $s\colon 1\kto X$ is an arrow in $\Kl(T)$ for \emph{initial states}.
\end{itemize}
For each $i\in[1,n]$
we define $c_{i}\colon X_{i}\kto \overline{F}X$ to be
the restriction $c\circ \kappa_{i}\colon X_{i}\kto \overline{F}X$ along the
coprojection $\kappa_{i}\colon X_{i}\hookrightarrow X$,
in case the maximum priority is $n=2$, a parity $(T,F)$-system is
referred to as a \emph{B\"uchi $(T,F)$-system}.
%Here the intuition is that $X_{2}$ is the collection of accepting states $\accstate$. 
\end{mydefinition}

\auxproof{
The intuitions are as follows.
The object  $X_{i}$ is the collection of states with the priority $i\in [1,m]$; and
$X:=X_{1}+\cdots + X_{m}$ is the collection of all states.
The arrow $c$ specifies the dynamics of
the system---with linear-time behaviors of type $F$ and additional
branching of type $T$---and
 % their cotupling 
% \begin{displaymath}
%  [c_{1},\dotsc,c_{m}]
%  \;\colon \;
%  X_{1}+\cdots + X_{m} = X \longrightarrow 
%  TFX
% \end{displaymath}
is a \emph{$TF$-coalgebra} with a carrier object $X$. Finally, much
like in~\cite{Hasuo06gfa,UrabeH14gfa,UrabeH15cit}, we
prefer to work with explicit initial states (unlike most coalgebraic
studies).
They are expressed as the arrow $s\colon 1\to TX$; the
occurrence of $T$ in the codomain means that $T$-branching is present
also in the choice of an initial state.

In other words, the notion of parity $(T,F)$-system is that of
$(T,F)$-system (in e.g.~\cite{Hasuo10gfa}) additionally
equipped with a partition $(X_{1},\dotsc,X_{m})$ of the state space $X$.

In what follows we will prefer the following presentation of parity $(T,F)$-systems
in the Kleisli category $\Kl(T)$:
\begin{equation}\label{eq:parityTFSysInKleisliCat}
 \X \;=\;
 \bigl(\,
  (X_{1},X_{2},\dotsc,X_{m}),
\;
%  (\,c_{i}\colon X_{i}\kto FX\,)_{i\in[1,m]}
c\colon X\kto FX
,
\;
  s\colon 1\kto X
\,\bigr)
\end{equation}
 where  $\kto$ denotes an arrow in $\Kl(T)$
 (see \S{}\ref{subsec:coalgebra}).
}

\subsection{Coalgebraic Trace Semantics under the Parity Acceptance Condition}
On top of the modeling in Def.~\ref{def:parityTFSys} we characterize
\emph{accepted languages}---henceforth referred to as \emph{trace
semantics}---of parity $(T,F)$-systems. \sloppy We use 
systems of fixed-point equations; this naturally extends the previous
characterization of infinitary traces (i.e.\ under the trivial acceptance
conditions) by maximality
(Thm.~\ref{thm:infinitaryTrace}; see also~(\ref{eq:fromMaximalTractToBuechiTrace})).
%~\cite{Jacobs04tsf,Cirstea10git,UrabeH15cit}.

% The (accepted) language of a B\"uchi automaton is captured as
% follows, in coalgebraic terms, with the help of equational systems. The
% characterization naturally extends the previous observations of
% infinitary traces as greatest coalgebra
% homomorphisms~\cite{Jacobs04tsf,Cirstea10git,kerstan13coalgebraictrace,UrabeH15cit}.

\begin{mydefinition}[trace semantics of parity $(T,F)$-systems]
\label{def:acceptedLangParityTFSys}
Let
$\X=\bigl(\,(X_{1},\dotsc,X_{n}),
%\,c\colon X\kto FX,\, s\colon 1\ktoX\,
c,s\bigr)$
be a parity $(T,F)$-system.
It induces the following equational system $E_{\X}$,
where $\zeta\colon Z \iso FZ$ is a final coalgebra in $\C$
(see Asm.~\ref{asm:assumptionsOnCTF}).
The variable $u_{i}$ ranges over the poset $\Kl(T)(X_{i},Z)$.
\begin{equation*}
  E_{\X}
  \;\coloneqq\;
  \left[
    \begin{array}{rllr}
      u_{1}
      &=_{\mu}&
      (J\zeta)^{-1}
      \odot
      \oF[u_{1},\dotsc,u_{n}]
      \odot
      c_{1}
      & \in\Kl(T)(X_{1},Z)
      \\
      u_{2}
      &=_{\nu}&
      (J\zeta)^{-1}
      \odot
      \oF[u_{1},\dotsc,u_{n}]
      \odot
      c_{2}
      & \in\Kl(T)(X_{2},Z)
      \\
      & \;\vdots &
      \\
      u_{n}
      &=_{\eta_{n}}&
      (J\zeta)^{-1}
      \odot
      \oF[u_{1},\dotsc,u_{n}]
      \odot
      c_{n}
      & \in\Kl(T)(X_{n},Z)
    \end{array}
  \right]
\end{equation*}
Here $\eta_{i}=\mu$ if $i$ is odd and $\eta_{i}=\nu$ if $i$ is even.
%We note that $[u_{1},u_{2},\cdots,u_{n}]\colon X\kto Z$ is the cotupling of $u_{1}$ and $u_{2}$.
The functions in the equations are seen to be monotone, thanks to the
monotonicity assumptions on cotupling, $\oF$ and $\odot$
(Asm.~\ref{asm:assumptionsOnCTF}).

We say that $(T,F)$ constitutes a $\emph{parity trace situation}$, if
$E_{\X}$ has a solution for any parity $(T,F)$-system $\X$, 
 denoted by
\begin{math}
  \trp_{1}(\X)\colon X_{1} \kto Z
,
 \dotsc
,
  \trp_{n}(\X)\colon X_{n} \kto Z
\end{math}.
The composite
\begin{align*}
  \trp(\X)
  \;\coloneqq\;
  \bigl(\,
    1
    \stackrel{s}{\longkto}
    X=X_{1}+X_{2}+\cdots+X_{n}
    \relarrow{\xrightarrow{[\trp_{1}(\X),\trp_{2}(\X),\dotsc,\trp_{n}(\X)]}}
    Z
  \,\bigr)
\end{align*}
is called the \emph{trace semantics} of the parity $(T,F)$-system $\X$.
\end{mydefinition}
\noindent
If $\X$ is a B\"uchi $(T,F)$-system,
the equational system $E_{\X}$---with their
solutions $\trp_{1}(\X)$ and $\trp_{2}(\X)$ in place---can be expressed as
the following diagrams (with explicit $\mu$ and $\nu$). See~(\ref{eq:fromMaximalTractToBuechiTrace}).
\begin{equation}\label{eq:diagramsForEqSysForBuechiAcceptance}
  \vcenter{\xymatrix@C+3.3em@R=.8em{
    {FX}
      \kar[r]^{\overline{F}[\trp(c_{1}),\trp(c_{2})]}
      \ar@{}[rd]|{\color{blue}=_{\mu}}
    &
    {FZ}
    \\
    {X_{1}}
      \kar[u]^{c_{1}}
      \kar[r]_{\trp(c_{1})}
     &
    {Z \mathrlap{\enspace}}
      \kar[u]_{J\zeta}^{\cong}
  }}
  \quad
  \vcenter{\xymatrix@C+3.3em@R=.8em{
    {FX}
       \kar[r]^{\overline{F}[
       \trp(c_{1}),\trp(c_{2})
       ]}
       \ar@{}[rd]|{\color{red}=_{\nu}}
    &
    {FZ}
    \\
    {X_{2}}
      \kar[u]^{c_{2}}
      \kar[r]_{\trp(c_{2})}
    &
    {Z \mathrlap{\enspace}}
      \kar[u]_{J\zeta}^{\cong}
  }}
\end{equation}

\section{Coincidence with the Conventional Definition: Nondeterministic}
\label{sec:nondetParitySys}
The rest of the paper is devoted to showing that our coalgebraic
characterization (Def.~\ref{def:acceptedLangParityTFSys}) indeed captures
the conventional definition of accepted languages.
In this section we study the nondeterministic case; we let $\C=\Sets$,
$T=\pow$, and $F$ be a polynomial functor.

We first have to check that Def.~\ref{def:acceptedLangParityTFSys} makes sense.
Existence of enough fixed points is obvious because $\Kl(\pow)(X_{i},Z)$
is a complete lattice (Lem.~\ref{lem:eqSysSolvableCompLat}).
See also Example~\ref{ex:nondet}.
\begin{mytheorem}\label{thm:constParityTraceSituPow}
  $T=\pow$ and a polynomial $F$ constitute a parity trace situation (Def.~\ref{def:acceptedLangParityTFSys}).
  \qed
\end{mytheorem}

Here is the conventional definition of automata~\cite{GradelTW02ala}. 
\begin{mydefinition}[NPTA]\label{def:nondetParityTreeAutom}
A \emph{nondeterministic parity tree automaton} (NPTA) is  a quadruple
\begin{displaymath}
  \X
  \;=\;
  \bigl(\,
   (X_{1},\ldots, X_{n}),\,
   \Sigma,\,
   \delta\colon X\to \pow\bigl(\textstyle\coprod_{\sigma\in\Sigma} X^{|\sigma|}\bigr),\,
   s\in \pow X
  \,\bigr)
  \,,
\end{displaymath}
where $X=X_1+\cdots+X_n$,
each $X_{i}$ is the set of \emph{states} with the \emph{priority} $i$,
$\Sigma$ is a ranked alphabet (with the arity map $|\place|\colon \Sigma\to\nat$),
$\delta$ is a \emph{transition function} and
$s$ is the set of \emph{initial states}.
\end{mydefinition}

The accepted language of an NPTA $\X$ is conventionally defined in the following way.
Here we are sketchy due to the lack of space;
precise definitions are  in Appendix~\ref{sec:detailsOfTreeRunAcc}.
% Before the definition of its accepted language,
% we go on to briefly review the standard notions of \emph{run}, \emph{accepting run},
% and so on.
% Precise definitions are found in Appendix~\ref{sec:detailsOfTreeRunAcc}.

%\begin{mydefinition}\label{def:treesRunsAccRunsBriefly}
A (possibly infinite) $(\Sigma \times X)$-labeled tree $\rho$ is
a \emph{run} of an NPTA $\X=(\vec{X},\Sigma,\delta,s)$ if:
for each node with a label $(\sigma,x)$,
it has $|\sigma|$ children and we have
$\bigl(\sigma,(x_{1},\dotsc,x_{|\sigma|})\bigr)\in \delta(x)$ where
$x_{1},\dotsc,x_{|\sigma|}$ are the $X$-labels of its children.
For a pedagogical reason we do not require the root $X$-label to be an initial state.
A run $\rho$ of an NPTA $\X$ is  \emph{accepting} if
any infinite branch $\pi$ of the tree $\rho$ satisfies the parity acceptance condition
(i.e.\ $\max\{i\mid\pi$ visits states in $X_{i}$ infinitely often$\}$ is even).
The sets of runs and accepting runs of $\X$ are denoted by
$\Run_{\X}$ and $\AccRun_{\X}$, respectively.

The function $\myroot\colon\Run_{\X}\to X$
is defined to return the root $X$-label of a run.
For each $X'\subseteq X$, we define $\Run_{\X,X'}$ by
$\{\rho\in\Run_{\X}\mid\myroot(\rho)\in X'\}$; the set
$\AccRun_{\X,X'}$ is similar.
The map $\DelSt\colon \Run_{\X}\to \myTree_{\Sigma}$
takes a run, removes all $X$-labels and returns a $\Sigma$-tree.
%\end{mydefinition}

\begin{mydefinition}[$\Lang(\X)$ for NPTAs]\label{def:acceptedLangNPTA}
Let $\X$ be an NPTA.
Its \emph{accepted language} $\Lang(\X)$ is defined by $\DelSt(\AccRun_{\X,s})$.
\end{mydefinition}

\begin{wrapfigure}[4]{r}[0pt]{0pt}
  \hspace{-0.25em}%
  \raisebox{-0.25cm}[0pt][0cm]{
    \begin{tikzpicture}[baseline=-20pt, level distance=3em,sibling distance=.2em]
     \Tree[.\node{$(\sigma,x)$};
       [.\node{$\rho_{1}$};
         \edge[roof];
         {\phantom{hoge}}
       ]
       \edge[draw=none];
       [.\node{$\cdots$}; ]
       [.\node{$\rho_{|\sigma|}$}; 
         \edge[roof];
         {\phantom{hoge}}
       ]
     ]
    \end{tikzpicture}%
    \begin{minipage}{2em}\begin{equation}\label{eq:DiaX}\end{equation}\end{minipage}%
  }%
\end{wrapfigure}
The following coincidence result for the nondeterministic setting is fairly
straightforward. A key is the fact that accepting runs are
characterized---among all possible runs---using an equational system
that is parallel to the one in
Def.~\ref{def:acceptedLangParityTFSys}. 
\begin{mylemma}
\label{lem:eqSysCharacterizationOfAcceptingRuns}

Let $\X=(
%(X_{1},\ldots, X_{n})
\vec{X}
,\Sigma,\delta,s)$ be an NPTA, and
 $\seq{l^{\sol}_{\i}}{n}$ be the solution of the following equational
 system, whose variables $u_{1},\dotsc,u_{n}$ range over  $\pow(\Run_\X)$. 
\begin{equation}\label{eq:eqSysCharacterizationOfAcceptingRuns:eqSys}
  u_{1}
  \;=_{\eta_{1}}\;
  \Diamond_{\X}(u_{1} \cup\cdots\cup u_{n}) \cap \Run_{\X,X_{1}}
  \;,
  \quad
  \dotsc
  \;,
  \quad
  u_{n}
  \;=_{\eta_{n}}\;
  \Diamond_{\X}(u_{1} \cup\cdots\cup u_{n}) \cap \Run_{\X,X_{n}}
%  \,,
\end{equation}
Here:
$\Diamond_\X:\pow(\Run_\X)\to\pow(\Run_\X)$ is given by
\begin{math}
  \Diamond_{\X} R
  \coloneqq
  \bigl\{
    \bigl((\sigma,x),(\seq{\rho_{\i}}{|\sigma|})  \bigr) \in\Run_\X
    \,\big|\,
    \sigma\in\Sigma,x\in X,
 %   \qnta{i\in[1,|\sigma|]} 
 \rho_{i} \in R
  \bigr\}
\end{math}
(see the figure~(\ref{eq:DiaX}) above);
 $X=X_1+\cdots+X_n$; and $\eta_{i}$ is $\mu$ (for odd $i$) or $\nu$
 (for even $i$).
Then
the $i$-th solution  $l^{\sol}_{i}$
coincides with
 $\AccRun_{\X,{X_i}}$.
\qed
\end{mylemma}

We shall translate the above result to
the characterization of accepted trees
(Lem.~\ref{lem:eqSysCharacterizationOfAcceptingTrees}).
In its proof (that is deferred to the appendix)
Lem.~\ref{lem:eqSysWithFixedPtIso}---on homomorphisms of equational
systems---plays an important role.

\begin{mylemma}
\label{lem:eqSysCharacterizationOfAcceptingTrees}
Let $\X=(
%(\X_{1},\ldots,X_{n})
\vec{X}
,\Sigma,\delta,s)$ be an NPTA, and
let $\seq{l'^{\sol}_{\i}}{n}$ be the solution of the following
 equational system,
%over $\bigl(\pow(\myTree_{\Sigma})\bigr)^{X}$
where $u'_{i}$ ranges over the complete lattice $\bigl(\pow(\myTree_{\Sigma})\bigr)^{X_{i}}$:
\begin{equation}\label{eq:eqSysCharacterizationOfAcceptingTrees:eqSys}
  u'_{1}
  \;=_{\eta_{1}}\;
  \Diamond_{\delta}([u'_{1},\dotsc,u'_{n}]) \upharpoonright {X_{1}}
  \;,
  \quad
  \dotsc
  \;,
  \quad
  u'_{n}
  \;=_{\eta_{n}}\;
  \Diamond_{\delta}([u'_{1},\dotsc,u'_{n}]) \upharpoonright {X_{n}}
  \;.
\end{equation}
Here  $\eta_{i}$ is $\mu$ (for odd $i$) or $\nu$
 (for even $i$); 
$(\place)\upharpoonright X_{i}
\colon
\bigl(\pow(\myTree_{\Sigma})\bigr)^{X}\to 
\bigl(\pow(\myTree_{\Sigma})\bigr)^{X_{i}}
$ denotes  domain restriction; and
the function $\Diamond_{\delta}\colon \bigl(\pow(\myTree_{\Sigma})\bigr)^{X}\to
\bigl(\pow(\myTree_{\Sigma})\bigr)^{X}$ is given by
\begin{displaymath}
  (\Diamond_{\delta}T)(x)
  \;\coloneqq\;
  \bigl\{
    \bigl(\sigma,(\seq{\tau_{\i}}{|\sigma|})\bigr)
    \,\big|\,
    \bigl(\sigma,(\seq{x_{\i}}{|\sigma|})\bigr)\in\delta(x),
%    \qnta{i\in[1,|\sigma|]}
    \tau_{i} \in T(x_{i})
  \bigr\}\enspace.
\end{displaymath}
Then we have a coincidence
\begin{math}
 l'^{\sol}_{i}=  \DelSt'(\AccRun_{\X,X_{i}})
\end{math},
where the function 
$\DelSt' \colon \pow(\Run_{\X}) \to (\pow(\myTree_{\Sigma}))^{X}$ 
is given by
\begin{math}
  \DelSt'(R)(x)
  \coloneqq
  \DelSt(\{\rho\in R\mid \myroot(\rho)=x\})
\end{math}. Recall that $\myroot$ returns a run's root $X$-label.
\qed
\end{mylemma}

\begin{mytheorem}[coincidence, in the nondeterministic setting]\label{thm:coincidenceForNPTA}
Let $\X=((X_1,\ldots,X_n),\Sigma,\delta,s)$ be an NPTA, and
$F_{\Sigma}=\coprod_{\sigma\in\Sigma}(\place)^{|\sigma|}$ be the polynomial functor on $\Sets$
that corresponds to $\Sigma$.
Then $\X$ is identified with a parity $(\pow,F_{\Sigma})$-system;
moreover $\Lang(\X)$ (in the conventional sense of
 Def.~\ref{def:acceptedLangNPTA}) coincides with the coalgebraic trace
 semantics $\trp(\X)$ (Def.~\ref{def:acceptedLangParityTFSys}).
Note here that $\myTree_{\Sigma}$ 
carries a final $F_{\Sigma}$-coalgebra (Lem.~\ref{lem:treeAndFinalCoalg}).
\end{mytheorem}
\begin{proof}
We identify $\X$ with the $(\pow,\FSigma)$-system
$\bigl((\seq{X_{\i}}{n}),\delta\colon X\kto\overline{\FSigma}X,s\colon 1\kto X\bigr)$, and
let $1=\{\bullet\}$.
The equational system $E_{\X}$ in Def.~\ref{def:acceptedLangParityTFSys}
is easily seen to coincide with (\ref{eq:eqSysCharacterizationOfAcceptingRuns:eqSys})
in Lem.~\ref{lem:eqSysCharacterizationOfAcceptingTrees}.
The claim is then shown as follows, exploiting the last coincidence.
\begin{align*}
  \trp(\X)
&
%  \overset{\text{Def.~\ref{def:acceptedLangParityTFSys}}}{=}
=
  [\seq{\trp_{\i}(\X)}{n}]\odot s(\bullet)
\qquad\text{by Def.~\ref{def:acceptedLangParityTFSys}}
  \\
&=
  [\seq{\DelSt'(\AccRun_{\X,X_{\i}})}{n}](s)
  \\
  &=
  \DelSt(\AccRun_{\X,s})
%  \overset{\text{Def.\ref{def:acceptedLangNPTA}}}{=}
  =
  \Lang(\X)
  \qquad\text{by Def.~\ref{def:acceptedLangNPTA}.}
  \qedhere
\end{align*}
\end{proof}

\section{Coincidence with the Conventional Definition: Probabilistic}
\label{sec:probParitySys}
In the probabilistic setting
the coincidence result is much more intricate.  Even the
well-definedness of parity trace semantics
(Def.~\ref{def:acceptedLangParityTFSys}) is nontrivial: the posets
$\Kl(\giry)(X_{i},Z)$ of our interest are not complete lattices, and
they even lack the greatest element $\top$. Therefore neither of
Lem.~\ref{lem:eqSysSolvableCompLat}--\ref{lem:eqSysSolvableOmegaCont}
ensures a solution of $E_{\X}$ in
Def.~\ref{def:acceptedLangParityTFSys}. As we hinted
in~\S{}\ref{subsec:eqsys} our strategy is:
1) to apply the 
Lem.~\ref{lem:eqSysSolvableOmegaCont} to the homset $\Kl(\giry)(X,1)$; and 2) to ``extend''
fixed points in $\Kl(\giry)(X,1)$ along a final $F$-sequence. Implicit 
in  the proof
details below, in fact, is  a correspondence between: abstract categorical
arguments along a final sequence; and concrete operational intuitions
on probabilistic parity automata.

In this section we let $\C=\Meas$, $T=\giry$ (Def.~\ref{def:powersetMonadAndSubGiryMonad}), and $F$ be a polynomial functor.

\begin{myremark}\label{rem:weAreInterestedInGenerativeProbSys}
The class of probabilistic systems of our
interest are \emph{generative} (as opposed to \emph{reactive}) ones.
Their difference is eminent in the types of transition functions:
\begin{displaymath}\textstyle
\begin{array}{llllllll}
  &X\longrightarrow \giry(A\times X)
  \quad\text{(word)}
  &
  & X\longrightarrow\textstyle \giry(\coprod_{\sigma\in\Sigma} X^{|\sigma|})
  \quad\text{(tree)}
  &&
  \quad\text{for generative;} \\ %\qquad
  &X\longrightarrow (\giry X)^{A}
  \quad\text{(word)}
  &&
  X\longrightarrow\textstyle \prod_{\sigma\in\Sigma}\giry( X^{|\sigma|})
  \quad\text{(tree)}
  &&
  \quad\text{for reactive.}
\end{array}
\end{displaymath}
A generative system (probabilistically) chooses which
character to \emph{generate}; while a reactive one \emph{receives} a
character from the environment.
Reactive variants of probabilistic tree automata have
been studied e.g.\ in~\cite{CarayolHS14ria}, following earlier works
 like~\cite{BaierG05} on reactive probabilistic word automata.
Further discussion is in Appendix~\ref{subsec:GenAndReactSys}.
\end{myremark}

\subsection{Trace Semantics of Parity $(\giry,F)$-Systems is Well-Defined}
\label{subsec:fixedPtProbParitySys}
In the following key lemma---that is inspired by the observations
in~\cite{Cirstea10git,Schubert09tcf,UrabeH15cit}---a
typical usage is for
$X_{A}=X_{1}+\cdots+X_{i}$ and $X_{B}=X_{i+1}+\cdots+X_{n}$.

\begin{mylemma}
\label{lem:inverseExists}
Let $\X=((X_1,\ldots,X_n),s,c)$ be a parity $(\giry,F)$-system, and
suppose that we are given a partition
$X=X_{A}+ X_{B}$ of $X:=X_1+\cdots+ X_n$.

We define a function
$\Gamma\colon\Kl(\giry)(X,Z)\to\Kl(\giry)(X,1)$ by $\Gamma(g)=J!_{Z}\odot g$,
where $!\colon Z\to 1$ is the unique function of the type.
Its variants $\Gamma_{A}:\Kl(\giry)(X_{A},Z)\to\Kl(\giry)(X_{A},1)$ and
$\Gamma_{B}:\Kl(\giry)(X_{B},Z)\to\Kl(\giry)(X_{B},1)$ are
defined similarly.

For arbitrary $g_{B}\colon X_{B}\kto Z$, we define 
$\mathfrak{G}^{g_{B}}$ and $\mathfrak{H}^{g_{B}}$ as the following sets
of ``fixed points'':
\begin{equation}\label{eq:mathfrakGAndMathfrakH}
  \vspace{.5em}
  \mathrlap{\raisebox{32pt}[0pt][0pt]{$\hspace{-.5em}\mathfrak{G}^{g_{B}}\coloneqq$}}
  \left\{
    \begin{array}{l}
      g_{A} \colon \\ X_{A} \kto Z
    \end{array}
  \,\middle|
    \vcenter{\xymatrix@R=.6em@C+2em{
      {\overline{F}X}
        \kar[r]^-{\mathstrut\smash{\oF [g_{A},g_{B}]}}
        \ar@{}[rd]|{=}
      &
      {\overline{F}Z}
        \kar[d]^{J\zeta^{-1}}
      \\
      {X_A}
        \kar[u]^{c_{A}}
        \kar[r]_-{\mathstrut\smash{g_{A}}}
      &
      {Z}
    }}
    \hspace{-0.25em}
  \right\}
  \text{ and }
  \mathrlap{\raisebox{32pt}[0pt][0pt]{$\hspace{-.5em}\mathfrak{H}^{g_{B}}\coloneqq$}}
  \left\{
    \begin{array}{l}
      h_{A} \colon \\ X_{A} \kto 1
    \end{array}
  \,\middle|
    \vcenter{\xymatrix@R=.6em@C+2em{
      {\overline{F}X}
        \kar[r]^-{\mathstrut\smash{\oF [h_{A},\Gamma_{B}(g_{B})]}}
        \ar@{}[rd]|{=}
      &
      {\overline{F}1}
        \kar[d]^{J\bang_{F1} }
      \\
      {X_{A}}
        \kar[u]^{c_{A}}
        \kar[r]_-{\mathstrut\smash{h_{A}}}
      &
      {1}
    }}
    \hspace{-0.25em}
  \right\}
\end{equation}
Then $\Gamma_{A}$ restricts to a function
$\mathfrak{G}^{g_{B}} \to \mathfrak{H}^{g_{B}}$.
Moreover, the restriction is an order isomorphism,
with its inverse denoted by
$\Delta^{g_{B}}\colon \mathfrak{H}^{g_{B}} \iso \mathfrak{G}^{g_{B}}$.
\qed
\end{mylemma}
\noindent
In the proof of the last lemma (deferred to the appendix),
the inverse $\Delta^{g_{B}}$ is defined by ``extending''
$h_{A} \colon X_{A} \kto 1$ to $X_{A}\kto Z$, along
the final $F$-sequence $1\leftarrow F 1\leftarrow\cdots$
(more precisely: the image of the sequence under
the Kleisli inclusion $J\colon \Meas\to\Kl(\giry)$).

% Here we obtain $h_{A}$ by extending $g_{A}$ along a
% \emph{final sequence}~\cite{Worrell05otf}. For example:
% the maximal trace---i.e.\ the accepted language with the trivial
% acceptance condition---on the right corresponds to the
% \emph{no-deadend} probability (Def.~\ref{def:noDeadendProb}) on the
% left; and 
%  the acceptance
% language under a parity condition, on the right, corresponds to the
% probability of generating a successful run, on the left.

% In this lemma, we also parameterize the above observation by
% $g_{B}$---that corresponds to unsolved parameters on
% a equational system---to cope with nested fixed points.
% This parameter indeed helps us,
% thanks to our another result of Lem.~\ref{lem:eqSysWithFixedPtIso}.
% The proofs of those lemmas are deferred to Appendix~\ref{sec:omittedProofs}
% for space reasons; nevertheless we take those as one of
% our main technical contributions.

We are ready to prove existence of $E_{\X}$'s solution (Def.~\ref{def:acceptedLangParityTFSys}).
\begin{mylemma}
\label{lem:eqSysCoincidenceTraceAndAccProb}
Assume the same setting as in Lem.~\ref{lem:inverseExists}.
We define $\Phi_{\X}\colon \Kl(\giry)(X,Z)\kto\Kl(\giry)(X,Z)$ and
$\Psi_{\X}\colon \Kl(\giry)(X,1)\kto\Kl(\giry)(X,1)$, respectively, by
\begin{equation*}
  \Phi_{\X}(g)
  \;\coloneqq\;
  J\zeta^{-1} \kco \overline{F}g \kco  c
  \quad\text{and}\quad
  \Psi_{\X}(h)
  \;\coloneqq\;
  J\bang_{F1} \kco \overline{F}h \kco c
  \;;
\end{equation*}
these are like the diagrams in~(\ref{eq:mathfrakGAndMathfrakH}), 
except that the latter are parametrized by $X_A, X_B, g_{B}$.
Now consider the following equational systems, where:
 $\eta_{i}=\mu$ if $i$ is odd and $\eta_{i}=\nu$ if $i$ is even; 
$u_{i}$ ranges over $\Kl(\giry)(X_{i},Z)$; and 
$u'_{i}$ ranges over $\Kl(\giry)(X_{i},1)$.
  \begin{equation}\label{eq:eqSysCoincidenceTraceAndAccProb:eqSysTrace}
    E
    =
    \left[
      \begin{array}{c}
        u_{1}
        =_{\eta_{1}}
        \Phi_{\X}([u_{1},\dotsc,u_{n}]) \kco \kappa_{1}
        \\
        \vdots
        \\
        u_{n}
        =_{\eta_{n}}
        \Phi_{\X}([u_{1},\dotsc,u_{n}]) \kco \kappa_{n}
      \end{array}
    \right]
    \quad
    E'
    =
    \left[
      \begin{array}{c}
        u'_{1}
        =_{\eta_{1}}
        \Psi_{\X}([u'_{1},\dotsc,u'_{n}]) \kco \kappa_{1}
        \\
        \vdots
        \\
        u'_{n}
        =_{\eta_{n}}
        \Psi_{\X}([u'_{1},\dotsc,u'_{n}]) \kco \kappa_{n}
      \end{array}
    \right]
  \end{equation}
We claim that the equational systems have solutions
$(\seq{l^{\sol}_{\i}}{n})$ and $(\seq{l'^{\sol}_{\i}}{n})$;
and moreover, we have
$\Gamma(\trp(\X))=\Gamma([\seq{l^{\sol}_{\i}}{n}])=[\seq{l'^{\sol}_{\i}}{n}]$.
\qed
\end{mylemma}

\begin{mytheorem}\label{thm:constParityTraceSituGiry}
$T=\giry$ and a polynomial  $F$  constitute a parity trace situation  (Def.~\ref{def:acceptedLangParityTFSys}).
\qed
\end{mytheorem}

\begin{myremark}
The process-theoretic interpretation of the isomorphism
$\mathfrak{G}^{g_{B}} \cong \mathfrak{H}^{g_{B}}$ is interesting.
Let us set $X_{A}=X$ and $X_{B}=\emptyset$ for simplicity.
The greatest element on the left is the \emph{infinitary trace semantics}
(i.e.\ accepted languages under the trivial acceptance condition), as in
Thm.~\ref{thm:infinitaryTrace} (cf.\ Table~\ref{table:traceAndSim}).
The corresponding greatest element on the right---a function
$h_{A}\colon X_{A}\to \giry 1\cong [0,1]$---assigns to each state $x\in X$
the probability with which a run from $x$ \emph{does not diverge}
(recall from Rem.~\ref{rem:deadendProb} that the \emph{sub-}Giry monad $\giry$
allows divergence probabilities).
The accepted language under the  parity condition is
in general an element of $\mathfrak{G}^{g_{B}}$ that is neither greatest nor least;
the corresponding element in $\mathfrak{H}^{g_{B}}$ assigns to each state
the probability with which it generates a \emph{accepting} run (over any $\Sigma$-tree).
\end{myremark}

\subsection{Probabilistic Parity Tree Automata and Its Languages}
\label{subsec:PPTAConventionally}
%As usual we first need to prove that the constructs and properties of
%our interest are all measurable. The following definition is standard
%(see e.g.~\cite{CarayolHS14ria}), and is precisely stated in
%Appendix~\ref{sec:GenProbBuechiTreeAutom}.

\begin{mydefinition}[PPTA]\label{def:generativeProbParityTreeAutom}
A \emph{(generative) probabilistic parity tree automaton (PPTA)} is  
\begin{displaymath}
 \X
 \;=\;
 \bigl(\,
   (X_{1},\ldots, X_{n}),\,
   \Sigma,\,
   \delta\colon X\to \giry\bigl(\textstyle\coprod_{\sigma\in\Sigma} X^{|\sigma|}\bigr),\,
   s\in \giry X
\,\bigr)
\enspace,
\end{displaymath}
where 
$X=X_1+\cdots+X_n$,
each $X_i$ is a countable set and $\Sigma$ is a countable ranked alphabet.
The subdistribution $s$ over $X$ is for the choice of \emph{initial states}.
\end{mydefinition}
% Recall that use of the \emph{sub-}Giry monad $\giry$  allows
% divergence with certain probabilities. 
% Note the use of the sub-Giry monad $\giry$ in
% Def.~\ref{def:generativeProbParityTreeAutom}: transitions, as well as
% initial states, are given by \emph{sub}-distributions in general. In
% case the whole space gets assigned a probability strictly less than $1$,
% the missing probability is understood as the one with which the
% automaton gets into \emph{deadend}.
In Def.~\ref{def:generativeProbParityTreeAutom} the size restrictions on
$X$ and $\Sigma$ are not essential: restricting to discrete
$\sigma$-algebras, however, makes the following arguments much simpler.

We shall concretely define accepted languages of PPTAs,
continuing~\S{}\ref{sec:nondetParitySys} and deferring precise definitions to
Appendix~\ref{sec:detailsOfTreeRunAcc}.
This is mostly standard; a reactive variant is found in~\cite{CarayolHS14ria}.

\begin{mydefinition}[$\myTree_{\Sigma}$ and $\Run_{\X}$]\label{def:measurableStrOfTreeAndRunBriefly}
% We shall go ahead to introduce the measurable structure (i.e.\ a
% $\sigma$-algebra) to the set $\myTree_{\Sigma}$ of $\Sigma$-trees, and
% to $\Run_\X $ of runs. The definitions here are
% straightforward generalization of those in~\cite[\S{}4]{CarayolHS14ria}. 
Let $\Sigma$ be a ranked alphabet; $\myTree_{\Sigma}$ is the set of
 $\Sigma$-trees. 
A finite $(\Sigma \cup \{\ast\})$-labeled tree $\lambda$, 
with its  branching degrees compatible with the label arities, 
is called a \emph{partial $\Sigma$-tree}. 
Here
the new symbol $\ast$ (``continuation'') is deemed to be $0$-ary.
The \emph{cylinder set} 
%in $\myTree_{\Sigma}$ 
associated to $\lambda$, denoted by
$\Cyl_{\Sigma}(\lambda)$, is the set of (non-partial) $\Sigma$-trees
that have $\lambda$ as their prefix (in the sense that a subtree is
replaced by $\ast$).
The (smallest) $\sigma$-algebra on $\myTree_{\Sigma}$ generated by the family
$\{\Cyl_{\Sigma}(\lambda)\mid \text{$\lambda$ is a partial $\Sigma$-tree}\}$
will be denoted by $\mathfrak{F}_{\Sigma}$.

A \emph{run} of a PPTA $\X$ with state space $X$ is a (possibly infinite)
 $(\Sigma\times X)$-labeled tree whose branching degrees are compatible
 with the arities of $\Sigma$-labels. $\Run_\X$ denotes the set of runs.
The measurable structure $\mathfrak{F}_{\X}$ on
$\Run_{\X}$ is defined analogously to $\mathfrak{F}_{\Sigma}$: a \emph{partial run} 
$\xi$ of $\X$ is a suitable  $(\Sigma \cup \{\ast\}) \times X$-labeled
tree; it generates a \emph{cylinder set} $\Cyl_{\X}(\xi)\subseteq\Run_{\X}$;
and these cylinder sets generate the $\sigma$-algebra
$\mathfrak{F}_{\X}$. 
% its leaves are labeled by $\ast$ or $\sigma$ such that
% $|\sigma| = 0$, and it is the ``prefix'' of
% a $\Sigma$-tree $\tau$ (expect $\ast$-labeled nodes).
%
% A \emph{partial run} $\xi$, that is
% a $(\Sigma \cup \{\ast\}) \times X$-labeled tree,
% is defined similary.
% For a partial run $\xi$ of $\X$, the \emph{cylinder set}
% $\Cyl_{\X}(\xi)\subseteq \Run_\X $ associated to $\xi$, and
% the $\sigma$-algebra over $\Run_\X$ $\mathfrak{F}_{\X}$
% generated by $\Cyl_{\X}(\xi)$, are defined similarly.
Finally, the set $\AccRun_{\X}$ of \emph{accepting runs} consists of all
 those runs all branches of which satisfy the (usual) \emph{parity
 acceptance condition}
(namely: $\max\{i\mid\pi$ visits states in $X_{i}$ infinitely often$\}$ is
 even). 
\end{mydefinition}

% We are interested in the probability with which an automaton
% $\X$ generates an \emph{accepting} run. Therefore we need the
% following result; it is much like~\cite[Lem.~36]{CarayolHS14ria} and
% hardly novel.

The following result is much like~\cite[Lem.~36]{CarayolHS14ria} and hardly novel.
\begin{mylemma}\label{lem:acceptingRunsAreMeasurable}
The set $\AccRun_{\X}$ of accepting runs is
an $\mathfrak{F}_{\X}$-measurable subset of $\Run_{\X}$.
\qed
\end{mylemma}

In the following   $\NDL_{\X}(x)$ is the probability with which an
execution from $x$ does not diverge:
since we use the \emph{sub}-Giry monad 
(Def.~\ref{def:generativeProbParityTreeAutom}),
a PPTA can exhibit divergence.
%This is used to define  probabilistic accepted languages.
% Finally, we are ready to define the subprobability measures on runs 
%  and $\Sigma$-trees,
% induced by a probabilistic B\"uchi tree automaton.
\begin{mydefinition}
%[$\Lang$ for PPTAs]
[$\mu_{\X}^{\Run}$ over $\Run^\giry_{\X}$]
\label{def:NoDivergence}
%\label{def:LangForPPTAs}
Let $\X=((X_1,\ldots,X_n),\Sigma,\delta,s)$ be a PPTA.
%, and $X=\coprod_{i}X_{i}$.

 Firstly, for each $k\in \nat$, let
 $\NDL_{\X,k}\colon X\to [0,1]$ (``no divergence in $k$ steps'') be
 defined inductively by:
 \begin{math}
     \NDL_{\X,0}(x)
   \coloneqq
    1
 \end{math} and
 \begin{equation*}%\label{eq:10151159}
  \begin{aligned}
    % &\NDL_{\X,0}(x)
    % &&\hspace{-1em}\coloneqq
    % 1\enspace,
    % \\
    &\NDL_{\X,k+1}(x)
    &&\hspace{-1em}\coloneqq{}
    \hspace{-1em}
    \textstyle
    \sum\limits_{
      (\sigma,(\seq{x_{\i}}{|\sigma|}))
      \in \coprod_{\sigma\in \Sigma} X^{|\sigma|}
    }
    \hspace{-1em}
    \delta(x)\bigl(\sigma,(\seq{x_{\i}}{|\sigma|})\bigr)\cdot
    \prod_{i\in[1,|\sigma|]}
    \NDL_{\X,k} (x_{i})
    \,.
  \end{aligned}
 \end{equation*}
 We define $\NDL_{\X}(x)\coloneqq \bigwedge_{k\in\nat}\NDL_{\X,k}(x)$.

Secondly we  define a subprobability measure
 $\mu_{\X}^{\Run}$ over  $\Run_{\X}$. It is given by
 % We define a subprobability measure $\mu_{\X}^{\Run_{\X}}$ on the measurable set
 %   $(\Run_{\X},\mathfrak{F}_{\X})$ by
 \begin{equation}\label{eq:01171719}
 \small
\begin{aligned}
   &\mu_{\X}^{\Run}(\Cyl_{\X}(\xi))
  \coloneqq
  s\bigl(\myroot(\xi)\bigr)\cdot P_{\X}(\xi)  
 \quad\text{for each partial run $\xi$, where $P_{\X}(\xi)$ is given by}
 \\
  &P_{\X}(\xi)
  \coloneqq
  \begin{cases}
    \NDL_{\X}(x)
    &\text{if $\xi=\bigl((\ast,x)\bigr)$;}
    \\
    \delta(x)\bigl(
      \sigma,\bigl(\seq{\myroot(\xi_{\i})}{|\sigma|}\bigr)
    \bigr)
    \cdot
    \textstyle\prod_{i\in [1,|\sigma|]}
    P_{\X}(\xi_{i})
    &\text{if $\xi=\bigl((\sigma,x),(\seq{\xi_{\i}}{|\sigma|})\bigr)$.}
  \end{cases}
\end{aligned} 
\end{equation}
 %here $s\in \giry X$ is $\X$'s distribution of initial states; and
\auxproof{ here $P_{\X}(\xi)$---the
 probability with which $\X$ generates a run that has $\xi$ as a prefix---is
 defined in the following bottom-up manner.
 \begin{equation*}%\label{eq:10142322}
  P_{\X}(\xi)
  \coloneqq
  \begin{cases}
    \NDL_{\X}(x)
    &\text{if $\xi=\bigl((\ast,x)\bigr)$;}
    \\
    \delta(x)\bigl(
      \sigma,\bigl(\seq{\myroot(\xi_{\i})}{|\sigma|}\bigr)
    \bigr)
    \cdot
    \textstyle\prod_{i\in [1,|\sigma|]}
    P_{\X}(\xi_{i})
    &\text{if $\xi=\bigl((\sigma,x),(\seq{\xi_{\i}}{|\sigma|})\bigr)$}
  \end{cases}
 \end{equation*}
}
The above extends to a  measure thanks to Carath\'{e}odory's theorem.
See Lem.~\ref{lem:muXOnRunXExt}.

 Thirdly we introduce a measure $\mu^{\myTree}_{\X}$ over
 $\myTree_{\Sigma}$ (``which trees are generated by what
probabilities''). It is a \emph{push-forward
measure} of $\mu_{\X}^{\Run}$ along $\DelSt\colon \Run_{\X}\to
 \myTree_{\Sigma}$:
% $\mu_{\X}^{\myTree_{\Sigma}}(\Cyl(\lambda))$ by 
\begin{equation}\label{eq:01171734}
  \mu_{\X}^{\myTree}(\Cyl_{\Sigma}(\lambda))
  \;\coloneqq\;
  \mu_{\X}^{\Run}
  \left(\,
    \DelSt^{-1}(\Cyl_{\Sigma}(\lambda))
    \cap
    \AccRun_{\X}
  \,\right)
  \quad\text{for each  partial $\Sigma$-tree $\lambda$.}
\end{equation}

Since $X$ is countable  $\DelSt$ is easily seen to be measurable.
\auxproof{
$\DelSt^{-1}(\Cyl(\lambda))$ is indeed measurable in $\Run_{\X}$ because
\begin{displaymath}
  \DelSt^{-1}(\Cyl(\lambda))
  \;=\;
  \textstyle\bigcup_{\xi\in\DelSt^{-1}(\{\lambda\})}\Cyl_{\X}(\xi)
  \,.
\end{displaymath}
where $\DelSt$ on the right-hand side is the straightforward adaptation
of $\DelSt\colon \Run_{\X}\to \myTree_{\Sigma}$ to partial runs/trees.
Since $\lambda$ is a finite tree and the state space $X$ is
countable, the above union is a countable one. 
}

Finally, the \emph{accepted language} $\Lang(\X)\in \giry(\myTree_\Sigma)$
of $\X$ is defined by $\mu_{\X}^{\myTree}$ in the above.
\end{mydefinition}

\subsection{Coincidence between Conventional and Coalgebraic Languages}
\begin{mylemma}\label{lem:fixedPtCharacterizationOfAcceptanceProb}
Let $\X=((X_{1},\cdots, X_{n}),\Sigma,\delta,s)$ be a PPTA with
  $X=\coprod_{i}X_{i}$, and
 $\Psi'_{\X}$ be
%$\colon [0,1]^{X} \to [0,1]^{X}$ 
\begin{equation*}
  \textstyle
  \Psi'_{\X}\colon [0,1]^{X} \to [0,1]^{X},\;
  \Psi'_{\X}(p)(x)
  \;\coloneqq\;
  \sum_{
    (\sigma,\seq{x_{\i}}{|\sigma|})
    \in \coprod_{\sigma} X^{|\sigma|}
  }
  \delta(x)
  \bigl(\sigma,(\seq{x_{\i}}{|\sigma|})\bigr)
  \cdot
  \prod_{i\in[1,|\sigma|]}
  p(x_{i})
  \,.
\end{equation*}
Let us define 
$\mu^{\myTree}_{\X,x}:=\mu^{\myTree}_{\X(x)}$
where $\X(x)$ is the PPTA obtained from $\X$ by changing
its initial  distribution $s$ into the Dirac distribution $\delta_{x}$;
$\mu^{\Run}_{\X,x}$ is similar. 
We define $\AccProb_{\X}\colon X\to [0,1]$---it assigns to each state the probability of
generating an accepting run---by
\begin{math}
  \AccProb_{\X}(x)
  \coloneqq
  \mu^{\Run}_{\X,x}(\AccRun_{\X})
\end{math}.

Consider the following equational system,
where $u'_{i}$ ranges over $\Kl(\giry)(X_{i},1)$, and
 $(\place)\upharpoonright X_{i}$ denotes domain restriction.
\begin{equation*}%\label{eq:fixedPtCharacterizationOfAcceptanceProb:eqSys}
  u'_{1}
  \;=_{\eta_{1}}\;
  \Psi'_{\X}([u'_{1},\cdots,u'_{n}]) \upharpoonright X_{1}
  ,
  \quad
  \dotsc
  ,
  \quad
  u_{n}
  \;=_{\eta_{n}}\;
  \Psi'_{\X}([u'_{1},\cdots,u'_{n}]) \upharpoonright X_{n}
\end{equation*}
We claim:
1) the system has a solution $\seq{l'^{\sol}_{\i}}{n}$; and
2)
\begin{math}
  [\seq{l'^{\sol}_{\i}}{n}]
  =
  \AccProb_{\X}
\end{math}.
\qed
\end{mylemma}
Its proof (in the appendix) relies on Lem.~\ref{lem:eqSysWithOmegaChainMap}
on homomorphisms of equational systems.

\begin{mytheorem}[coincidence, in the probabilistic setting]\label{thm:coincidenceForPPTA}
Let $\X=((X_1,\ldots,X_n),\Sigma,\delta,s)$ be a PPTA, and % where
$X=X_1+\cdots+X_n$, and $F_{\Sigma}$ be the polynomial functor on
$\Meas$ that corresponds to $\Sigma$.  Then $\X$ is identified with a
parity $(\giry,F_{\Sigma})$-system; moreover its coalgebraic trace
semantics $\trp(\X)$ (Def.~\ref{def:acceptedLangParityTFSys}) coincides
with the (probabilistic) language $\Lang(\mathcal{X})$
concretely defined in Def.~\ref{def:NoDivergence}.
Precisely: $\trp(\X)(\bullet)(U) = \Lang(\mathcal{X})(U)$
for any measurable subset
$U$ of $\myTree_{\Sigma}$, where $\bullet$ is the unique element of $1$ in
$\trp(\X)\colon 1\to \giry (\myTree_{\Sigma})$.
\qed
\end{mytheorem}

\paragraph*{Acknowledgments}
Thanks are due to
Corina C\^{\i}rstea,
Kenta Cho,
Bartek Klin,
Tetsuri Moriya and
Shota Nakagawa
for useful discussions; and to the anonymous referees for their comments.
The authors are supported by
Grants-in-Aid No. 24680001 \& 15KT0012, JSPS; N.U.\ is supported by
Grant-in-Aid for JSPS Fellows.

\bibliography{02publabbr,02jrsrabbr,00myref,02procabbr}%

\begin{thebibliography}{10}

\bibitem{AdamekBHKMS12acp}
Jir\'{\i} Ad\'{a}mek, Filippo Bonchi, Mathias H\"{u}lsbusch, Barbara K\"{o}nig,
  Stefan Milius, and Alexandra Silva.
\newblock A coalgebraic perspective on minimization and determinization.
\newblock In {\em Proc. {FoSSaCS}'12}, volume 7213 of {\em LNCS}, pages 58--73.
  Springer, 2012.
\newblock \href {http://dx.doi.org/10.1007/978-3-642-28729-9_4}
  {\path{doi:10.1007/978-3-642-28729-9_4}}.

\bibitem{AdamekK79lfp}
Ji\v{r}\'{\i} Ad\'{a}mek and V\'{a}clav Koubek.
\newblock Least fixed point of a functor.
\newblock {\em J. Comp. \& Syst. Sci.}, 19(2):163--178, 1979.
\newblock \href {http://dx.doi.org/10.1016/0022-0000(79)90026-6}
  {\path{doi:10.1016/0022-0000(79)90026-6}}.

\bibitem{ArnoldN01rom}
Andr\'{e} Arnold and Damian Niwi\'{n}ski.
\newblock {\em Rudiments of {$\mu$}-Calculus}, volume 146 of {\em Studies in
  Logic and the Foundations of Mathematics}.
\newblock North-Holland, 2001.
\newblock \href {http://dx.doi.org/10.1016/S0049-237X(01)80001-X}
  {\path{doi:10.1016/S0049-237X(01)80001-X}}.

\bibitem{BaierG05}
Christel Baier and Marcus Gr{\"{o}}{\ss}er.
\newblock Recognizing omega-regular languages with probabilistic automata.
\newblock In {\em 20th {IEEE} Symposium on Logic in Computer Science {(LICS}
  2005), 26-29 June 2005, Chicago, IL, USA, Proceedings}, pages 137--146.
  {IEEE} Computer Society, 2005.
\newblock URL: \url{http://dx.doi.org/10.1109/LICS.2005.41}, \href
  {http://dx.doi.org/10.1109/LICS.2005.41} {\path{doi:10.1109/LICS.2005.41}}.

\bibitem{BrengosMP15bef}
Tomasz Brengos, Marino Miculan, and Marco Peressotti.
\newblock Behavioural equivalences for coalgebras with unobservable moves.
\newblock {\em J. Logical \& Algebraic Methods in Prog.}, 84(6):826--852, 2015.
\newblock \href {http://dx.doi.org/10.1016/j.jlamp.2015.09.002}
  {\path{doi:10.1016/j.jlamp.2015.09.002}}.

\bibitem{CarayolHS14ria}
Arnaud Carayol, Axel Haddad, and Olivier Serre.
\newblock Randomization in automata on infinite trees.
\newblock {\em {ACM} Trans. Comp. Logic}, 15(3):24:1--24:33, 2014.
\newblock \href {http://dx.doi.org/10.1145/2629336}
  {\path{doi:10.1145/2629336}}.

\bibitem{CianciaV12saa}
Vincenzo Ciancia and Yde Venema.
\newblock Stream automata are coalgebras.
\newblock In {\em Selected Papers of {CMCS}'12}, volume 7399 of {\em LNCS},
  pages 90--108. Springer, 2012.
\newblock \href {http://dx.doi.org/10.1007/978-3-642-32784-1_6}
  {\path{doi:10.1007/978-3-642-32784-1_6}}.

\bibitem{Cirstea10git}
Corina C\^{\i}rstea.
\newblock Generic infinite traces and path-based coalgebraic temporal logics.
\newblock {\em Electr. Notes in Theor. Comp. Sci.}, 264(2):83--103, 2010.
\newblock \href {http://dx.doi.org/10.1016/j.entcs.2010.07.015}
  {\path{doi:10.1016/j.entcs.2010.07.015}}.

\bibitem{CirsteaKPSV11mla}
Corina C\^{\i}rstea, Alexander Kurz, Dirk Pattinson, Lutz Schr\"{o}der, and Yde
  Venema.
\newblock Modal logics are coalgebraic.
\newblock {\em Comp. Journ.}, 54(1):31--41, 2011.
\newblock \href {http://dx.doi.org/10.1093/comjnl/bxp004}
  {\path{doi:10.1093/comjnl/bxp004}}.

\bibitem{CleavelandKS92fmc}
Rance Cleaveland, Marion Klein, and Bernhard Steffen.
\newblock Faster model checking for the modal mu-calculus.
\newblock In {\em Proc. {CAV}'92}, volume 663 of {\em LNCS}, pages 410--422.
  Springer, 1992.
\newblock \href {http://dx.doi.org/10.1007/3-540-56496-9_32}
  {\path{doi:10.1007/3-540-56496-9_32}}.

\bibitem{Doberkat09scl}
Ernst-Erich Doberkat.
\newblock {\em Stochastic Coalgebraic Logic}.
\newblock Monographs in Theoretical Computer Science. An EATCS Series.
  Springer, 2009.
\newblock \href {http://dx.doi.org/10.1007/978-3-642-02995-0}
  {\path{doi:10.1007/978-3-642-02995-0}}.

\bibitem{Giry82cao}
Mich\`{e}le Giry.
\newblock Categorical aspects of topology and analysis.
\newblock In {\em A categorical approach to probability theory, an Intl.
  Conference at Carleton University, 1981, Proceedings}, volume 915 of {\em
  Lect. Notes in Math.}, pages 68--85. Springer, 1982.
\newblock \href {http://dx.doi.org/10.1007/BFb0092872}
  {\path{doi:10.1007/BFb0092872}}.

\bibitem{GoncharovP14cwb}
Sergey Goncharov and Dirk Pattinson.
\newblock Coalgebraic weak bisimulation from recursive equations over monads.
\newblock In {\em Proc. {ICALP}'14, Part {II}}, volume 8573 of {\em LNCS},
  pages 196--207. Springer, 2014.
\newblock \href {http://dx.doi.org/10.1007/978-3-662-43951-7_17}
  {\path{doi:10.1007/978-3-662-43951-7_17}}.

\bibitem{GradelTW02ala}
Erich Gr\"{a}del, Wolfgang Thomas, and Thomas Wilke, editors.
\newblock {\em Automata, Logics, and Infinite Games: {A} Guide to Current
  Research}, volume 2500 of {\em LNCS}.
\newblock Springer, 2002.
\newblock \href {http://dx.doi.org/10.1007/3-540-36387-4}
  {\path{doi:10.1007/3-540-36387-4}}.

\bibitem{HasuoJS07gts}
Ichiro Hasuo, Bart Jacobs, and Ana Sokolova.
\newblock Generic trace semantics via coinduction.
\newblock {\em Logical Methods in Comp. Sci.}, 3(4):11:1--11:36, 2007.
\newblock \href {http://dx.doi.org/10.2168/LMCS-3(4:11)2007}
  {\path{doi:10.2168/LMCS-3(4:11)2007}}.

\bibitem{HasuoSC16lpm}
Ichiro Hasuo, Shunsuke Shimizu, and Corina C\^{\i}rstea.
\newblock Lattice-theoretic progress measures and coalgebraic model checking.
\newblock In {\em Proc. {POPL}'16}, pages 718--732. {ACM}, 2016.
\newblock \href {http://dx.doi.org/10.1145/2837614.2837673}
  {\path{doi:10.1145/2837614.2837673}}.

\bibitem{Jacobs04tsf}
Bart Jacobs.
\newblock Trace semantics for coalgebras.
\newblock {\em Electr. Notes in Theor. Comp. Sci.}, 106:167--184, 2004.
\newblock \href {http://dx.doi.org/10.1016/j.entcs.2004.02.031}
  {\path{doi:10.1016/j.entcs.2004.02.031}}.

\bibitem{Jacobs12itc}
Bart Jacobs.
\newblock Introduction to coalgebra. {Towards} mathematics of states and
  observations.
\newblock Draft of a book (ver. 2.0), available online, 2012.
\newblock URL: \url{http://www.cs.ru.nl/B.Jacobs/CLG/JacobsCoalgebraIntro.pdf}.

\bibitem{JacobsSS15tsv}
Bart Jacobs, Alexandra Silva, and Ana Sokolova.
\newblock Trace semantics via determinization.
\newblock {\em J. Comp. \& Syst. Sci.}, 81(5):859--879, 2015.
\newblock \href {http://dx.doi.org/10.1016/j.jcss.2014.12.005}
  {\path{doi:10.1016/j.jcss.2014.12.005}}.

\bibitem{Klin09bma}
Bartek Klin.
\newblock Bialgebraic methods and modal logic in structural operational
  semantics.
\newblock {\em Inf. \& Comp.}, 207(2):237--257, 2009.
\newblock \href {http://dx.doi.org/10.1016/j.ic.2007.10.006}
  {\path{doi:10.1016/j.ic.2007.10.006}}.

\bibitem{LynchV95fab}
Nancy Lynch and Frits Vaandrager.
\newblock Forward and backward simulations.
\newblock {\em Inf. \& Comp.}, 121(2):214--233, 1995.
\newblock \href {http://dx.doi.org/10.1006/inco.1995.1134}
  {\path{doi:10.1006/inco.1995.1134}}.

\bibitem{Mulry94ltf}
Philip~S. Mulry.
\newblock Lifting theorems for {Kleisli} categories.
\newblock In {\em Proc. {MFPS}'93}, volume 802 of {\em LNCS}, pages 304--319.
  Springer, 1994.
\newblock \href {http://dx.doi.org/10.1007/3-540-58027-1_15}
  {\path{doi:10.1007/3-540-58027-1_15}}.

\bibitem{Panangaden09lmp}
Prakash Panangaden.
\newblock {\em Labelled Markov Processes}.
\newblock Imperial College Press, 2009.

\bibitem{PowerT97ecs}
John Power and Hayo Thielecke.
\newblock Environments, continuation semantics and indexed categories.
\newblock In {\em Proc. {TACS}'97}, volume 1281 of {\em LNCS}, pages 391--414.
  Springer, 1997.
\newblock \href {http://dx.doi.org/10.1007/BFb0014560}
  {\path{doi:10.1007/BFb0014560}}.

\bibitem{Rutten00uca}
Jan J. M.~M. Rutten.
\newblock Universal coalgebra: a theory of systems.
\newblock {\em Theor. Comp. Sci.}, 249(1):3--80, 2000.
\newblock \href {http://dx.doi.org/10.1016/S0304-3975(00)00056-6}
  {\path{doi:10.1016/S0304-3975(00)00056-6}}.

\bibitem{Schubert09tcf}
Christoph Schubert.
\newblock Terminal coalgebras for measure-polynomial functors.
\newblock In {\em Proc. {TAMC}'09}, volume 5532 of {\em LNCS}, pages 325--334.
  Springer, 2009.
\newblock \href {http://dx.doi.org/10.1007/978-3-642-02017-9_35}
  {\path{doi:10.1007/978-3-642-02017-9_35}}.

\bibitem{Silva15asi}
Alexandra Silva.
\newblock A short introduction to the coalgebraic method.
\newblock {\em ACM SIGLOG News}, 2(2):16--27, April 2015.
\newblock \href {http://dx.doi.org/10.1145/2766189.2766193}
  {\path{doi:10.1145/2766189.2766193}}.

\bibitem{Sokolova05cao}
Ana Sokolova.
\newblock {\em Coalgebraic Analysis of Probabilistic Systems}.
\newblock PhD thesis, Technische Universiteit Eindhoven, 2005.

\bibitem{UrabeH15cit}
Natsuki Urabe and Ichiro Hasuo.
\newblock Coalgebraic infinite traces and kleisli simulations.
\newblock In Lawrence~S. Moss and Pawel Sobocinski, editors, {\em Proc.
  {CALCO}'15}, volume~35 of {\em LIPIcs}, pages 320--335. Schloss Dagstuhl,
  2015.
\newblock \href {http://dx.doi.org/10.4230/LIPIcs.CALCO.2015.320}
  {\path{doi:10.4230/LIPIcs.CALCO.2015.320}}.

\bibitem{UrabeH16qsb}
Natsuki Urabe and Ichiro Hasuo.
\newblock Quantitative simulations by matrices.
\newblock {\em Inf. \& Comp.}, 2016.
\newblock In press.
\newblock \href {http://dx.doi.org/10.1016/j.ic.2016.03.007}
  {\path{doi:10.1016/j.ic.2016.03.007}}.

\bibitem{UrabeSH16buechiSimulationArXiv}
Natsuki Urabe, Shunsuke Shimizu, and Ichiro Hasuo.
\newblock Fair simulation for nondeterministic and probabilistic {B\"uchi}
  automata: a coalgebraic perspective.
\newblock {\em CoRR}, abs/1606.04680, 2016.
\newblock URL: \url{http://arxiv.org/abs/1606.04680}.

\bibitem{vanGlabbeek01tlt}
R.J. van Glabbeek.
\newblock The linear time -- branching time spectrum {I}: The semantics of
  concrete, sequential processes.
\newblock In J.A. BergstraA.~PonseS.A. Smolka, editor, {\em Handbook of Process
  Algebra}, chapter~1, pages 3--99. Elsevier, 2001.
\newblock \href {http://dx.doi.org/10.1016/B978-044482830-9/50019-9}
  {\path{doi:10.1016/B978-044482830-9/50019-9}}.

\bibitem{vanGlabbeekSST90rga}
Rob~J. van Glabbeek, Scott~A. Smolka, Bernhard Steffen, and Chris M.~N. Tofts.
\newblock Reactive, generative, and stratified models of probabilistic
  processes.
\newblock In {\em Proc. {LICS}'90}, pages 130--141. {IEEE} Comput. Soc., 1990.
\newblock \href {http://dx.doi.org/10.1109/LICS.1990.113740}
  {\path{doi:10.1109/LICS.1990.113740}}.

\bibitem{Vardi95aaa}
Moshe~Y. Vardi.
\newblock An automata-theoretic approach to linear temporal logic.
\newblock In {\em Logics for Concurrency, the of 8th Banff Higher Order
  Workshop, 1995, Proceedings}, volume 1043 of {\em LNCS}, pages 238--266.
  Springer, 1995.
\newblock \href {http://dx.doi.org/10.1007/3-540-60915-6_6}
  {\path{doi:10.1007/3-540-60915-6_6}}.

\bibitem{Venema06aaf}
Yde Venema.
\newblock Automata and fixed point logic: {A} coalgebraic perspective.
\newblock {\em Inf. \& Comp.}, 204(4):637--678, 2006.
\newblock \href {http://dx.doi.org/10.1016/j.ic.2005.06.003}
  {\path{doi:10.1016/j.ic.2005.06.003}}.

\bibitem{Worrell05otf}
James Worrell.
\newblock On the final sequence of a finitary set functor.
\newblock {\em Theor. Comp. Sci.}, 338(1-3):184--199, 2005.
\newblock \href {http://dx.doi.org/10.1016/j.tcs.2004.12.009}
  {\path{doi:10.1016/j.tcs.2004.12.009}}.

\end{thebibliography}

%\end{document}

\ifoutputappendix
\newpage
\appendix

\section{Tree, Run, and Accepting Run}
\label{sec:detailsOfTreeRunAcc}
Here are some supplementary definitions on (conventional notions) of
nondeterministic/probabilistic tree
automata. See first~\S{}\ref{sec:nondetParitySys}
and~\S\ref{subsec:PPTAConventionally}.

\begin{myremark}\label{rem:generativeProbBuechiTreeAutom}
We let $\nat^{*}$ and $\nat^{\omega}$ denote the sets of finite and
infinite sequences over natural numbers, respectively.  We let
$\nat^{\infty}\coloneqq\nat^{*}\cup\nat^{\omega}$.
Concatenation of finite/infinite sequences, and/or characters are denoted
simply by juxtaposition. Given an infinite sequence
$\pi=\pi_{1}\pi_{2}\dotsc\in\nat^{\omega}$ (here $\pi_{i}\in\nat$), its
prefix $\pi_{1}\dotsc\pi_{n}$ is denoted by $\pi_{\le n}$. 
\end{myremark}

The following formalization of trees and related notions is standard,
with its variations used e.g.\ in~\cite{CarayolHS14ria}. 
A sequence $w\in \nat^{*}$ is understood as a \emph{position} in a tree.
\begin{mydefinition}[$\Sigma$-tree]\label{def:treesAndRelatedNotions}
Let $\Sigma$ be a ranked alphabet, with each element $\sigma\in \Sigma$
coming with its arity $|\sigma|\in \nat$. 
A \emph{$\Sigma$-tree} $\tau$ is given by a nonempty subset
$\Dom(\tau)\subseteq\nat^{*}$ (called the \emph{domain} of $\tau$)
and a \emph{labeling} function $\tau\colon \Dom(\tau)\to\Sigma$ that
are subject to the following conditions.\footnote{We shall use the same
notation $\tau$ for a tree itself and its labeling
function. Confusion is unlikely.}
\begin{enumerate}
  \item $\Dom(\tau)$ is \emph{prefix-closed}: for any $w\in \nat^{*}$  and
        $i\in \nat$, $wi\in \Dom(\tau)$ implies $w\in \Dom(\tau)$.
        See Fig.~\ref{fig:positionInATree}.
  \item $\Dom(\tau)$ is \emph{lower-closed}:  for any $w\in \nat^{*}$  and
        $i,j\in \nat$, $wj\in \Dom(\tau)$ and $i\le j$ imply $wi\in
        \Dom(\tau)$. 
        See Fig.~\ref{fig:positionInATree}.
  \item The branching degrees are consistent with the label arities: for any $w\in
        \Dom(\tau)$, let $\sigma=\tau(w)$. Then $w0,w1,\dotsc,
        w(|\sigma|-1)$ belong to $\Dom(\tau)$, and $wi\not\in
        \Dom(\tau)$ for any $i$ such that $|\sigma|\le i$.
        See Fig.~\ref{fig:arityOfLabelsAndNumOfSucc}.
\end{enumerate}

The set of all $\Sigma$-trees shall be denoted by $\myTree_{\Sigma}$.
\end{mydefinition}

\begin{figure}[tbp]
\begin{minipage}{.5\textwidth}\centering
       \begin{tikzpicture}[level distance=3em,sibling distance=1em]
 \Tree[.$\varepsilon$ [.$0$ [.$00$ [.{\scriptsize $\vdots$} ]
	[.{\scriptsize $\vdots$} ] ] ] [.$1$ [.$10$ ] [.$11$
	[.{\scriptsize $\vdots$} ] ] [.$12$ ] ] ]
       \end{tikzpicture}
\caption{Positions in a tree}
\label{fig:positionInATree}
\end{minipage}%
\begin{minipage}{.5\textwidth}\centering
       \begin{tikzpicture}[level distance=3em,sibling distance=1em]
 \Tree[.$\sigma^{(2)}_{}$ [.$\sigma^{(1)}_{}$ [.$\sigma^{(2)}_{}$ [.{\scriptsize $\vdots$} ]
	[.{\scriptsize $\vdots$} ] ] ] [.$\sigma^{(3)}_{}$ [.$\sigma^{(0)}_{}$ ] [.$\sigma^{(1)}_{}$
	[.{\scriptsize $\vdots$} ] ] [.$\sigma^{(0)}_{}$ ] ] ]
       \end{tikzpicture}
\caption{ Arities  of labels, and  the numbers of successors. Here
 $\sigma^{(i)}_{}\in \Sigma$ is assumed to be of arity $i$. }
\label{fig:arityOfLabelsAndNumOfSucc}
\end{minipage}
\end{figure}

The following definitions are almost standard in the tree-automata literature, too.
A notable difference  here, that is for a pedagogical reason, is
that the root of a run is \emph{not} required to be a initial state.
That is also natural in our current coalgebraic study;
in the coalgebraic contexts initial states are usually unspecified.
\begin{mydefinition}[run]\label{def:runOfProbTreeAutom}
A \emph{run} $\rho$ of
an NPTA (Def.~\ref{def:nondetParityTreeAutom})
$\X:=((X_1,\ldots,X_n),\Sigma,\delta,s)$ 
is a (possibly infinite) tree whose nodes are $(\Sigma\times X)$-labeled---
here $X=X_1+\cdots+X_n$---subject to the following conditions.
\begin{enumerate}
  \item (Tree)
    The nonempty subset $\Dom(\rho)\subseteq\nat^{*}$ that is subject to
    the same conditions (of being prefix-closed and lower-closed) as
    for $\Sigma$-trees (Def.~\ref{def:treesAndRelatedNotions}).
  \item (Branching degree)
    The labeling function $\rho\colon
    \Dom(\rho)\to \Sigma\times X$ is
    such that, if $\rho(w)=(\sigma,x)$, then $w$ has precisely
    $|\sigma|$ successors $w0,w1,\dotsc,w(|\sigma|-1)\in \Dom(\rho)$.
  \item (Transition)
    Successors are reachable by a transition, in the sense that
    $\bigl(\sigma_{w}, (x_{w0}, \dotsc, x_{w|\sigma|-1})\bigr) \in \delta(x_{w})$ holds, where
    $\rho(w)$ is labeled with $(\sigma_{w},x_{w})$, and
    $\rho(wi)$ is labeled with $(\sigma_{wi},x_{wi})$ for any $0 \leq i < |\sigma|$.
\end{enumerate}
The set of all runs of the NPTA $\X$ is denoted by $\Run_{\X}$.

A \emph{run} $\rho$ of a PPTA (Def.~\ref{def:generativeProbParityTreeAutom}) is defined similarly,
though it is required to satisfy only Cond.~1--2 in the above.
This relaxed condition is natural---for impossible transitions we simply assign the probability $0$.
The set of all runs of a PPTA $\X$ is also denoted by $\Run_{\X}$.

The map that takes a run $\rho\in\Run_{\X}$, removes its $X$-labels
(i.e.\ applies the first projection to each label), and returns the
resulting $\Sigma$-labeled tree (that is easily seen to be a
$\Sigma$-tree, Def.~\ref{def:treesAndRelatedNotions}) is denoted by
$\DelSt\colon \Run_{\X}\to \myTree_{\Sigma}$. We say that a run $\rho$
is \emph{over} the $\Sigma$-tree $\DelSt(\rho)$.
\end{mydefinition}

A branch of a tree is a maximal path from its root $\varepsilon$.
\begin{mydefinition}[branch]\label{def:branch}
Let $\tau$ be a $\Sigma$-tree. An \emph{(infinitary) branch} of $\tau$ is either: 
\begin{itemize}
  \item an infinite sequence $\pi=\pi_{1}\pi_{2}\dotsc\in\nat^{\omega}$
        (where $\pi_{i}\in\nat$) such that any finite prefix
        $\pi_{\le n}=\pi_{1}\dotsc\pi_{n}$ of it belongs to $\Dom(\tau)$; or
  \item a finite sequence $\pi=\pi_{1}\dotsc\pi_{n}\in \nat^{*}$
        (where $\pi_{i}\in\nat$) that belongs to $\Dom(\tau)$ and
        such that $\pi0\not\in\Dom(\tau)$ (meaning that $\pi$ is a
        leaf of $\tau$, and that $\tau(\pi)$ is a $0$-ary symbol).
\end{itemize}
The set of all branches of a $\Sigma$-tree $\tau$ is denoted by $\Branch(\tau)$.
\end{mydefinition}

\begin{mydefinition}[accepting run]\label{def:acceptingRun}
A run $\rho$ of an NPTA (or a PPTA)
$\X=((X_1,\ldots,X_n),\Sigma,\delta,s)$ is said to be \emph{accepting}
if any branch $\pi\in\Branch(\rho)$ of $\mathcal{X}$
satisfies either of the following conditions:
\begin{itemize}
  \item the branch $\pi$ is an infinite sequence
        $\pi=\pi_{1}\pi_{2}\dotsc\in\nat^{\omega}$, and
        the $X$-labels $x_{\varepsilon},x_{\pi_{1}},x_{\pi_{1}\pi_{2}},\ldots$
        along the branch satisfies the parity acceptance condition, that is,
        $\max\{i\in[1,n]\,\mid\, \text{$x_{\pi_1\ldots\pi_k}\in X_i$ for infinitely many $k\in\omega$}\}$ is even; or
  \item the branch $\pi$ is a finite sequence
        $\pi=\pi_{1}\dotsc\pi_{m}\in\nat^{*}$.
\end{itemize}
The set of all accepting runs over $\X$ is denoted by $\AccRun_{\X}$.
\end{mydefinition}

\begin{mydefinition}[partial $\Sigma$-tree, partial run]
\label{def:partialSigmaTreeAndPartialRun}
A \emph{partial $\Sigma$-tree} $\lambda$ is a finite prefix tree of a
$\Sigma$-tree $\tau$ that is \emph{proper}, in the sense that if a node $w$ of
$\tau$ is in $\lambda$ then all the siblings of the node $w$ are also in
$\lambda$.
Its branching degrees are compatible of arities of the $\Sigma$-labels,
and its leaves are labeled by
an additional symbol $\ast$ (``continuation'') or a $0$-ary symbol $\sigma$.
Precisely: a partial $\Sigma$-tree $\lambda$ is given by a
subset $\Dom(\lambda)\subseteq \nat^{*}$ together with
a labeling function $\lambda\colon \Dom(\lambda)\to (\Sigma \cup \{*\})$, such that:
\begin{enumerate}
  \item $\Dom(\lambda)$ is a nonempty and \emph{finite} subset of $\nat^{*}$, that is
        prefix-closed and lower-closed, in the sense of Def.~\ref{def:treesAndRelatedNotions}.
  \item\emph{(Properness)}
    Let $w\in \Dom(\lambda)$. The labeling function $\lambda$ satisfies:
    \begin{itemize}
      \item if $\lambda(w)=\sigma$,
            then $w0,w1,\dotsc,w(|\sigma|-1)\in \Dom(\lambda)$ and
            $wi\not\in \Dom(\lambda)$ for any $i\ge |\sigma|$ (like in
            Def.~\ref{def:treesAndRelatedNotions}); and
      \item if $\lambda(w)=\ast$,
            then $wi\not\in\Dom(\lambda)$ for any $i\in\nat$.
    \end{itemize}
\end{enumerate}

Similarly, a \emph{partial run} $\xi$ of an NPTA
$\X=((X_1,\ldots,X_n),\Sigma,\delta,s)$ is a finite tree subject to the following.
\begin{enumerate}
  \item\label{item:16012317011}
    Its domain $\Dom(\xi)$ is a nonempty, finite, prefix-closed and
    lower-closed subset of $\nat^{*}$.
  \item\label{item:16012317012}\emph{(Properness)}
    A labeling function
    $\xi\colon \Dom(\xi) \to (\Sigma \cup \{*\})\times X$
    such that, for each $w\in \Dom(\xi)$:
    \begin{itemize}
      \item if $\xi(w)=(\sigma,x)$,
            then $w0,w1,\dotsc,w(|\sigma|-1)\in \Dom(\xi)$ and
            $wi\not\in \Dom(\xi)$ for any $i\ge |\sigma|$; and
      \item if $\xi(w)=(\ast,x)$, then
            $wi\not\in\Dom(\lambda)$ for any $i\in\nat$.
    \end{itemize}
  \item\label{item:16012317013}
    Successors are reachable by a transition, in the sense that
    $(\sigma_{w}, (x_{w0}, \dotsc, x_{w|\sigma|-1})) \in \delta(x_{w})$ holds, where
    $\rho(w)$ is labeled with $(\sigma_{w},x_{w})$ such that $\sigma_w\neq *$, and
    $\rho(wi)$ is labeled with $(\sigma_{wi},x_{wi})$ for any $0 \leq i < |\sigma|$.
\end{enumerate}

A partial run of a PPTA is defined similarly, except that
Cond.~\ref{item:16012317013} in the above is not required.
\end{mydefinition}
A partial run is thought of as an interim result of running an automaton
$\X$, after only finitely many steps. The properness requirement
embodies the intuition that, in one-step  execution of an
automaton from some state, all of the successors of the state (together with
the $\Sigma$-label for the state) are created at once.
See Fig.~\ref{fig:execProbTreeAutom}; each of the five trees there
are examples of partial runs.

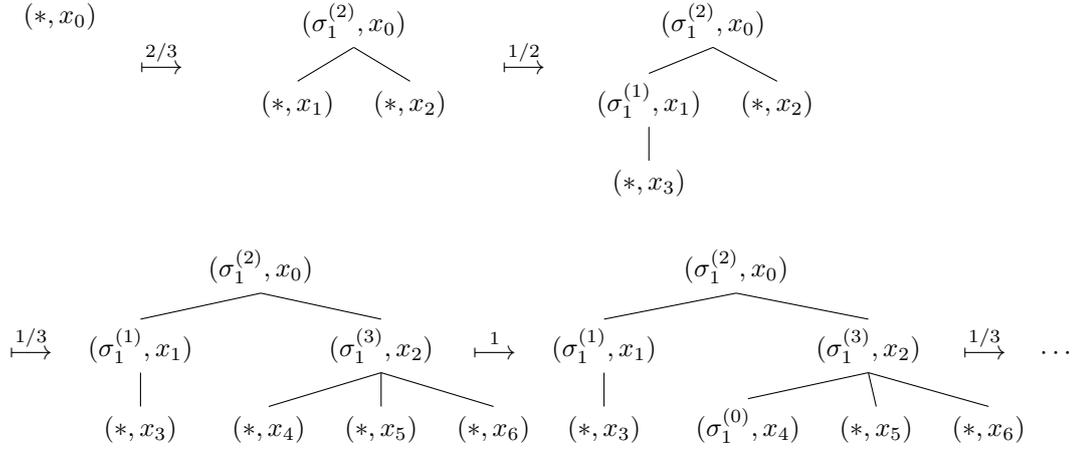
\begin{figure}[tbp]
\begin{align*}
  &
  \hspace{-1.5em}
  \begin{minipage}[t]{.1\textwidth}\centering
    \vspace{0pt}
    \begin{tikzpicture}[level distance=3em,sibling distance=.8em]
      \Tree[.$(\ast,x_{0})$ ]
    \end{tikzpicture}
  \end{minipage}
  \quad\raisebox{-3em}{$\stackrel{2/3}{\longmapsto}$}\quad
  \begin{minipage}[t]{.25\textwidth}\centering
    \vspace{0pt}
    \hspace{0.5em}
    \begin{tikzpicture}[level distance=3em,sibling distance=.8em]
      \Tree[.$(\sigma^{(2)}_{1},x_{0})$
        [.$(\ast,x_{1})$ ]
        [.$(\ast,x_{2})$ ]
      ]
    \end{tikzpicture}
  \end{minipage}
  \quad\raisebox{-3em}{$\stackrel{1/2}{\longmapsto}$}\quad
  \begin{minipage}[t]{.25\textwidth}\centering
    \vspace{0pt}
    \begin{tikzpicture}[level distance=3em,sibling distance=.8em]
      \Tree[.$(\sigma^{(2)}_{1},x_{0})$
        [.$(\sigma^{(1)}_{1},x_{1})$
          [.$(\ast,x_{3})$ ]
        ]
        [.$(\ast,x_{2})$ ]
      ]
    \end{tikzpicture}
  \end{minipage}
  \\
  &
  \hspace{-2.4em}
  \quad\raisebox{-4.5em}{$\stackrel{1/3}{\longmapsto}$}\quad
  \begin{minipage}[t]{.345\textwidth}\centering
    \vspace{0pt}
    \begin{tikzpicture}[level distance=3em,sibling distance=.8em]
      \Tree[.$(\sigma^{(2)}_{1},x_{0})$
        [.$(\sigma^{(1)}_{1},x_{1})$
          [.$(\ast,x_{3})$ ]
        ]
        [.$(\sigma^{(3)}_{1},x_{2})$
          [.$(\ast,x_{4})$ ]
          [.$(\ast,x_{5})$ ]
          [.$(\ast,x_{6})$ ]
        ]
      ]
    \end{tikzpicture}
  \end{minipage}
  \quad\raisebox{-4.5em}{$\stackrel{1}{\longmapsto}$}\quad
  \begin{minipage}[t]{.395\textwidth}\centering
  \vspace{0pt}
    \begin{tikzpicture}[level distance=3em,sibling distance=.8em]
      \Tree[.$(\sigma^{(2)}_{1},x_{0})$
        [.$(\sigma^{(1)}_{1},x_{1})$ 
           [.$(\ast,x_{3})$ ]
        ]
        [.$(\sigma^{(3)}_{1},x_{2})$
          [.$(\sigma^{(0)}_{1},x_{4})$ ]
          [.$(\ast,x_{5})$ ]
          [.$(\ast,x_{6})$ ]
        ]
      ]
    \end{tikzpicture}
  \end{minipage}
  \hspace{-1em}%
  \quad\raisebox{-4.5em}{$\stackrel{1/3}{\longmapsto}\quad\cdots$}
\end{align*}
\caption{Execution of a generative probabilistic tree automaton.
 Here
 $\delta(x_{0})(\sigma^{(2)}_{1},(x_{1},x_{2}))=2/3$,
 $\delta(x_{1})(\sigma^{(1)}_{1},x_{3})=1/2$,
 and so on.
}
\label{fig:execProbTreeAutom}
\end{figure}

\begin{mydefinition}\label{def:subtree}
For $\tau\in\myTree_{\Sigma}$ and $w\in\Dom(\tau)$,
the \emph{$w$-th subtree of $\tau$} is a tree $\tau_w\in\myTree_{\Sigma}$ that is 
defined by $\Dom(\tau_w)=\{w'\in\mathbb{N}^*\mid ww'\in\Dom(\tau)\}$ and
$\tau_w(w')=\tau(ww')$.

A subtree of a run is called a \emph{subrun}.
\end{mydefinition}

\section{Supplementary Materials on Equational Systems}
\label{sec:appendixEqSys}
The intuitions in Def.~\ref{def:eqSys} are put  in the following precise terms. 
\begin{mydefinition}[solution]\label{def:solOfEqSys}
The \emph{solution} of an equational system~(\ref{eq:sysOfEq}) is
 defined as follows, provided that all the necessary greatest and least
 fixed points exist.
For each $i\in[1,n]$ and $j\in[1,i]$, we define monotone functions
\begin{equation*}
  f^{\ddagger}_i\colon
  L_{i}\times L_{i+1}\times\cdots\times L_{n}
  \longrightarrow
  L_{i}
  \quad\text{and}\quad
  l^{(i)}_{j}\colon
  L_{i+1}\times L_{i+2}\times\cdots \times L_{n}
  \longrightarrow
  L_{j}
\end{equation*}
as follows, inductively on $i$. For the base case $i=1$:
\begin{equation*}
  f^{\ddagger}_{1}(l_{1},\dotsc,l_{n})
  \;\coloneqq\;
  f_{1}(l_{1},\dotsc,l_{n})
  \quad\text{and}\quad
  l^{(1)}_{1}(l_{2},\dotsc,l_{n})
  \;\coloneqq\;
  \eta_{1}\bigl[f^{\ddagger}_{1}(\place,l_{2},\dotsc,l_{n})\colon L_{1}\to L_{1}\bigr]
  \,.
\end{equation*}
In the last line we take the lfp or gfp (according to
$\eta_{1}\in\{\mu,\nu\}$) of the (monotone) function
$f^{\ddagger}_{1}(\place,l_{2},\dotsc, l_{n})\colon L_{1}\to L_{1}$.
\auxproof{
Note here that it is possible that the monotone function
$f^{\ddagger}_{1}(\place,l_{2},\dotsc, l_{n}) \colon L_{1}\to L_{1}$
does not have the lfp or gfp.
}

For the step case, the function $f^{\ddagger}_{i+1}$ makes use of the
$i$-th interim solutions $l^{(i)}_{1},\dotsc,l^{(i)}_{i}$ for the variables
$u_{1},\dotsc,u_{i}$ obtained so far:
\begin{equation*}
  f^{\ddagger}_{i+1}(l_{i+1},\dotsc,l_{n})
  \;\coloneqq\;
  f_{i+1}\bigl(\,
    l^{(i)}_{1}(l_{i+1},\dotsc,l_{n}),\,
    \dotsc,\,
    l^{(i)}_{i}(l_{i+1},\dotsc,l_{n}),\,
    l_{i+1},\dotsc,l_{n}
  \,\bigr)
  \enspace.
\end{equation*}
We then let
\begin{equation*}
  l^{(i+1)}_{i+1}(l_{i+2},\dotsc,l_{n})
  \;\coloneqq\;
  \eta_{i+1}\bigl[\,
    f^{\ddagger}_{i+1}(\place,l_{i+2},\dotsc,l_{n})\colon L_{i+1}\to L_{i+1}
  \,\bigr]
\end{equation*}
and use it to obtain the $(i+1)$-th interim solutions $l^{(i+1)}_{1},\dotsc,l^{(i+1)}_{i}$.
That is, for each $j\in [1,i]$,
\begin{equation*}
  l^{(i+1)}_{j}(l_{i+2},\dotsc,l_{n})
  \;\coloneqq\;
  l^{(i)}_{j}\bigl(\,
    l^{(i+1)}_{i+1}(l_{i+2},\dotsc,l_{n}),\,
    l_{i+2},\dotsc,l_{n}
  \,\bigr)
  \enspace.
\end{equation*}
Finally, the \emph{solution}
$(l^{\sol}_{1},\dotsc,l^{\sol}_{n}) \in L_{1}\times\cdots\times L_{n}$
of the equational system~(\ref{eq:sysOfEq}) is defined by
\begin{math}
  (l^{\sol}_{1},\dotsc,l^{\sol}_{n})
  \coloneqq
  (l^{(n)}_{1},\dotsc,l^{(n)}_{n})
\end{math},
where we identify a function $l^{(n)}_{j}\colon 1\to L_{j}$ with an element of $L_{j}$.
It is easy to see that all the functions $f^{\ddagger}_{i}$ and
$l^{(i)}_{j}$ involved here are monotone.
\end{mydefinition}

\begin{myexample}[$\Kl(\giry)(X,1)$ is not a complete lattice]\label{example:KlGiryX1IsNotACompleteLattice}
Since $\giry$ is the \emph{sub-}Giry monad we have that  $\giry 1$ is
isomorphic to the unit interval $[0,1]$, and they are complete lattices.
The homset $\Kl(\giry)(X,1)=\Meas(X,\giry 1)$ is however not a complete
lattice in general, because of the measurability requirement.

For a counterexample let $X=[0,1]$ and $X_{0}\subseteq X$ be a non-measurable subset
(it is well-known that such  $X_{0}$ exists).
For each measurable subset $P\subseteq X$ consider
its characteristic function $\chi_{P}\colon X\to \giry 1\cong[0,1]$,
$\chi_{P}(x)=1$ if $x\in P$ and $\chi_{P}(x)=0$ otherwise.
Then $\chi_{P}$ is measurable and hence an element of $\Kl(\giry)(X,1)$.
Now assume that the supremum
\begin{math}
  f:=\bigsqcup_{P\subseteq X_{0}}\chi_{P}
\end{math}
exists in $\Kl(\giry)(X,1)$.
\begin{itemize}
  \item For each $x\in X_{0}$ we have $f(x)=1$ since $\{x\}\subseteq X_{0}$ is measurable.
  \item For each $x\in X\setminus X_{0}$ we have $f(x)=0$.
        Assume otherwise: then the function $f[x\mapsto 0]$, defined by
        $y\mapsto f(y)$ (if $y\neq x$) and $x\mapsto 0$, is greater than
        $\chi_{P}$ (for each measurable $P\subseteq X_{0}$) and
        measurable (since for every measurable $Q$, the sets $Q\cup\{x\}$ and
        $Q\setminus\{x\}$ are measurable).
        This contradicts with the minimality of the supremum $f$.
\end{itemize}
Therefore we conclude $f=\chi_{X_{0}}$. This is a contradiction, since
 $\chi_{X_{0}}$ is not a measurable function.
\end{myexample}

The following results are about notions of \emph{homomorphism} of
equational systems and preservation of solutions; they are inspired
by a similar result in domain theory (about preservation of least fixed points).
Lem.~\ref{lem:eqSysWithOmegaChainMap} is a rather straightforward
generalization of the domain theory result.
The condition we require in Lem.~\ref{lem:eqSysWithFixedPtIso} is
rather restrictive---especially
Cond.~\ref{item:homOfEqSysStrongRestriction}---but they are satisfied by
our applications.

\begin{mylemma}\label{lem:eqSysWithFixedPtIso}
Let $E$ and $E'$ be the following equational systems, over
posets $L_{1},\dotsc,L_{n}$ and $L'_{1},\dotsc,L'_{n}$, respectively.
Note that the ``polarities'' $\eta_{1},\dotsc,\eta_{n}$ are the same.
\begin{equation*}
  E
  \;\coloneqq\;
  \left[
    \begin{array}{rll}
      u_{1}
      &=_{\eta_{1}}&
      f_{1}(\seq{u_{\i}}{n})
      \\
      &\;\vdots&\\
      u_{n}
      &=_{\eta_{n}}&
      f_{n}(\seq{u_{\i}}{n})
    \end{array}
  \right]
  \qquad
  E'
  \;\coloneqq\;
  \left[
    \begin{array}{rll}
      u'_{1}
      &=_{\eta_{1}}&
      f'_{1}(\seq{u'_{\i}}{n})
      \\
      &\;\vdots&\\
      u'_{n}
      &=_{\eta_{n}}&
      f'_{n}(\seq{u'_{\i}}{n})
    \end{array}
  \right]
\end{equation*}
Let
\begin{equation*}
  \varphi_{1}\colon L_{1} \to L'_{1}\,,
  \quad\dotsc\,,\quad
  \varphi_{n}\colon L_{n} \to L'_{n}
\end{equation*}
be a family of monotone functions, subject to the following conditions.
\begin{enumerate}
  \item\label{item:homOfEqSysCompatibility}
    \begin{math}
      \varphi_{i}\bigl(f_{i}(\seq{l_{\i}}{n})\bigr)
      =
      f'_{i}\bigl(\seq{\varphi_{\i}(l_{\i})}{n}\bigr)
    \end{math} for each $i\in[1,n]$ and $l_{i}\in L_{i}$.
    That is,
    \begin{equation*}
      \vcenter{\xymatrix@R=.6em@C+2em{
        {L_{1}\times\cdots \times L_{n}}
            \ar[d]_{f_{i}}
            \ar[r]^{\varphi_{1}\times\cdots\times\varphi_{n}}
        &
        {L'_{1}\times\cdots \times L'_{n}}
            \ar[d]^{f'_{i}}
        \\
        {L_{i}}
            \ar[r]_{\varphi_{i}}
        &
        {L'_{i}}
     }}
    \end{equation*} commutes
    for each $i\in [1,n]$.
  \item\label{item:homOfEqSysStrongRestriction}
    Let $i\in[1,n]$, and $l_{i+1}\in L_{i+1},\, \dotsc,\, l_{n}\in L_{n}$.
    Let us define the following posets of ``interim fixed points
    under parameters $l_{i+1},\dotsc, l_{n}$.''
    \begin{align*}
      L^{(\seqby{l_{\i}}{i+1}{n})}
      &\;\coloneqq\;
      \bigl\{\,
        (\seq{l_{\i}}{i})
        \mid
        \qnta{j\in[1,i]}\:
        l_{j} = f_{j}(l_{1},\dotsc,l_{i},l_{i+1},\dotsc,l_{n})
      \,\bigr\}
      \\
      L'^{(\seqby{l_{\i}}{i+1}{n})}
      &\;\coloneqq\;
      \bigl\{\,
        (\seq{l'_{\i}}{i})
        \mid
        \qnta{j\in[1,i]}\:
        l'_{j} = f'_{j}\bigl(\seq{l'_{\i}}{i},\seqby{\varphi_{\i}(l_{\i})}{i+1}{n}\bigr)
      \,\bigr\}
    \end{align*}
    Let us further define a function 
    \begin{math}
      \varphi^{(\seqby{l_{\i}}{i+1}{n})}\colon
      L^{(\seqby{l_{\i}}{i+1}{n})}
      \to L'^{(\seqby{l_{\i}}{i+1}{n})}
    \end{math}
    by:
    \begin{displaymath}
      \varphi^{(\seqby{l_{\i}}{i+1}{n})}
      \bigl(l_{1},\dotsc,l_{i}\bigr)
      \;\coloneqq\;
      \bigl(\,
      \varphi_{1}(l_{1}),\dotsc, \varphi_{i}(l_{i})\,\bigr)\enspace,
    \end{displaymath}
    where its well-definedness---i.e.\ that
    $\bigl(\, \varphi_{1}(l_{1}),\dotsc, \varphi_{i}(l_{i})\,\bigr)$
    indeed belongs to $L'^{(\seqby{l_{\i}}{i+1}{n})}$---is
    readily verified from Cond.~\ref{item:homOfEqSysCompatibility}.

    We require that $\varphi^{(\seqby{l_{\i}}{i+1}{n})}$ is an order
    isomorphism, for each $i$ and $l_{i+1},\dotsc,l_{n}$, with its inverse
    denoted by $\psi^{(\seqby{l_{\i}}{i+1}{n})}$.
\end{enumerate}

Under these assumptions,
if the system $E'$ has a solution $\seq{l'^{\sol}_{\i}}{n}$, the other
system $E$ also has a solution $\seq{l^{\sol}_{\i}}{n}$.
Moreover $\seq{\varphi_{\i}(l^{\sol}_{\i})=l'^{\sol}_{\i}}{n}$.
\end{mylemma}
\begin{proof}
By induction on $i$  we shall prove
existence of $\seq{l^{(i)}_{\i}}{i}$ such that
\begin{math}
  \varphi_{j}\bigl(l^{(i)}_{j}(\seqby{l_{\i}}{i+1}{n})\bigr)
  =
  l'^{(i)}_{j}\bigl(\seqby{\varphi_{\i}(l_{\i})}{i+1}{n}\bigr)
\end{math},
for each $j\in[1,i]$ and for any ``parameters'' $\seqby{l_{\i}\in L_{\i}}{i+1}{n}$.
Let us fix $i$ and assume that the claim holds up-to $i-1$. 
There is no need of distinguishing the base case ($i=1$) from the step case:
it is easy to take proper care of the occurrences of $i-1$ in the proof below.
We also assume $\eta_{i}=\mu$; the case when $\eta_{i}=\nu$ is symmetric.

First we shall describe a construction that turns 
 a fixed point of 
 $f'^{\ddagger}_{i}\bigl(\place,\seqby{\varphi(l_{\i})}{i+1}{n}\bigr)$ 
(in $L'_{i}$)
 into 
that of
 $f^{\ddagger}_{i}(\place,\seqby{l_{\i}}{i+1}{n})$
(in $L_{i}$). Recall that we have assumed existence of a solution of $E'$; 
according to Def.~\ref{def:solOfEqSys} this requires that
 $f'^{\ddagger}_{i}\bigl(\place,\seqby{\varphi(l_{\i})}{i+1}{n}\bigr)$ 
has a least point; let it be denoted by
 $\tilde{l}'_{i}$.
Then the following fixed point equality
about $l'^{(i-1)}_{j}\bigl(\tilde{l}'_{i},\,\seqby{\varphi_{\i}(l_{\i})}{i+1}{n}\bigr)$
also holds for each $j\in[1,i-1]$, by the definition of $l'^{(i-1)}_{j}$.
\begin{align*}
  l'^{(i-1)}_{j}\bigl(\tilde{l}'_{i},\,\seqby{\varphi_{\i}(l_{\i})}{i+1}{n}\bigr)
  &\;=\;
  f'_{j}\left(\,
    \begin{multlined}
      l'^{(i-1)}_{1}\bigl(\tilde{l}'_{i},\,\seqby{\varphi_{\i}(l_{\i})}{i+1}{n}\bigr),\:
      \dotsc,
      \\
      l'^{(i-1)}_{j}\bigl(\tilde{l}'_{i},\,\seqby{\varphi_{\i}(l_{\i})}{i+1}{n}\bigr),\:
      \dotsc,
      \\
      l'^{(i-1)}_{i-1}\bigl(\tilde{l}'_{i},\,\seqby{\varphi_{\i}(l_{\i})}{i+1}{n}\bigr),\:
      \tilde{l}'_{i},
      \\
      \seqby{\varphi_{\i}(l_{\i})}{i+1}{n}
    \end{multlined}
  \,\right)
\end{align*}
This means that the following tuple belongs to
$L'^{(\seqby{l_{\i}}{i+1}{n})} \subseteq L'_{1}\times\cdots\times L'_{i}$,
 the domain of
 $\psi^{(\seqby{l_{\i}}{i+1}{n})}$.
\begin{equation*}
  \left(\:
    l'^{(i-1)}_{1}\bigl(\tilde{l}'_{i},\,\seqby{\varphi_{\i}(l_{\i})}{i+1}{n}\bigr),\:
    \dotsc,
    l'^{(i-1)}_{i-1}\bigl(\tilde{l}'_{i},\,\seqby{\varphi_{\i}(l_{\i})}{i+1}{n}\bigr),\:
    \tilde{l}'_{i}
  \:\right)
\end{equation*}
We use a (somewhat confusing) notation of letting $\psi^{(\seqby{l_{\i}}{i+1}{n})}_{i}(\tilde{l}'_{i})$  denote the $i$-th coprojection of
the applied result. That is,
\begin{align*}
  \psi^{(\seqby{l_{\i}}{i+1}{n})}_{i}(\tilde{l}'_{i})
  \;\coloneqq\;
  (\kappa_{i}\co\psi^{(\seqby{l_{\i}}{i+1}{n})})\left(\,
    \begin{multlined}
      l'^{(i-1)}_{1}\bigl(\tilde{l}'_{i},\,\seqby{\varphi_{\i}(l_{\i})}{i+1}{n}\bigr),\:
      \dotsc,
      \\
      l'^{(i-1)}_{i-1}\bigl(\tilde{l}'_{i},\,\seqby{\varphi_{\i}(l_{\i})}{i+1}{n}\bigr),\:
      \tilde{l}'_{i}
    \end{multlined}
  \,\right)
\end{align*}
We shall see that $\psi^{(\seqby{l_{\i}}{i+1}{n})}_{i}(\tilde{l}'_{i})$ is indeed a fixed point of
$f^{\ddagger}_{i}\bigl(\place,\seqby{l_{\i}}{i+1}{n}\bigr)$.
We have
\begin{math}
  \varphi_{i}\bigl(\psi^{(\seqby{l_{\i}}{i+1}{n})}_{i}(\tilde{l}'_{i})\bigr) = \tilde{l}'_{i}
\end{math}, 
 since $\varphi^{(\seqby{l_{\i}}{i+1}{n})}$
is assumed to be the inverse of $\psi^{(\seqby{l_{\i}}{i+1}{n})}$.
Therefore for each $j\in[1,i-1]$ the following holds.
\begin{equation}\label{eq:eqSysFixedPtIsoIndFixedPtExist1}
  l'^{(i-1)}_{j}\left(\tilde{l}'_{i},\,\seqby{\varphi_{\i}(l_{\i})}{i+1}{n}\right)
  \;=\;
  l'^{(i-1)}_{j}\left(\varphi_{i}\bigl(\psi^{(\seqby{l_{\i}}{i+1}{n})}_{i}(\tilde{l}'_{i})\bigr),\,\seqby{\varphi_{\i}(l_{\i})}{i+1}{n}\right)
\end{equation}
Now we use the induction hypothesis
\begin{math}
  \varphi_{j}\bigl(l^{(i-1)}_{j}(\seqby{l_{\i}}{i}{n})\bigr)
  =
  l'^{(i-1)}_{j}\bigl(\seqby{\varphi_{\i}(l_{\i})}{i}{n}\bigr)
\end{math}; substituting
 $\psi^{(\seqby{l_{\i}}{i+1}{n})}_{i}(\tilde{l}'_{i})$ for $l_{i}$ in it
we  have the following.
\begin{equation}\label{eq:eqSysFixedPtIsoIndFixedPtExist2}
  \varphi_{j}\left(\,
    l^{(i-1)}_{j}\bigl(\psi^{(\seqby{l_{\i}}{i+1}{n})}_{i}(\tilde{l}'_{i}),\,\seqby{l_{\i}}{i+1}{n}\bigr)
  \,\right)
  \;=\;
  l'^{(i-1)}_{j}\left(
    \varphi_{i}\bigl(\psi^{(\seqby{l_{\i}}{i+1}{n})}_{i}(\tilde{l}'_{i})\bigr),\,\seqby{\varphi_{\i}(l_{\i})}{i+1}{n}
  \right)
\end{equation}
By (\ref{eq:eqSysFixedPtIsoIndFixedPtExist2}) and
(\ref{eq:eqSysFixedPtIsoIndFixedPtExist1}) we have
\begin{equation*}
  \varphi_{j}\left(\,
    l^{(i-1)}_{j}\bigl(\psi^{(\seqby{l_{\i}}{i+1}{n})}_{i}(\tilde{l}'_{i}),\,\seqby{l_{\i}}{i+1}{n}\bigr)
  \,\right)
  \;=\;
  l'^{(i-1)}_{j}\left(\tilde{l}'_{i},\,\seqby{\varphi_{\i}(l_{\i})}{i+1}{n}\right)
  \,,
\end{equation*}
which implies the following equalities.
\begin{align*}
  \tilde{l}'_{i}
  &\;=\;
  f'_{i}\left(\,
    \begin{multlined}
      \varphi_{1}\left(\,
        l^{(i-1)}_{1}\bigl(\psi^{(\seqby{l_{\i}}{i+1}{n})}_{i}(\tilde{l}'_{i}),\,\seqby{l_{\i}}{i+1}{n}\bigr)
      \,\right),\,
      \dotsc,
      \\
      \varphi_{i-1}\left(\,
        l^{(i-1)}_{i-1}\bigl(\psi^{(\seqby{l_{\i}}{i+1}{n})}_{i}(\tilde{l}'_{i}),\,\seqby{l_{\i}}{i+1}{n}\bigr)
      \,\right),\:
      \varphi_{i}\bigl(\psi^{(\seqby{l_{\i}}{i+1}{n})}_{i}(\tilde{l}'_{i})\bigr),
      \\
      \seqby{\varphi_{\i}(l_{\i})}{i+1}{n}
    \end{multlined}
  \,\right)
  \\
  &\;=\;
  (\varphi_{i}\co f_{i})\left(\,
    \begin{multlined}
      l^{(i-1)}_{1}\bigl(\psi^{(\seqby{l_{\i}}{i+1}{n})}_{i}(\tilde{l}'_{i}),\,\seqby{l_{\i}}{i+1}{n}\bigr),\,
      \dotsc,
      \\
      l^{(i-1)}_{i-1}\bigl(\psi^{(\seqby{l_{\i}}{i+1}{n})}_{i}(\tilde{l}'_{i}),\,\seqby{l_{\i}}{i+1}{n}\bigr),
      \\
      \psi^{(\seqby{l_{\i}}{i+1}{n})}_{i}(\tilde{l}'_{i}),\,
      \seqby{l_{\i}}{i+1}{n}
    \end{multlined}
  \,\right)
  \\
  &\;=\;
  (\varphi_{i}\co f^{\ddagger}_{i})\bigl(
    \psi^{(\seqby{l_{\i}}{i+1}{n})}_{i}(\tilde{l}'_{i}),\,
    \seqby{l_{\i}}{i+1}{n}
  \bigr)
\end{align*}
We shall apply $\psi^{(\seqby{l_{\i}}{i+1}{n})}$ again;
since $\psi^{(\seqby{l_{\i}}{i+1}{n})}$
is the inverse of $\varphi^{(\seqby{l_{\i}}{i+1}{n})}$,
we obtain
\begin{equation*}
  \psi^{(\seqby{l_{\i}}{i+1}{n})}_{i}(\tilde{l}'_{i})
  \;=\;
  f^{\ddagger}_{i}\bigl(
    \psi^{(\seqby{l_{\i}}{i+1}{n})}_{i}(\tilde{l}'_{i}),\,
    \seqby{l_{\i}}{i+1}{n}
  \bigr)
  \,,
\end{equation*}
which means $\psi^{(\seqby{l_{\i}}{i+1}{n})}_{i}(\tilde{l}'_{i})$
is a fixed point of $f^{\ddagger}_{i}(\place,\seqby{l_{\i}}{i+1}{n})$.

Then we focus on the special case
 $\tilde{l}'_{i}=l'^{(i)}_{i}\bigl(\seqby{\varphi_{\i}(l'_{\i})}{i+1}{n}\bigr)$
 (i.e.\ when we specifically choose the least fixed point as
 $\tilde{l}'_{i}$);
we shall
show that $\psi^{(\seqby{l_{\i}}{i+1}{n})}_{i}\bigl(\,l'^{(i)}_{i}\bigl(\seqby{\varphi_{\i}(l'_{\i})}{i+1}{n}\bigr)\,\bigr)$
is the least fixed point. (Recall that $\eta_{i}$ is assumed to be $\mu$.)
Let $\tilde{l}_{i}$ be an arbitrary fixed point, i.e.\  $\tilde{l}_{i}=f^{\ddagger}_{i}(\tilde{l}_i,\seqby{l_{\i}}{i+1}{n})$.
By applying $\varphi_{i}$, we have
\begin{align*}
 & \varphi_{i}(\tilde{l}_{i})
  \\
  &
  \;=\;
  \varphi_{i}\bigl(f^{\ddagger}_{i}(\tilde{l}_i,\seqby{l_{\i}}{i+1}{n})\bigr)
  \\
  &\;=\;
  f'_{i}\left(\,
    \begin{multlined}
      \varphi_{1}\bigl(\,
        l^{(i-1)}_{1}\bigl(\tilde{l}_{i},\seqby{l_{\i}}{i+1}{n}\bigr)
      \,\bigr),\,
      \dotsc,
      \\
      \varphi_{i-1}\bigl(\,
        l^{(i-1)}_{i-1}\bigl(\tilde{l}_{i},\seqby{l_{\i}}{i+1}{n}\bigr)
      \,\bigr),\:
      \varphi_{i}(\tilde{l}_{i}),
      \\
      \seqby{\varphi_{\i}(l_{\i})}{i+1}{n}
    \end{multlined}
  \,\right)
  \\
  &\;=\;
  f'_{i}\left(\,
    \begin{multlined}
      l^{(i-1)}_{1}\bigl(\varphi_{i}(\tilde{l}_{i}),\,\seqby{\varphi_{\i}(l_{\i})}{i+1}{n}\bigr),\,
      \dotsc,
      \\
      l^{(i-1)}_{i-1}\bigl(\varphi_{i}(\tilde{l}_{i}),\,\seqby{\varphi_{\i}(l_{\i})}{i+1}{n}\bigr),\,
      \varphi_{i}(\tilde{l}_{i}),
      \\
      \seqby{\varphi_{\i}(l_{\i})}{i+1}{n}
    \end{multlined}
  \,\right)
   \quad\text{by the induction hypothesis}
  \\
  &\;=\;
  f'^{\ddagger}_{i}\bigl(\varphi_{i}(\tilde{l}_{i}),\,\seqby{\varphi_{\i}(l_{\i})}{i+1}{n}\bigr)
  \,,\quad\text{and}
  \\
  &
  \varphi_{i}\bigl(\,
    \psi^{(\seqby{l_{\i}}{i+1}{n})}_{i}\bigl(l'^{(i)}_{i}(\seqby{l'_{\i}}{i+1}{n})\bigr)
  \,\bigr)
\\
  &\;=\;
  l'^{(i)}_{i}\bigl(\seqby{\varphi_{\i}(l'_{\i})}{i+1}{n}\bigr)
  \,.
\end{align*}
Since $l'^{(i)}_{i}\bigl(\seqby{\varphi_{\i}(l'_{\i})}{i+1}{n}\bigr)$ is the least fixed point of
$f'^{\ddagger}_{i}\bigl(\place,\,\seqby{\varphi_{\i}(l_{\i})}{i+1}{n}\bigr)$,
$l'^{(i)}_{i}\bigl(\seqby{\varphi_{\i}(l'_{\i})}{i+1}{n}\bigr)\sqsubseteq\varphi_{i}(\tilde{l}_{i})$
holds.
Now by applying $\psi^{(\seqby{l_{\i}}{i+1}{n})}_{i}$,
we obtain
\begin{equation*}
  \psi^{(\seqby{l_{\i}}{i+1}{n})}_{i}\bigl(\,l'^{(i)}_{i}\bigl(\seqby{\varphi_{\i}(l'_{\i})}{i+1}{n}\bigr)\,\bigr)
  \;\sqsubseteq\;
  \tilde{l}_{i}
  \;;
\end{equation*}
thus $\psi^{(\seqby{l_{\i}}{i+1}{n})}_{i}\bigl(\,l'^{(i)}_{i}\bigl(\seqby{\varphi_{\i}(l'_{\i})}{i+1}{n}\bigr)\,\bigr)$.

Now we have shown
\begin{math}
  l^{(i)}_{i}(\seqby{l_{\i}}{i+1}{n})
  \;=\;
  \psi^{(\seqby{l_{\i}}{i+1}{n})}_{i}\bigl(\,l'^{(i)}_{i}\bigl(\seqby{\varphi_{\i}(l'_{\i})}{i+1}{n}\bigr)\,\bigr)
\end{math}, from which
\begin{math}
  \varphi_{i}(l^{(i)}_{i}(\seqby{l_{\i}}{i+1}{n}))
  =
  l'^{(i)}_{i}\bigl(\seqby{\varphi_{\i}(l_{\i})}{i+1}{n}\bigr)
\end{math}
easily follows by applying $\varphi_{i}$.
Furthermore, for each $j$ such that $j<i$,  we have
\begin{align*}
  \varphi_{j}(l^{(i)}_{j}(\seqby{l_{\i}}{i+1}{n}))
  &\;=\;
  \varphi_{j}\bigl(\,
    l^{(i-1)}_{j}\bigl(l^{(i)}_{i}(\seqby{l_{\i}}{i+1}{n}),\,\seqby{l_{\i}}{i+1}{n}\bigr)
  \,\bigr)
  \\
  &\;=\;
  l^{(i-1)}_{j}\bigl(\,
    \varphi_{i}\bigl(l^{(i)}_{i}(\seqby{l_{\i}}{i+1}{n})\bigr),\,\seqby{\varphi_{\i}(l_{\i})}{i+1}{n}
  \,\bigr)
  \\
  &\;=\;
  l^{(i-1)}_{j}\bigl(\,
    l^{(i)}_{i}\bigl(\seqby{\varphi_{\i}(l_{\i})}{i+1}{n}\bigr),\,\seqby{\varphi_{\i}(l_{\i})}{i+1}{n}
  \,\bigr)
  \\
  &\;=\;
  l'^{(i)}_{j}\bigl(\seqby{\varphi_{\i}(l_{\i})}{i+1}{n}\bigr)
  \,.
  \qedhere
\end{align*}
\end{proof}

\begin{mylemma}
\label{lem:eqSysWithOmegaChainMap}
Let $E$ and $E'$ be equational systems, as in Lem.~\ref{lem:eqSysWithFixedPtIso},
over $L_{1},\dotsc,L_{n}$ and $L'_{1},\dotsc,L'_{n}$, respectively.
We assume the same conditions as in Lem.~\ref{lem:eqSysSolvableOmegaCont} for both $E$ and $E'$, that is:
all $L_{i}$ and $L'_{i}$ are pointed $\omega$/$\omega^{\op}$-cpo's; and
all $f_{i}$ and $f'_{i}$ are $\omega$/$\omega^{\op}$-continuous.
Let 
\begin{equation*}
  \varphi_{1}\colon L_{1} \to L'_{1}\,,
  \quad\dotsc\,,\quad
  \varphi_{n}\colon L_{n} \to L'_{n}
\end{equation*}
be monotone functions such that:
\begin{enumerate}
  \item each $\varphi_{i}$ is both $\omega$-continuous and
        $\omega^{\op}$-continuous;
  \item each $\varphi_{i}$ preserves greatest and least elements
        ($\varphi_{i}(\top) = \top$ and $\varphi_{i}(\bot) = \bot$); and
  \item the following diagram commutes for each $i\in [1,n]$.
    \begin{displaymath}
      \vcenter{\xymatrix@R=.6em@C+2em{
        {L_{1}\times\cdots \times L_{n}}
            \ar[d]_{f_{i}}
            \ar[r]^{\varphi_{1}\times\cdots\times\varphi_{n}}
        &
        {L'_{1}\times\cdots \times L'_{n}}
            \ar[d]^{f'_{i}}
        \\
        {L_{i}}
            \ar[r]_{\varphi_{i}}
        &
        {L'_{i}}
      }}
    \end{displaymath}
\end{enumerate}
Then $\varphi_{i}(l^{\sol}_{i}) = l'^{\sol}_{i}$ holds for each $i\in[1,n]$.
\end{mylemma}
\begin{proof}
By induction on $i$, we shall show
\begin{math}
  \varphi_{i}\bigl(l^{(i)}_{j}(\seqby{l_{\i}}{i+1}{n})\bigr)
  =
  l'^{(i)}_{j}\bigl(\seqby{\varphi_{\i}(l_{\i})}{i+1}{n}\bigr)
\end{math}
and
\begin{math}
  \varphi_{i}(f^{\ddagger}_{i}\bigl(\seqby{l_{\i}}{i+1}{n})\bigr)
  =
  f'^{\ddagger}_{i}\bigl(\seqby{\varphi_{\i}(l_{\i})}{i+1}{n}\bigr)
\end{math}.
As in the proof of Lem.~\ref{lem:eqSysSolvableOmegaCont},
we do not distinguish the case $i=1$, and assume $\eta_{i}=\mu$.

The structure of the proof also resembles to that of Lem.~\ref{lem:eqSysSolvableOmegaCont}.
We can easily check the claim for the function $f^{\ddagger}_{i}$, by induction hypothesis.
For the solution $l^{(i)}_{i}$, recall that the proof of Lem.~\ref{lem:eqSysSolvableOmegaCont} asserts that $l^{(i)}_{i}$ is equal to
$\bigsqcup_{j<\omega}[f^{\ddagger}_{i}(\place,\seqby{l_{\i}}{i+1}{n})]^{j}(\bot)$,
with continuity of $f^{\ddagger}_{i}$.
Thus we have
\begin{math}
  \varphi_{i}\bigl(l^{(i)}_{i}(\seqby{l_{\i}}{i+1}{n})\bigr)
  =
  l'^{(i)}_{i}\bigl(\seqby{\varphi_{\i}(l_{\i})}{i+1}{n}\bigr)
\end{math}
by straightforward induction.
Then the claim easily follows also for each $l'^{(i)}_{j}$.
\end{proof}

\section{Generative Probabilistic Parity Tree Automata}
\label{sec:GenProbParityTreeAutom}
\subsection{Generative Systems and Reactive Systems}
\label{subsec:GenAndReactSys}
The notion of probabilistic tree automaton we study in this paper as an
example is a \emph{generative} one. This is in contrast to
\emph{reactive} probabilistic systems
(studied e.g.\ in~\cite{CarayolHS14ria}): a generative system
\emph{generates} a (possibly infinite) tree
(Fig.~\ref{fig:execProbTreeAutom} is a step-by-step illustration
of a generation process)---hence the probability 
with which each single tree is generated is zero except for some
singular cases---whereas a reactive system takes a tree as input and 
assigns a probability to it.
The difference can be technically
formulated in the types of transition functions:
\begin{align*}
& X\longrightarrow \giry(\coprod_{\sigma\in\Sigma} X^{|\sigma|})
\quad\text{for generative;} \\ %\qquad
& X\longrightarrow \prod_{\sigma\in\Sigma}\giry( X^{|\sigma|})
\quad\text{for reactive.}
\end{align*}
The difference has been discussed extensively for \emph{word} (instead of tree) automata.
See e.g.~\cite{vanGlabbeekSST90rga,Sokolova05cao,Doberkat09scl}.

In the current generative (as opposed to reactive)
setting, it does not make much sense to talk about the probability with
which each single tree is generated. For example let
$\Sigma=\{\hd,\tl\}$ and assume that each operation is unary.
The generative automaton in Fig.~\ref{fig:faircoin} is then a model of a fair coin; and
it generates any single infinite sequence with probability $0$.
This is a prototypical one that motivates the need for measure
theory in the context of probabilistic systems (like in~\cite{Panangaden09lmp});
note that the set of $\Sigma$-trees (that are $\Sigma$-words if all symbols are unary)
is uncountable.

\begin{figure}[tbp]
  \begin{math}
  \vcenter{\xymatrix{
    *++[o][F=]{\phantom{3}}
      \ar@(ul,ur)[]^{\hd[\frac{1}{2}]}
      \ar@(dr,dl)[]^{\tl[\frac{1}{2}]}
    \save+<-2pc,0pc>\ar[]\restore
  }}
  \end{math}
  \caption{A generative automaton for a fair coin.}
  \label{fig:faircoin}
\end{figure}
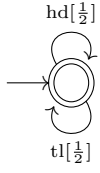

\subsection{Languages of Generative Probabilistic Tree Automata}
\begin{mydefinition}[measurable structures of $\myTree_{\Sigma}$ and $\Run_{\X}$]
\label{def:measurableStrOnTreesAndRuns}
Let $\lambda$ be a partial $\Sigma$-tree.
The \emph{cylinder set} associated to $\lambda$, denoted by
$\Cyl_{\Sigma}(\lambda)$, is the set of (proper, non-partial) $\Sigma$-trees
that have $\lambda$ as their ``prefix.'' That is,
\begin{equation*}
  \Cyl_{\Sigma}(\lambda)
  \;\coloneqq\;
  \bigl\{
    \tau\in\myTree_{\Sigma}
    \,\big|\,
    \forall w\in \Dom(\lambda).\;
    \bigl(
      \lambda(w)=\tau(w)
      \text{ or }
      \lambda(w)=\ast
    \bigr)
  \bigr\}
  \,.
\end{equation*}
The (smallest) $\sigma$-algebra generated by the family
$\{\Cyl_{\Sigma}(\lambda)\mid \text{$\lambda$ is a partial $\Sigma$-tree}\}$
will be denoted by $\mathfrak{F}_{\Sigma}$. 

For a partial run $\xi$ of $\X$, the \emph{cylinder set}
$\Cyl_{\X}(\xi)\subseteq \Run_\X $ associated to $\xi$ is defined
similarly but slightly differently. Precisely:
\begin{equation*}
  \Cyl_{\X}(\xi)
  \;\coloneqq\;
  \left\{
    \rho\in\Run_{\X}
    \,\middle|\,
    \forall w\in \Dom(\xi).\;
    \left(
      \begin{aligned}
        &
        \pi_{1}(\xi(w))=\pi_{1}(\rho(w))
        \text{ or }
        \pi_{1}(\xi(w))=\ast
        \text{; and}
        \\
        &
        \pi_{2}(\xi(w))=\pi_{2}(\rho(w))
      \end{aligned}
    \right)
  \right\}
  \,.
\end{equation*}
Here $\pi_{1}(\sigma,x)=\sigma$ and $\pi_{2}(\sigma,x)=x$.
% Similarly, given a partial run $\xi$ of $\X$, we define the
% \emph{cylinder set} $\Cyl_{\X}(\xi)\subseteq \Run_\X $ associated to $\xi$ to be
% the set of all runs of $\X$ that have $\xi$ as their prefix. 
These cylinder sets generate a $\sigma$-algebra over
$\Run_{\X}$, which shall be denoted by $\mathfrak{F}_{\X}$.
\end{mydefinition}

\auxproof{
The notion of \emph{no-divergence}, in Def.~\ref{def:LangForPPTAs} defined
as the function $\NDL_{\X}\colon X\to [0,1]$ by the infimum
$\bigwedge_{k\in\nat} \NDL_{\X,k}(x)$, might need some explanation.
\begin{myremark}[$\NDL_{\X},\DL_{\X}$]\label{rem:noDeadendDeadendProb}
Let $\DL_{\X}(x):=1-\NDL_{\X}(x)$ (the probability for \emph{divergence}),
$\DL_{\X,k}(x):=1-\NDL_{\X,k}(x)$ (the probability for
observing divergence within depth $k$) and let $\DL_{\X,=k}(x):=
\DL_{\X,k}(x)-\DL_{\X,k-1}(x)$ for each $k\ge 1$. The latter probability
$\DL_{\X,=k}(x)$ is for the event ``the first divergence is encountered 
exactly at depth $k$''; and we can take the following equality for
granted, since the events for $\DL_{\X,=k}(x)$ for different $k$ are
clearly mutually
exclusive.
\begin{displaymath}
 \DL_{\X}(x) \;=\; \sum_{k=1}^{\infty} \DL_{\X,=k}(x)
\end{displaymath}
Now we have
\begin{multline*}
 \NDL_{\X}(x)
=
1-\DL_{\X}(x)
\\
\stackrel{\text{above}}{=}
1-\sum_{k=1}^{\infty} \DL_{\X,=k}(x)
=
1-\lim_{k\to\infty} \DL_{\X,k}(x)
\\
=
\lim_{k\to\infty} 1-\DL_{\X,k}(x)
\\
=
\lim_{k\to\infty} \NDL_{\X,k}(x)
=
\bigwedge_{k\in\nat} \NDL_{\X,k}(x)\enspace,
\end{multline*}
which derives the definition in Def.~\ref{def:LangForPPTAs}. 

Note that the infimum $\bigwedge_{k\in\nat}  \NDL_{\X,k}(x)$
is in fact a limit $\lim_{k\to\infty} \NDL_{\X,k}(x)$ for any $x\in X$,
because the family $\bigl(\,\NDL_{\X,k}(x)\,\bigr)_{k\in\nat}$ is easily seen to be a
descending chain with a lower bound $0$.
\end{myremark}

The notion of \emph{no-divergence}, in Def.~\ref{def:NoDivergence},
can be characterized as the (greatest) fixed-point, as follows.
We will use this fixed-point characterization of $\NDL_{\X}$ later.
A similar observation is found in~\cite{EtessamiSY15gfp}.
}

\begin{mylemma}\label{lem:NDLAsAGreatestFixedPoint}
 The function $\NDL_{\X}\colon X\to [0,1]$ is the greatest fixed point 
 of the function $\Psi'_{\X}\colon [0,1]^{X}\to [0,1]^{X}$,
 defined by% (cf.~(\ref{eq:10151159}))
\begin{equation*}%\label{eq:10151732}
  \Psi_{\X}'(f)(x)
  \;\coloneqq\;
  \textstyle
  \sum\limits_{
    (\sigma,(\seq{x_{\i}}{|\sigma|}))
    \in \coprod_{\sigma\in \Sigma} X^{|\sigma|}
  }
  \delta(x)\bigl(\sigma,(\seq{x_{\i}}{|\sigma|})\bigr)\cdot
  \prod_{i\in[1,|\sigma|]}
  \NDL_{\X,k} (x_{i})
  \,.
\end{equation*}
\end{mylemma}
\begin{proof}
The proof is essentially by Kleene's fixed point theorem.
Consider the sequence $\NDL_{\X,0},\NDL_{\X,1},\dotsc\colon X\to
[0,1]$ in Def.~\ref{def:NoDivergence}.
Then $\NDL_{\X,0}$ is the greatest element in $[0,1]^{X}$
(with respect to the pointwise order) and
the sequence is obviously decreasing.
Moreover, the function $\Psi'_{\X}$ is easily seen
to be ``continuous'' in the sense that
$\Psi'_{\X}(\bigwedge_{k\in\nat}\allowbreak\NDL_{\X,k})=\bigwedge_{k\in\nat}\Psi'_{\X}(\NDL_{\X,k})$.
Therefore by an argument similar to the one for Kleene's theorem,
$\bigwedge_{k\in\nat}\allowbreak\NDL_{\X,k}=\NDL_{\X}$
is the greatest fixed point of $\Psi'_{\X}$.
\end{proof}

\begin{mylemma}\label{lem:muXOnRunXExt}
The subprobability pre-measure $\mu_{\X}^{\Run}$ over cylinder sets
defined in~(\ref{eq:01171719}) of
Def.~\ref{def:NoDivergence} determines uniquely a subprobability measure
over the whole $\sigma$-algebra $\sigalg_{\X}$.
\end{mylemma}
\begin{proof}
We rely on Carath\'{e}odory's extension theorem~\cite{Doberkat09scl} here.
For using the theorem, since we have
\begin{displaymath}
  \Cyl_{\X}(\xi)
  \;=\;
  \coprod\nolimits_{
    (\sigma,(\seq{x_{\i}}{|\sigma|}))
    \in \coprod\limits_{\sigma\in \Sigma} X^{|\sigma|}
  }
  \Cyl_{\X}(\xi_{w,\sigma,\seq{x_{\i}}{|\sigma|}})
  \,,
\end{displaymath}
it suffices to show what follows.
\begin{quotation}
  Let $\xi$ be a partial run of $\X$, $w\in\nat^{*}$ be such that $w\in
  \Dom(\xi)$ and $\xi(w)=(\ast, x)$ (hence $w$ is a leaf of $\xi$).
  For each $\sigma\in\Sigma$ and $x_{0},\dotsc,x_{|\sigma|-1}\in X$,
  let $\xi_{w,\seq{x_{\i}}{|\sigma|}}$ be the partial run
  that ``extends'' the leaf $w$ with
  $\bigl(\sigma,(\seq{x_{\i}}{|\sigma|})\bigr)$. Precisely:
  \begin{align*}
    \Dom(\xi_{w,\seq{x_{\i}}{|\sigma|}})
    &\;:=\;
    \Dom(\xi)\cup\{\seq{w\i}{|\sigma|}\}
    \,,
    \\
    \xi_{w,\sigma,\seq{x_{\i}}{|\sigma|}}(w')
    &\;:=\;
    \begin{cases}
      (\sigma,\xi(w)) &\text{if $w'=w$}
      \\
      (\ast,x_{i}) &\text{if $w'=wi$}
      \\
      \xi(w') &\text{otherwise.}
    \end{cases}
  \end{align*}
  Then
  \begin{displaymath}
    \mu_{\X}^{\Run_{\X}}\bigl(\Cyl_{\X}(\xi)\bigr)
    \;=\;
    \sum\nolimits_{
      (\sigma,(\seq{x_{\i}}{|\sigma|}))
      \in \coprod\limits_{\sigma\in \Sigma} X^{|\sigma|}
    }
    \mu_{\X}^{\Run_{\X}}\bigl(\Cyl_{\X}(\xi_{w,\sigma,\seq{x_{\i}}{|\sigma|}})\bigr).
  \end{displaymath}
\end{quotation}
To show this claim, by the bottom-up way of
the definition of $P_{\X}$ (Def.~\ref{def:NoDivergence}),
it suffices to show that
\begin{equation*}
%   &P_{\X}(x) =
% \sum_{\bigl(\sigma,(x_{0},\dotsc,x_{|\sigma|-1})\bigr)\in
%   \coprod_{\sigma\in \Sigma} X^{|\sigma|}}
%      P_{\X}\left(
%    \begin{minipage}{.25\textwidth}\centering
%        \begin{tikzpicture}[level distance=3em,sibling distance=.5em]
%  \Tree[.$(\sigma,x)$
% 	[.$x_{0}$ 
% 	]
%              \edge[draw=none] node {};
%              [.\node (hoge) {$\cdots$}; ]
% 	[.$x_{|\sigma|-1}$ 
% 	]
% 	]
%        \end{tikzpicture}
%    \end{minipage}
%   \right)\enspace, \quad\text{that is}
%  \\
  \NDL_{\X}(x)
  \;=\;
  \sum\nolimits_{
    (\sigma,(\seq{x_{\i}}{|\sigma|}))
    \in \coprod\limits_{\sigma\in \Sigma} X^{|\sigma|}
  }
  \delta(x)
  \bigl(\sigma,(\seq{x_{\i}}{|\sigma|})\bigr)
  \cdot
  \prod_{i\in[1,|\sigma|]}
  \NDL_{\X}(x_{i})
  \,.
\end{equation*}
This just means that $\NDL_{\X}$ is a fixed point of $\Psi'_{\X}$,
a fact proved in Lem.~\ref{lem:NDLAsAGreatestFixedPoint}.
\end{proof}

\section{Omitted Proofs}
\label{sec:omittedProofs}
\subsection{Proof of Lem.~\ref{lem:eqSysSolvableOmegaCont}}
\begin{proof}
By induction on $i$, we shall show $\omega$- and $\omega^{\op}$-continuity of
$f^{\ddagger}_{i}$ and $l^{(i)}_{j}$ (here $j\leq i$), and existence of the solution $l^{(i)}_{i}$.
(Monotonicity of those is almost clear.)
Let us fix $i$ and assume that the claim holds up-to $i-1$. 
There is no need of distinguishing the base case ($i=1$) from the step case:
it is easy to take proper care of the occurrences of $i-1$ in the proof below.
We also assume $\eta_{i}=\mu$; the case when $\eta_{i}=\nu$ is symmetric.
% %since we can consider the empty equational system, which has the obvious solution of the nullary product.

We can easily show that the function
\begin{equation}\label{eq:fDdaggCont}
  f^{\ddagger}_{i}(l_{i},\dotsc,l_{n})
  \;\coloneqq\;
  f_{i}\bigl(\,
    l^{(i-1)}_{1}(l_{i},\dotsc,l_{n}),\,
    \dotsc,\,
    l^{(i-1)}_{i-1}(l_{i},\dotsc,l_{n}),\,
    l_{i},\dotsc,l_{n}
  \,\bigr)
\end{equation}
is $\omega$- and $\omega^{\op}$-continuous,
by continuity of $l^{(i-1)}_{j}$ on induction hypothesis, and
continuity of $f_{i}$ in the assumption.
By Kleene's fixed point theorem
we can construct $l^{(i)}_{i}(\seqby{l_{\i}}{i+1}{n})$---which is
defined to be the least fixed point of $f^{\ddagger}_{i}(\place,\seqby{l_{\i}}{i+1}{n})$---
together with the above $\omega$-continuity of $f^{\ddagger}_{i}$ (\ref{eq:fDdaggCont}).
We let
\begin{equation}\label{eq:lKleeneDef}
  \textstyle
  l^{(i)}_{i}(\seqby{l_{\i}}{i+1}{n})
  \;=\;
  \bigsqcup_{j<\omega}[f^{\ddagger}_{i}(\place,\seqby{l_{\i}}{i+1}{n})]^{j}(\bot)
  \,.
\end{equation}
Since $\bot$ is the least element in $L_{i}$,
we are ensured to obtain an $\omega$ chain
of $\bigl([f^{\ddagger}_{i}(\place,\seqby{l_{\i}}{i+1}{n})]^{j}(\bot)\bigr)_j$.
The supremum is a fixed point because we have
\begin{align*}
  &
  \textstyle
  f^{\ddagger}_{i}\bigl(\,
    \bigsqcup_{j<\omega}[f^{\ddagger}_{i}(\place,\seqby{l_{\i}}{i+1}{n})]^{j}(\bot)
    ,\,
    \seqby{l_{\i}}{i+1}{n}
  \,\bigr)
  \\
  &\;=\;
  \textstyle
  f^{\ddagger}_{i}\bigl(\,
    \bigsqcup_{j<\omega}[f^{\ddagger}_{i}(\place,\seqby{l_{\i}}{i+1}{n})]^{j}(\bot)
    ,\,
    \seqby{\bigsqcup_{j<\omega}l_{\i}}{i+1}{n}
  \,\bigr)
  &
  \textstyle
  \text{by $l_{i'} = \bigsqcup_{j<\omega}l_{i'}$}
  \\
  &\;=\;
  \textstyle
  \bigsqcup_{j<\omega}f^{\ddagger}_{i}\bigl(\,
    [f^{\ddagger}_{i}(\place,\seqby{l_{\i}}{i+1}{n})]^{j}(\bot)
    ,\,
    \seqby{l_{\i}}{i+1}{n}
  \,\bigr)
  &
  \textstyle
  \text{by induction hypothesis}
  \\
  &\;=\;
  \textstyle
  \bigsqcup_{j<\omega}[f^{\ddagger}_{i}(\place,\seqby{l_{\i}}{i+1}{n})]^{j+1}(\bot)
  \\
  &\;=\;
  \textstyle
  \bigsqcup_{j<\omega}[f^{\ddagger}_{i}(\place,\seqby{l_{\i}}{i+1}{n})]^{j}(\bot)
  \,.
\end{align*}
The obtained fixed point is readily verified to be the least, thanks to the minimality of $\bot$.

Now we show $\omega$-continuity of $l^{(i)}_{i}$.
To this end we use the following easy observation:
it is shown for  each $j<\omega$ by induction.
\begin{equation}\label{eq:fDdaggIterCont}
  \textstyle
  [f^{\ddagger}_{i}(\,\place,\,
    \bigsqcup_{k<\omega}l_{i+1,k}
    ,\,\dotsc,\,
    \bigsqcup_{k<\omega}l_{n,k}
  \,)]^{j}(\bot)
  \;=\;
  \textstyle
  \bigsqcup_{k<\omega}
  [f^{\ddagger}_{i}(\place,
    l_{i+1,k}
    ,\,\dotsc,\,
    l_{n,k}
  )]^{j}(\bot)
  \,.
\end{equation}
By taking supremum of the above for $j<\omega$,
the $\omega$-continuity of $l^{(i)}_{i}$ is shown as follows.
\begin{align*}
  &
  \textstyle
  l^{(i)}_{i}\bigl(\,
    \bigsqcup_{k<\omega}l_{i+1,k}
    ,\,\dotsc,\,
    \bigsqcup_{k<\omega}l_{n,k}
  \,\bigr)
  \\
  &\;=\;
  \textstyle
  \bigsqcup_{j<\omega}
  [f^{\ddagger}_{i}(\,\place,\,
    \bigsqcup_{k<\omega}l_{i+1,k}
    ,\,\dotsc,\,
    \bigsqcup_{k<\omega}l_{n,k}
  \,)]^{j}(\bot)
  &
  \textstyle
  \text{by (\ref{eq:lKleeneDef})}
  \\
  &\;=\;
  \textstyle
  \bigsqcup_{j<\omega}
  \bigsqcup_{k<\omega}
  [f^{\ddagger}_{i}(\place,
    l_{i+1,k}
    ,\,\dotsc,\,
    l_{n,k}
  )]^{j}(\bot)
  &
  \textstyle
  \text{by (\ref{eq:fDdaggIterCont})}
  \\
  &\;=\;
  \textstyle
  \bigsqcup_{k<\omega}
  \bigsqcup_{j<\omega}
  [f^{\ddagger}_{i}(\place,
    l_{i+1,k}
    ,\,\dotsc,\,
    l_{n,k}
  )]^{j}(\bot)
  \\
  &\;=\;
  \textstyle
  \bigsqcup_{k<\omega}
  l^{(i)}_{i}\bigl(\,
    l_{i+1,k}
    ,\,\dotsc,\,
    l_{n,k}
  \,\bigr)
  &
  \textstyle
  \text{by (\ref{eq:lKleeneDef})}
\end{align*}

Next we show $\omega^{\op}$-continuity of $l^{(i)}_{i}$.
It suffices to show that
\begin{math}
  \textstyle
  \bigsqcap_{k<\omega}
  l^{(i)}_{i}\bigl(
    l_{i+1,k}
    ,\dotsc,
    l_{n,k}
  \bigr)
\end{math}
is the least (pre-) fixed point of
\begin{math}
  f^{\ddagger}_{i}\bigl(\,\place,\,
    \bigsqcap_{k<\omega}
    l_{i+1,k}
    ,\,\dotsc,\,
    \bigsqcap_{k<\omega}
    l_{n,k}
  \,\bigr)
\end{math},
since $l^{(i)}_i(\seqby{l_{\i}}{i+1}{n})$ is defined to be
$\lfp[f^{\ddagger}_{i}(\place,\seqby{l_{\i}}{i+1}{n})]$.
Let us take an arbitrary pre-fixed point
\begin{math}
  \tilde{l}_{i}
  \,\sqsupseteq\,
  \textstyle
  f^{\ddagger}_{i}\bigl(\,\tilde{l}_{i},\,
    \bigsqcap_{k<\omega}
    l_{i+1,k}
    ,\,\dotsc,\,
    \bigsqcap_{k<\omega}
    l_{n,k}
  \,\bigr)
\end{math}.
Then we have
\begin{align*}
  \tilde{l}_{i}
  &\;\sqsupseteq\;
  \textstyle
  f^{\ddagger}_{i}\bigl(\,
    \bigsqcap_{k<\omega}
    \tilde{l}_{i},\,
    \bigsqcap_{k<\omega}
    l_{i+1,k}
    ,\,\dotsc,\,
    \bigsqcap_{k<\omega}
    l_{n,k}
  \,\bigr)
  \\
  &\;=\;
  \textstyle
  \bigsqcap_{k<\omega}
  f^{\ddagger}_{i}(\tilde{l}_{i},\,
    l_{i+1,k}
    ,\dotsc,
    l_{n,k}
  )
  &
  \text{by induction hypothesis}
  \\
  &\;\sqsupseteq\;
  \textstyle
  \bigsqcap_{k\in\omega}
  f^{\ddagger}_{i}\bigl(\,
    l^{(i)}_{i}(
      l_{i+1,k}
      ,\dotsc,
      l_{n,k}
    )
    ,\,
    l_{i+1,k}
    ,\dotsc,
    l_{n,k}
  \,\bigr)
  &
  \text{$l^{(i)}_{i}$ is the \emph{least} (pre-) fixed point}
  \\
  &\;=\;
  \textstyle
  \bigsqcap_{k\in\omega}
  l^{(i)}_{i}(
    l_{i+1,k}
    ,\dotsc,
    l_{n,k}
  )
  \,,
\end{align*}
thus
\begin{math}
  \textstyle
  \bigsqcap_{k<\omega}
  l^{(i)}_{i}\bigl(
    l_{i+1,k}
    ,\dotsc,
    l_{n,k}
  \bigr)
\end{math}
is the least (pre-) fixed point.

We have shown $\omega$- and $\omega^{\op}$-continuity of $l^{(i)}_{i}$;
for the other interim solutions
\begin{equation*}
  l^{(i)}_{j}(l_{i+1},\dotsc,l_{n})
  \;\coloneqq\;
  l^{(i-1)}_{j}\bigl(\,
    l^{(i)}_{i}(l_{i+1},\dotsc,l_{n}),\,
    l_{i+1},\dotsc,l_{n}
  \,\bigr)
  \,,
\end{equation*}
for $j<i$, continuity is also shown, in the same manner as in~(\ref{eq:fDdaggCont}).
\end{proof}

\subsection{Proof of Lem.~\ref{lem:eqSysCharacterizationOfAcceptingRuns}}
\begin{proof}
For a branch $\pi=(x_1,\sigma_1)\cdot(x_2,\sigma_2)\cdots$ of
a run $\rho\in\Run_{\X}$,
let $|\pi|$ be the length of $\pi$, and %let
$|\pi|_{=j} \coloneqq \bigl|\{k \mid x_{k}\in X_{j}\}\bigr|$.
Note that $|\pi|$ and $|\pi|_{=j}$ can be $\omega$.
%$\rho_{\pi_{\leq m}}$ is the $\pi_{\leq m}$-th subrun of $\rho$.
For $m\in |\pi|$,
let $\rho_{\pi,m}$ be the subrun of $\rho$ that follows after
$(x_1,\sigma_1)\cdot(x_2,\sigma_2)\cdots(x_m,\sigma_m)$.
Moreover, we write $X_{\leq j}$ and $X_{>j}$ for
$\bigcup_{j'\leq j} X_{j'}$ and $\bigcup_{j'>j} X_{j'}$
respectively.
Recall that $l^{(j)}_{i}:\pow(\Run_{\X,X_{j+1}})\times \cdots \times\pow(\Run_{\X,X_{n}})
\rightarrow
%L
\pow(\Run_{\X,X_{i}})$ denotes the $j$-th interim solution 
(Def.~\ref{def:solOfEqSys}).

We first prove that:
for each $j\in[1,n]$,
sets $\seqby{l_{\i}\in\pow(\Run_{\X,X_{\i}})}{j+1}{n}$ of runs,
a priority $i\in[1,j]$, 
a run $\rho\in l^{(j)}_{i}(\seqby{l_{\i}}{j+1}{n})\cap \Run_{\mathcal{X},X_j}$, and
a (possibly-infinite) branch $\pi=(x_1,\sigma_1)\cdot(x_2,\sigma_2)\cdots$ of $\rho$,
we have either of the following conditions.
%
%The proof is by induction on priority $i$.
%The claim on induction is \emph{either of} the following conditions,
%for \emph{any of}:
%sets $\seqby{l_{\i}\in\pow(\Run_{\X,X_{\i}})}{i+1}{n}$ of runs with higher-prioritized roots;
%the priority $j\in[1,i]$ of the root;
%a run $\rho\in l^{(i)}_{j}(\seqby{l_{\i}}{i+1}{n})$ with a $j$-prioritized root; and
%a (possibly-infinite) branch $\pi=(x_1,\sigma_1)\cdot(x_2,\sigma_2)\cdots$ of $\rho$.
\begin{itemize}
  \item
    We have $x_m\in X_{\leq j}$ for each $m\in|\pi|$.
    Moreover,
    \begin{math}
      \max\bigl\{j'\,\big|\,
        |\pi|_{=j'}=\omega
      \bigr\}\,
    \end{math}
    is even when $|\pi|=\omega$. %; or
    %if $\pi$ is infinite-length.
  \item
    There exists $m\in|\pi|$ and such that $x_m\in X_{>j}$.
    Moreover,
    if we choose the minimum $m$ among such (i.e.\
    %to be the minimal, which satisfies
    $x_{m'}\in X_{\leq i}$ for every $m'<m$),
    then $\rho_{\pi,m}\in \bigcup_{j'>j}l_{j'}$.
\end{itemize}

We prove this by induction on $j$.
Note that there is no need of distinguishing the base case ($j=1$) from the step case.
%it is easy to take proper care of the occurrences of $i-1$ in the proof below.

%
%
%

\textbf{\underline{Case: $j$ is odd ($u_{j}$ is $\mu$-variable).}}
It is not hard to see, for each $k\in\omega$, that
\begin{equation}\label{eq:parityEqSysCoincidenceMuChain}
  \rho \in \bigl[\Diamond_{\X}\bigl(
    l^{(j-1)}_{1}(\place,\seqby{l_{\i}}{j+1}{n})
    \cup\cdots\cup
    l^{(j-1)}_{j-1}(\place,\seqby{l_{\i}}{j+1}{n})
    \cup\place\cup
    l_{j+1}
    \cup\cdots\cup
    l_{n}
  \bigr)
  \cap \Run_{\X,X_{j}}\bigr]^{k}(\emptyset)
\end{equation}
if and only if,
for every branch $\pi=(x_1,\sigma_1)\cdot(x_2,\sigma_2)\cdots$ of $\rho$,
either of the following conditions is satisfied.
%every branch $\pi$ in $\rho\in\Run_{\X,X_{i}}$ satisfies:
\begin{itemize}
  \item
    We have $x_m\in X_{\leq j}$ for each $m\in|\pi|$.
    Moreover,
    $|\pi|_{=j} \leq k$ \emph{and} %; \emph{and}
    \begin{math}
      \max\bigl\{j'\,\big|\,
        |\pi|_{=j'}=\omega
      \bigr\}\,
    \end{math}
    is even when $|\pi|=\omega$.
  \item
    There exists $m\in|\pi|$  such that $x_m\in X_{>j}$.
    Moreover,
    if we choose the minimum $m$ among such (i.e.\ %to be the minimal, which satisfies
    $x_{m'}\in X_{\leq j}$ for every $m'<m$),
    then $|(\sigma_1,x_1)\cdots(\sigma_{m-1},x_{m-1})|_{=j}\leq k$ \emph{and}
    $\rho_{\pi,m}\in \bigcup_{j'>j}l_{j'}$.
\end{itemize}
It is easy to see that
the interim solution $l^{(j)}_{i}(\seqby{l_{\i}}{j+1}{n})$ is obtained by taking
the supremum of~(\ref{eq:parityEqSysCoincidenceMuChain}), for $k\in\omega$.
Therefore, for $i=j$, the claim is discharged.
The proof for $i<j$ is easy.

\textbf{\underline{Case: $j$ is even ($u_{j}$ is a $\nu$-variable).}}
As in the former case, we can see that
\begin{equation}\label{eq:parityEqSysCoincidenceNuChain}
  \rho \in \bigl[\Diamond_{\X}\bigl(
    l^{(i-1)}_{1}(\place,\seqby{l_{\i}}{j+1}{n})
    \cup\cdots\cup
    l^{(i-1)}_{i-1}(\place,\seqby{l_{\i}}{j+1}{n})
    \cup\place\cup
    l_{j+1}
    \cup\cdots\cup
    l_{n}
  \bigr)
  \cap \Run_{\X,X_{j}}\bigr]^{k}(\Run_{\X,X_{j}})
\end{equation}
if and only if,
for every branch $\pi=(x_1,\sigma_1)(x_2,\sigma_2)\ldots$ of $\rho$,
either of the following conditions is satisfied.
\begin{itemize}
  \item
    We have $x_m\in X_{\leq j}$ for each $m\in|\pi|$.
    Moreover, $|\pi|_{=j}\geq k$; \emph{or}
    \begin{math}
      \max\bigl\{i'\,\big|\,
        |\pi|_{=j'}=\omega
      \bigr\}\,
    \end{math}
    is even when $|\pi|=\omega$. % ; or
    %if $\pi$ is infinite-length.
  \item
    There exists $m\in|\pi|$ and such that $x_m\in X_{>j}$.
    Moreover,
    if we choose 
    the minimum $m$ among such (i.e.\ %to be the minimal, which satisfies
    $x_{m'}\in X_{\leq j}$ for every $m'<m$),
    then $|(\sigma_1,x_1)\cdots(\sigma_m,x_m)|_{=j}\geq k$ \emph{or}
    $\rho_{\pi,m}\in \bigcup_{j'>j}l_{j'}$.
\end{itemize}
%In this case, 
It is easy to see that
the interim solution $l^{(j)}_{i}(\seqby{l_{\i}}{j+1}{n})$ is obtained by taking
the infimum of~(\ref{eq:parityEqSysCoincidenceNuChain}), for $k\in\omega$.
Therefore, for $i=j$, the claim is discharged and
the proof for $i<j$ is easy.
%we shall take the infimum of the chain for $k<\omega$.

Hence we can prove the claim for all $j\in[1,n]$.
Letting $j=n$, Lem.~\ref{lem:eqSysCharacterizationOfAcceptingRuns} follows.
\end{proof}

\subsection{Proof of
  Lem.~\ref{lem:eqSysCharacterizationOfAcceptingTrees}}
\begin{proof}
In what follows we shall work with the semantic domains
$L_{i}:= \prod_{x\in X_{i}}\pow(\Run_{\X})$ and
$L'_{i}:= \prod_{x\in X_{i}}\pow(\myTree_{\Sigma})$, which are
easily seen to be equivalent to the formulation in
 Lem.~\ref{lem:eqSysCharacterizationOfAcceptingTrees}. 
We write
\begin{displaymath}
 \textstyle
 \varphi_{i}:= 
\prod_{x\in X_{i}}\pow(\DelSt)
 \;\colon\;
% \prod_{x\in X_{i}}\pow(\Run_{\X})
 L_{i}
 \longrightarrow
% \prod_{x\in X_{i}}\pow(\myTree_{\Sigma})
L'_{i}
%\enspace,
\end{displaymath}
 for each $i\in [1,n]$. Here $\pow(\DelSt)\colon \pow(\Run_{\X})\to
 \pow(\myTree_{\Sigma})$ is defined by direct images.
 Furthermore we write $f_{i}, f'_{i}$ for the following functions
 (that occur on the right-hand sides of the relevant equational
 systems), for each $i\in [1,n]$.
 \begin{align*}
  &f_{i}\colon L_{1}\times\cdots\times L_{n}\longrightarrow L_{i}, 
  \quad
  f_{i}(u_{1},\dotsc,u_{n}):=
    \bigl(\,
  \Diamond_{\delta}\langle u_{1},\dotsc,u_{n}\rangle
  \,\bigr)
  \upharpoonright {X_{n}}\enspace,
\\
  &f'_{i}\colon L'_{1}\times\cdots\times L'_{n}\longrightarrow L'_{i}, 
  \quad
  f'_{i}(u'_{1},\dotsc,u'_{n}):=
    \bigl(\,
  \Diamond'_{\delta}\langle u'_{1},\dotsc,u'_{n}\rangle
  \,\bigr)
  \upharpoonright {X_{n}}\enspace.
 \end{align*}
 It is straightforward to see that the following diagram commutes, for each $i\in [1,n]$. 
    \begin{equation}\label{eq:homOfEqSysCompatibilityLALI}
      \vcenter{\xymatrix@R=.6em@C+2em{
        {L_{1}\times\cdots \times L_{n}}
            \ar[d]_{f_{i}}
            \ar[r]^{\varphi_{1}\times\cdots\times\varphi_{n}}
        &
        {L'_{1}\times\cdots \times L'_{n}}
            \ar[d]^{f'_{i}}
        \\
        {L_{i}}
            \ar[r]_{\varphi_{i}}
        &
        {L'_{i}}
     }}
     % \quad
     %  \vcenter{\xymatrix@R=.6em@C+2em{
     %    {L'_{1}\times\cdots \times L'_{n}}
     %        \ar[d]_{f'_{i}}
     %        \ar[r]^{\psi_{1}\times\cdots\times\psi_{n}}
     %    &
     %    {L_{1}\times\cdots \times L_{n}}
     %        \ar[d]^{f_{i}}
     %    \\
     %    {L'_{i}}
     %        \ar[r]_{\psi_{i}}
     %    &
     %    {L_{i}}
     % }}
    \end{equation}

In view of Lem.~\ref{lem:eqSysCharacterizationOfAcceptingRuns} it
 suffices to show that, on the solution $\seq{l^{\sol}_{\i}}{n}$ of 
the equational system $E$
 in~(\ref{eq:eqSysCharacterizationOfAcceptingRuns:eqSys}) and the
 solution
$\seq{{l'}^{\sol}_{\i}}{n}$ of 
the equational system $E'$
 in~(\ref{eq:eqSysCharacterizationOfAcceptingTrees:eqSys}), we have 
 $\varphi_{i}(l^{\sol}_{i})={l'}^{\sol}_{i}$ for each $i\in [1,n]$. 

Towards this end
we shall prove the following by induction on $i\in [1,n]$. 
\begin{quote}
% $(*)$ 
For each $l_{i+1}\in L_{i+1}, \dotsc,l_{n}\in L_{n}$:
 \begin{itemize}
  % \item The function ${f'}^{\ddagger}_{i}\bigl(\,\place,
  % 	\varphi_{i+1}(l_{i+1}),
  %       \dotsc,
  % 	\varphi_{n}(l_{n})
  % 	\bigr)\colon L'_{i}\to L'_{i}$ (in Def.~\ref{def:solOfEqSys},
  % 	for the equational system $E'$) is
  % 	well-defined.
  \item We have
	$\varphi_{i}\bigl(l^{(i)}_{i}(l_{i+1},\dotsc,l_{n})\bigr)=
	{l'}^{(i)}_{i}\bigl(\varphi_{i+1}(l_{i+1}),\dotsc,\varphi_{n}(l_{n})\bigr)
	% \eta_{i}\Bigl[\,{f'}^{\ddagger}_{i}\bigl(\place,
	% \varphi_{i+1}(l_{i+1}),
        % \dotsc,
	% \varphi_{n}(l_{n})
	% \bigr)\,\Bigr]
	$, where 
	% the latter is by definition equal to
	% \begin{math}
	%  	\eta_{i}\Bigl[\,{f'}^{\ddagger}_{i}\bigl(\place,
	% \varphi_{i+1}(l_{i+1}),
        % \dotsc,
	% \varphi_{n}(l_{n})
	% \bigr)\,\Bigr]
	% \end{math}.
%	Here
	$l^{(i)}_{i}\colon
	L_{i+1}\times\cdots\times
	L_{n}\to L_{i}$ is the $i$-th interim solution of $E$ for $u_{i}$
%	(that we have assumed to exist); 
	(Def.~\ref{def:solOfEqSys}); ${l'}^{(i)}_{i}$
	is the same for $E'$. 	

  \item On the other $i$-th interim solutions, too,  we have
	$
 \varphi_{j}\bigl(\,{l}^{(i)}_{j}(l_{i+1},\dotsc,l_{n})\,\bigr)
 =
{l'}^{(i)}_{j}\bigl(\varphi_{i+1}(l_{i+1}),\dotsc,\varphi_{n}(l_{n})\bigr)
    $,
        for each $j\in [1,i-1]$. 
 \end{itemize}
\end{quote}
By showing the above we will obtain $\varphi_{i}(l^{\sol}_{i})=l'^{\sol}_{i}$,
as a special case, for each $i\in[1,n]$.

The main technical difficulty lies in the first item; the second is easy.
% and its proof will  be presented  separately. 
Let us first assume that $i$ is odd, that is,  $\eta_{i}=\mu$. 
In this case, by the
 Cousot-Cousot construction of least fixed points (that is via
 transfinite induction), we have some ordinal $\alpha$ 
 where the increasing approximation sequence
 \begin{displaymath}
  \bot
\le 
\Bigl(f^{\ddagger}_{i}\bigl(\,
\place,\, l_{i+1},\dotsc,l_{n}
\bigr)\Bigr)(\bot)
\le 
\Bigl(f^{\ddagger}_{i}\bigl(\,
\place,\, l_{i+1},\dotsc,l_{n}
\bigr)\Bigr)^{2}(\bot)
 \le\cdots
 \end{displaymath}
 stabilizes, yielding
\begin{align*}
 l^{(i)}_{i}(l_{i+1},\dotsc,l_{n})
 \;&=\;
 \mu\bigl[\,f^{\ddagger}_{i}(\,\place, l_{i+1},\dotsc,l_{n})\,\bigr]
 \qquad\text{by def.\ of $l^{(i)}_{i}$}
\\ 
\;&=\;
 \bigl(f^{\ddagger}_{i}(\,\place, l_{i+1},\dotsc,l_{n})\bigr)^{\alpha}(\bot)
\end{align*}
for $E$. For $E'$ the situation is similar, and 
 ${l'}^{(i)}_{i}\bigl(\varphi_{i+1}(l_{i+1}),\dotsc,\varphi_{n}(l_{n})\bigr)$
 is given as a suitable limit of a (transfinite) increasing sequence. 

Let us note the following. 
\begin{equation}\label{eq:201606212354}
 \begin{aligned}
 & (\varphi_{i}\co f^{\ddagger}_{i})(\,\place, l_{i+1},\dotsc,l_{n})
 \\
 &=
  (\varphi_{i}\co f_{i})\left(
 \begin{array}{l}
  l^{(i-1)}_{1}(\,\place, l_{i+1},\dotsc,l_{n}), \,
 \\
\dotsc,\,
  l^{(i-1)}_{i-1}(\,\place, l_{i+1},\dotsc,l_{n}),
 \\
  \place,\, l_{i+1},\dotsc,l_{n}
 \end{array}
 \right)
 \qquad\text{by def.\ of $f^{\ddagger}_{i}$}
 \\
 &=
   f'_{i}\left(
 \begin{array}{l}
  \varphi_{1}\bigl(l^{(i-1)}_{1}(\,\place, l_{i+1},\dotsc,l_{n})\bigr), \,
 \\\dotsc,\,
    \varphi_{i-1}\bigl(l^{(i-1)}_{i-1}(\,\place, l_{i+1},\dotsc,l_{n})\bigr),
 \\
  \varphi_{i}(\place),\, \varphi_{i+1}(l_{i+1}),\dotsc,\varphi_{n}(l_{n})
 \end{array}
 \right)
 \qquad\text{by~(\ref{eq:homOfEqSysCompatibilityLALI})}
 \\
 &=
   f'_{i}\left(
 \begin{array}{l}
  {l'}^{(i-1)}_{1}\bigl(\,\varphi_{i}(\place), \varphi_{i+1}(l_{i+1}),\dotsc,\varphi_{n}(l_{n})\,\bigr), \,
 \\\dotsc,\,
  {l'}^{(i-1)}_{i-1}\bigl(\,\varphi_{i}(\place), \varphi_{i+1}(l_{i+1}),\dotsc,\varphi_{n}(l_{n})\,\bigr),
 \\
  \varphi_{i}(\place),\, \varphi_{i+1}(l_{i+1}),\dotsc,\varphi_{n}(l_{n})
 \end{array}
 \right)
 \qquad\text{by ind.\ hyp.}
 \\
 &=
 {f'}^{\ddagger}_{i}\bigl(
 \varphi_{i}(\place),\, \varphi_{i+1}(l_{i+1}),\dotsc,\varphi_{n}(l_{n})
 \bigr)
 \qquad\text{by def.\ of ${f'}^{\ddagger}_{i}$.}
 \end{aligned}
\end{equation}
We shall use this in showing that, for each ordinal $\beta$, 
we have the following. 
Here $\bot$ is the least element of $L_{i}$. 
\begin{displaymath}
\Bigl(\, \varphi_{i}\co \bigl(
f^{\ddagger}_{i}(\,\place, l_{i+1},\dotsc,l_{n})
\bigr)^{\beta}
\,\Bigr)(\bot)\;=\;
\Bigl(\,\Bigl({f'}^{\ddagger}_{i}\bigl(\,
\place,\, \varphi_{i+1}(l_{i+1}),\dotsc,\varphi_{n}(l_{n})
\bigr)
\Bigr)^{\beta}\co \varphi_{i}
\,\Bigr)(\bot)
\quad\in L'_{i}\enspace.
\end{displaymath}
Indeed: the base case ($\beta=0$) is obvious;  the step case follows
 from~(\ref{eq:201606212354}); and for the limit case ($\beta$ is a
 limit ordinal), we use the fact that $\varphi_{i}=\prod_{x\in
 X_{i}}\pow(\DelSt)$---defined by direct images---preserves supremums
 (i.e.\ unions). 
 Together with 
 the fact that $\varphi_{i}$ preserves least elements, we see that
 $\varphi_{i}$ carries the Cousot-Cousot sequence in $L_{i}$
 (for computing $l^{(i)}_{i}(l_{i+1},\dotsc,l_{n})$)
 to the one in $L'_{i}$ 
 (for computing
 ${l'}^{(i)}_{i}\bigl(\varphi_{i+1}(l_{i+1}),\dotsc,\varphi_{n+1}(l_{n})\bigr)$).
 This proves 
	$\varphi_{i}\bigl(l^{(i)}_{i}(l_{i+1},\dotsc,l_{n})\bigr)=
	{l'}^{(i)}_{i}\bigl(\varphi_{i+1}(l_{i+1}),\dotsc,\varphi_{n}(l_{n})\bigr)$. 
 
Let us now assume that $i$ is even, that is,  $\eta_{i}=\nu$.  We shall
 again prove the claim by scrutinizing the Cousot-Cousot sequences for 
	$l^{(i)}_{i}(l_{i+1},\dotsc,l_{n})$ and 
	${l'}^{(i)}_{i}\bigl(\varphi_{i+1}(l_{i+1}),\dotsc,\varphi_{n}(l_{n})\bigr)$. Writing 
\begin{align*}
 &\Phi
 \;:=\;
 f^{\ddagger}_{i}(\place, l_{i+1},\dotsc,l_{n})
 \;\stackrel{\text{by def.}}{=}\;
   f_{i}\left(
 \begin{array}{l}
  l^{(i-1)}_{1}(\,\place, l_{i+1},\dotsc,l_{n}), \,
 \\
\dotsc,\,
  l^{(i-1)}_{i-1}(\,\place, l_{i+1},\dotsc,l_{n}),
 \\
  \place,\, l_{i+1},\dotsc,l_{n}
 \end{array}
 \right)
 \quad\text{and}
 \\
&\Phi'
 \;:=\;
 {f'}^{\ddagger}_{i}(\place, \seqby{\varphi_{\i}(l_{\i})}{i+1}{n})
\\
&\qquad\qquad
 \;\stackrel{\text{by def.}}{=}\;
   f'_{i}\left(
 \begin{array}{l}
  {l'}^{(i-1)}_{1}(\,\place, \seqby{\varphi_{\i}(l_{\i})}{i+1}{n}), \,
 \\
\dotsc,\,
  {l'}^{(i-1)}_{i-1}(\,\place, \seqby{\varphi_{\i}(l_{\i})}{i+1}{n}),
 \\
  \place,\, \seqby{\varphi_{\i}(l_{\i})}{i+1}{n}
 \end{array}
 \right)\enspace,
\end{align*}
the relevant Cousot-Cousot sequences are as follows. 
\begin{align}\label{eq:201606222311}
 \top
 \ge
 \Phi(\top)
 \ge\cdots\ge
 \Phi^{\alpha}(\top)
 \ge\cdots
 \quad\text{in $L_{i}$, \quad and }
 \quad
 \top
 \ge
 \Phi'(\top)
 \ge\cdots\ge
 {\Phi'}^{\alpha}(\top)
 \ge\cdots
 \quad\text{in $L'_{i}$.}
\end{align}
Unlike the previous case where $\eta_{i}=\mu$, it is not the case that the first
 sequence is carried \emph{exactly} to the second by $\varphi_{i}$. Instead 
we shall show the following two claims.
\begin{enumerate}
 \item\label{item:06222325} 
      For each ordinal $\alpha$ we have
      $\varphi_{i}\bigl(\Phi^{\alpha}(\top)\bigr)\le 
      {\Phi'}^{\alpha}(\top)$.
 \item\label{item:06222326} 
%       For each ordinal $\alpha$
% % , there exists an ordinal $\beta$ such that
% %       $\varphi_{i}\bigl(\Phi^{\alpha}(\top)\bigr)\ge
% %       {\Phi'}^{\beta}(\top)$.       
%       we have 
%       $\varphi_{i}\bigl(\Phi^{\alpha}(\top)\bigr)\ge
%        \nu \Phi'$.       
  We have $R\in L_{i}$ such that: $R$ is a $\Phi$-postfixed point (i.e.\
      $R\le \Phi(R)$); and 
  $\varphi_{i}(R)=\nu\Phi'$. 
\end{enumerate}
Showing these items~\ref{item:06222325}--\ref{item:06222326} proves
the claim (namely $\varphi_{i}(\nu \Phi)=\nu \Phi'$). Indeed: taking $\alpha'_{0}$ such that
 $\nu \Phi'={\Phi'}^{\alpha'_{0}}(\top)$, we have 
 \begin{displaymath}
  \nu \Phi'={\Phi'}^{\alpha'_{0}}(\top)
%  \stackrel{\text{by~\ref{item:06222325}.}}{\ge}
  \ge
  \varphi_{i}(\Phi^{\alpha'_{0}})
%  \stackrel{\text{$\varphi_{i}$ is monotone}}{\ge}
  \ge
  \varphi_{i}(\nu \Phi)
  \quad\text{where we used monotonicity of $\varphi_{i}$;}
 \end{displaymath}
 conversely, for $R$ in the item~\ref{item:06222326}.\ we have 
 $R\le \nu \Phi$---because $\nu \Phi$ is the greatest $\Phi$-postfixed
 point (the Knaster-Tarski theorem)---hence
 \begin{displaymath}
  \nu \Phi'=\varphi_{i}(R)\le \varphi_{i}(\nu\Phi)\enspace.
 \end{displaymath}

The item~\ref{item:06222325}.\ is shown by (transfinite) induction on
 $\alpha$. The base case is obvious. For the step case, 
\begin{equation}\label{eq:06222356}
  \begin{aligned}
 & {\Phi'}^{\alpha+1}(\top)
 \\
 &\ge 
  \Phi'\bigl(\varphi_{i}(\Phi^{\alpha}(\top))\bigr)
 \quad\text{by ind.\ hyp. (for $\alpha$), and that $\Phi'$ is monotone}
 \\
 &=
   f'_{i}\left(
 \begin{array}{l}
  {l'}^{(i-1)}_{1}(\,\varphi_{i}(\Phi^{\alpha}(\top)), \seqby{\varphi_{\i}(l_{\i})}{i+1}{n}), \,
 \\
\dotsc,\,
  {l'}^{(i-1)}_{i-1}(\,\varphi_{i}(\Phi^{\alpha}(\top)), \seqby{\varphi_{\i}(l_{\i})}{i+1}{n}),
 \\
  \varphi_{i}(\Phi^{\alpha}(\top)),\,
   \seqby{\varphi_{\i}(l_{\i})}{i+1}{n}
 \end{array}
 \right)
  \quad\text{by def.\ of $\Phi'$}
 \\
 &=
   f'_{i}\left(
 \begin{array}{l}
  (\varphi_{1}\co l^{(i-1)}_{1})(\,\Phi^{\alpha}(\top), \seqby{l_{\i}}{i+1}{n}), \,
 \\
\dotsc,\,
  (\varphi_{i-1}\co l^{(i-1)}_{i-1})(\,\Phi^{\alpha}(\top), \seqby{l_{\i}}{i+1}{n}), \,
 \\
  \varphi_{i}(\Phi^{\alpha}(\top)),\,
   \seqby{\varphi_{\i}(l_{\i})}{i+1}{n}
 \end{array}
 \right)
 \quad\text{by ind.\ hyp. (for $i-1$)}
 \\
 &=
 (\varphi_{i}\co f_{i})
  \left(
 \begin{array}{l}
   l^{(i-1)}_{1}(\,\Phi^{\alpha}(\top), \seqby{l_{\i}}{i+1}{n}), \,
 \\
\dotsc,\,
   l^{(i-1)}_{i-1}(\,\Phi^{\alpha}(\top), \seqby{l_{\i}}{i+1}{n}), \,
 \\
  \Phi^{\alpha}(\top),\,
   \seqby{l_{\i}}{i+1}{n}
 \end{array}
 \right)
 \quad\text{by~(\ref{eq:homOfEqSysCompatibilityLALI})}
 \\
 &=
 \varphi_{i}\bigl(\Phi(\Phi^{\alpha}(\top))\bigr)
 =
 \varphi_{i}\bigl(\Phi^{\alpha+1}(\top)\bigr)
   \quad\text{by def.\ of $\Phi$.}
\end{aligned}
\end{equation}
 For the limit case, we have 
 \begin{displaymath}
  \varphi_{i}\bigl(\Phi^{\alpha}(\top)\bigr)
  =
  \varphi_{i}\bigl(\bigwedge_{\alpha'<\alpha}\Phi^{\alpha'}(\top)\bigr)
\le 
  \bigwedge_{\alpha'<\alpha}  \varphi_{i}\bigl(\Phi^{\alpha'}(\top)\bigr)
 \le
  \bigwedge_{\alpha'<\alpha}  {\Phi'}^{\alpha'}(\top)
 =
      {\Phi'}^{\alpha}(\top)\enspace,
 \end{displaymath} 
 where the first inequality is due to monotone of $\varphi_{i}$ and the
 second is by the induction hypothesis (on $\alpha'$). This proves the
 item~\ref{item:06222325}.

 For the item~\ref{item:06222326}.\ we first observe the fixed-point
 property of $\nu\Phi'$, expanding the definition of $\Phi'$ and
 furthermore
that of $f'$:
\begin{equation}\label{eq:06231236}
 (\nu\Phi')_{x} = 
\left\{
 \bigl(\sigma,(\seqby{\tau_{\i}}{1}{|\sigma|})\bigr)
\left|
\begin{array}{l}
\exists \seqby{x_{\i}}{1}{|\sigma|}. 
\\
\; (\sigma,(\seqby{x_{\i}}{1}{|\sigma|}))\in \delta(x),
\quad\text{and}\quad
 \forall k\in [1,|\sigma|]. 
\\
\left(
\begin{array}{l}
  x_{k}\in X_{1} 
 \;\Rightarrow\;
 \tau_{k}\in 
 \Bigl({l'}^{(i-1)}_{1}\bigl(\, \nu\Phi',
 \seqby{\varphi_{\i}(l_{\i})}{i+1}{n}\bigr)\Bigr)_{x_{k}},
 \\
  \dotsc,
 \\
  x_{k}\in X_{i-1} 
 \;\Rightarrow\;
 \tau_{k}\in 
 \Bigl({l'}^{(i-1)}_{i-1}\bigl(\, \nu\Phi',
 \seqby{\varphi_{\i}(l_{\i})}{i+1}{n}\bigr)\Bigr)_{x_{k}},
 \\
  x_{k}\in X_{i}
 \;\Rightarrow\;
  \tau_{k}\in (\nu\Phi')_{x_{k}},
 \\
  x_{k}\in X_{i+1}
 \;\Rightarrow\;
  \tau_{k}\in (\varphi_{i+1}(l_{i+1}))_{x_{k}},
 \\
  \dotsc,
 \\
  x_{k}\in X_{n}
 \;\Rightarrow\;
  \tau_{k}\in (\varphi_{n}(l_{n}))_{x_{k}}.
\end{array}
\right)
\end{array}
\right.
\right\}
\end{equation}
for each $x\in X_{i}$. 
It is then not hard to see that, for each $\Sigma$-tree
$\tau$ that belongs to $(\nu \Phi')_{x}$, we can find at least
one run $\rho$ of $\X$ so that $\DelSt(\rho)=\tau$. This fact is proved 
by decorating each node of $\tau$ with an $X$-label, coinductively from
 top to bottom, starting with $x$. Concretely, once an $X$-label $x'$ is assigned to
 a certain node, we operate as follows.
\begin{itemize}
 \item If $x'\in X_{k}$ with $k\in [i+1,n]$, then the subtree $\tau'$
       starting at the
       current node belongs to the set
       $(\varphi_{k}(l_{k}))_{x'}$. Recalling that
       $\varphi_{k}=\pow(\DelSt)$, we can find a run $\rho'\in l_{k}$
       such that $\DelSt(\rho')=\tau'$; we decorate $\tau'$ according to
       $\rho'$. 
 \item If $x'\in X_{i}$ then the subtree $\tau'$
       starting at the
       current node belongs to $(\nu \Phi')_{x'}$. We invoke the
       fixed-point property~(\ref{eq:06231236}) to find the $X$-labels 
       $\seqby{x_{\i}}{1}{|\sigma|}$ for the children of the current node.
 \item If $x'\in X_{k}$ with $k\in [1,i-1]$,  we note that 
     the set 
     \begin{math}
       \Bigl({l'}^{(i-1)}_{k}\bigl(\, \nu\Phi',
 \seqby{\varphi_{\i}(l_{\i})}{i+1}{n}\bigr)\Bigr)_{x'}
     \end{math}---to which the  subtree $\tau'$
       starting at the
       current node should belong to---consists of those trees $\tau$
       with the following property: $\tau$ has a prefix $\tau_{0}$ that is the
       image under $\DelSt$ of a prefix $\rho_{0}$ of some run of $\X$ starting
       from $x'$; $\rho_{0}$ has $X$-labels from $X_{i}\cup
       X_{i+1}\cup\cdots \cup X_{n}$ only at those nodes where
       $\tau_{0}$ ends but $\tau$ continues; and, at each such node
       $x''$,
\begin{itemize}
 \item 
        $x''\in X_{i}$ implies that the subtree of $\tau$ starting there
       belongs to $(\nu \Phi')_{x''}$, and
 \item 
       $x''\in X_{j}$ (for $j\in [i+1,n]$) implies that 
the subtree of $\tau$ starting there
       belongs to $(\varphi_{j}(l_{j}))_{x''}$.
\end{itemize}       
This fact is shown in the current induction on $i$. We can then 
decorate the prefix $\tau'_{0}$ of $\tau'$ according to $\rho'_{0}$ (in
       the above notations); once we hit $X$-labels from 
$X_{i}\cup
       X_{i+1}\cup\cdots \cup X_{n}$  we continue according to the above
       other cases. 
\end{itemize}
For each $\tau\in (\nu\Phi')_{x}$ we collect its decorations $\rho$; and
 we
 let $R\in L_{i}=\prod_{x\in X_{i}}\pow(\Run_{\X})$ defined by its
 closure under subtrees. It is then obvious that $R\le \Phi(R)$ (since
 $R$ is closed under subtrees) and
 $\varphi_{i}(R)=\nu \Phi'$ (since for each $\tau\in (\nu\Phi')_{x}$ we
 included its decoration). This proves the item~\ref{item:06222326}, and 
 proves the claim.
\end{proof}

\begin{myremark}
\label{rem:cousotcousotseqDoNotMatchStepByStep}
The sequences~(\ref{eq:201606222311}) do not match step-by-step, already
 in the following simple example. Assume that $F=\{*\}\times(\place)$,
 every edge below is labeled with $*$, and every state is accepting.
\begin{displaymath}
  \includegraphics{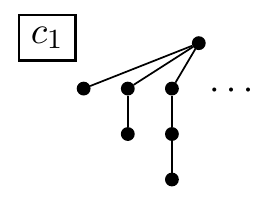}
\end{displaymath}
 Let the top node denoted by $x$. 
 Then after $\omega$ steps in the first Cousot-Cousot sequence  every potential run from $x$ is eliminated
 (one with length $n$ is eliminated after $n$ steps). However in the
 second Cousot-Cousot sequence, the word $*^{\omega}=**\cdots$ is eliminated only
 after $\omega +1$ steps: $*^{\omega}\in
 \bigcup_{n<\omega}{\Phi'}^{n}(\top)$ because, for each $n$, $x$ has a
 run of length $n$. 
\end{myremark}

\subsection{Proof of Lem.~\ref{lem:inverseExists}}
The following fact, which gives an explicit construction of
the final coalgebra $\zeta\colon Z\kto FX$, is standard.

\begin{mysublemma}[\cite{Schubert09tcf}]
 \label{sublem:finalCoalgInMeas}
 Let $F\colon \Meas\to\Meas$ be a (standard Borel) polynomial functor. 
 Let $Z$ be a limit of its final sequence (up to $\omega$)---the
  measurable structure of $Z$ is the weakest one such that all
  projections $\pi_{i}$ are measurable. In this case the functor  $F$
  preserves the limit $Z$ and we have the following mediating isomorphism $\zeta$. 
 \begin{equation}\label{eq:finalSeqInMeas}
  \vcenter{\xymatrix@R=1em{
  &&&{Z\mathrlap{\;\text{(limit)}}}
  \ar@/_1pc/[llld]_(.8){\pi_{0}}
  \ar@/_.7pc/[lld]_(.8){\pi_{1}}
  \ar[ld]_{\pi_{2}}
  \ar@{-->}@/^2pc/[dd]^{\zeta}_{\cong}
  \\
  {1}
  &
   {F1}
   \ar[l]_{\bang }
  &
   {F^{2}1}
   \ar[l]_{F\bang }
   {\cdots}
   \ar[l]
  \\
  &&&{FZ\mathrlap{\;\text{(limit)}}}
  \ar@/^1pc/[lllu]^(.8){\bang }
  \ar@/^.7pc/[llu]^(.8){F\pi_{0}}
  \ar[lu]^{F\pi_{1}}
}}
 \qquad\text{(in $\Meas$)}
 \end{equation}
 By a standard argument like in~\cite{AdamekK79lfp},
  $\zeta\colon Z\to FZ$ is a final coalgebra in $\Meas$. 
  \qed
\end{mysublemma}
We also use the fact that the Kleisli inclusion functor $J$ lifts the limit to 2-limit in $\Kl(\giry)$.
 \begin{mysublemma}[\cite{UrabeH15cit}]\label{lem:limitInKlGiry}
  The Kleisli inclusion functor $J\colon \Meas\to \Kl(\giry)$ for the
  sub-Giry monad $\giry$ preserves the limits
  in~(\ref{eq:finalSeqInMeas}). This yields, in particular, the 
  following limit.
 \begin{equation}\label{eq:finalSeqLiftedToKlGiry}
  \vcenter{\xymatrix@R=1em{
  &&&{Z\mathrlap{\;\text{(limit)}}}
  \kar@/_1pc/[llld]_(.8){J\pi_{0}}
  \kar@/_.7pc/[lld]_(.8){J\pi_{1}}
  \kar[ld]_{J\pi_{2}}
%  \kar@{-->}@/^2pc/[dd]^{\zeta}_{\cong}
  \\
  {1}
  &
   {F1}
   \kar[l]_{J\bang }
  &
   {F^{2}1}
   \kar[l]_{JF\bang }
  &
   {\cdots}
   \kar[l]
  % \\
  % &&&{FZ\mathrlap{\;\text{(limit)}}}
  % \ar@/^1pc/[lllu]^(.8){\bang }
  % \ar@/^.7pc/[llu]^(.8){F\pi_{0}}
  % \ar[lu]^{F\pi_{1}}
}}
 \qquad\text{(in $\Kl(\giry)$)}
 \end{equation}  
 Moreover $Z$ here is in fact a 2-limit: if two cones 
 $
 \bigl(\,\gamma_{k}\colon
  X\kto F^{k}1
  \,\bigr)_{k\in\omega}$ and 
 $
 \bigl(\,\gamma'_{k}\colon
  X\kto F^{k}1
  \,\bigr)_{k\in\omega}$ 
  satisfy $\gamma_{k}\le \gamma'_{k}$ for each $k\in\omega$, then 
 the mediating arrows
  $\tuple{\gamma_{k}}_{k\in\omega},\tuple{\gamma'_{k}}_{k\in\omega}\colon X\kto Z$
  satisfy $\tuple{\gamma_{k}}_{k\in\omega}\le\tuple{\gamma'_{k}}_{k\in\omega}$.
 \end{mysublemma}
 \begin{proof}
  The claim follows from the result
  in~\cite{Schubert09tcf}
  that: the sub-Giry monad $\giry$ preserves limits over an
  $\omega^{\text{op}}$-sequence, provided that the latter consists of standard Borel spaces 
and surjective measurable functions. This is indeed the setting
  in~(\ref{eq:finalSeqInMeas}), and the result yields the following
  limit. 
 \begin{equation}\label{eq:finalSeqInMeasCarriedBySubGiry}
  \vcenter{\xymatrix@R=1em{
  &&&{\giry Z\mathrlap{\;\text{(limit)}}}
  \ar@/_1pc/[llld]_(.8){\giry\pi_{0}}
  \ar@/_.7pc/[lld]_(.8){\giry\pi_{1}}
  \ar[ld]_{\giry\pi_{2}}
%  \ar@{-->}@/^2pc/[dd]^{\zeta}_{\cong}
  \\
  {\giry 1}
  &
   {\giry F1}
   \ar[l]_{\giry \bang }
  &
   {\giry F^{2}1}
   \ar[l]_{\giry F\bang }
  &
   {\cdots}
   \ar[l]
  % \\
  % &&&{FZ\mathrlap{\;\text{(limit)}}}
  % \ar@/^1pc/[lllu]^(.8){\bang }
  % \ar@/^.7pc/[llu]^(.8){F\pi_{0}}
  % \ar[lu]^{F\pi_{1}}
}}
 \qquad\qquad\text{(in $\Meas$)}
 \end{equation}
 It is straightforward to see that: cones over the sequence
  in~(\ref{eq:finalSeqLiftedToKlGiry}) are precisely those over
  the sequence in~(\ref{eq:finalSeqInMeasCarriedBySubGiry}); and the
  correspondence carries over to mediating arrows. Here the following 
 easy observation plays a crucial role: for any $f\colon Y\to X$,
  $g\colon Z\to \giry X$ and $h\colon Z\to \giry Y$, 
 \begin{equation}\label{eq:PreComposingJf}
  \vcenter{\xymatrix@R=1em{
  {X}
  &
  {Y}
    \kar[l]_{Jf}
  \\
  &
   {Z}
    \kar[lu]^{g}
    \kar[u]_{h}
}}
  \;\text{in $\Kl(\giry)$}
  \quad\text{if and only if}\quad
  \vcenter{\xymatrix@R=1em{
  {\giry X}
  &
  {\giry Y}
    \ar[l]_{\giry f}
  \\
  &
   {Z}
    \ar[lu]^{g}
    \ar[u]_{h}
}}
   \;\text{in $\Meas$.}
\end{equation}
The last ``monotonicity'' condition is easy, too, exploiting the fact
that the measurable structure of $Z$ is the weakest one such that all
projections $\pi_{i}$ are measurable.
\end{proof}

Now we shall prove Lem.~\ref{lem:inverseExists}.

\noindent
\begin{proof}
We first define $\Delta^{g_{B}}\colon\mathfrak{H}^{g_{B}}\to \mathfrak{G}^{g_{B}}$.
Let $(h_{A}\colon X\kto 1)\in \mathfrak{H}^{g_{B}}$.

For each $k\in\omega$, we define an arrow
$\gamma^{A}_{k}\colon X_{A}\kto\oF^{k}1$ by induction on $k$ as follows:
\begin{align*}
  \gamma^{A}_{0}
  &\;\coloneqq\;
  h_{A}
  \\
  \gamma^{A}_{k+1}
  &\;\coloneqq\;
  \overline{F}[\gamma^{A}_{k}, J\pi_{k}\odot g_{B}]\odot c_{A}
  \enskip.
\end{align*}
Here $c_{[A,B]} = c \odot \kappa_{[A,B]}$; and
$\kappa_{[A,B]}\colon X_{[A,B]} \kto X_{A}+X_{B}$
denotes the canonical coprojection.

We show that
$(X_A,(\gamma^{A}_{k}\colon X_{A}\kto\oF^{k}1)_{k\in\omega})$ 
is a cone over the sequence
$1\overset{J!_{F1}}{\longkto}
\oF 1\overset{\oF J!_{F1}}{\longkto}
\oF^{2}1\overset{\oF^{2}J!_{F1}}{\longkto}\cdots$\,.
To this end,
%By induction on $a\in\omega$, 
we show that for each $k\in\omega$,
$\oF^{k}J!_{F1}\odot \gamma^{A}_{k+1}=\gamma^{A}_{k}$ by induction on $k$.
If $k=0$, then:
\begin{align*}
&\overline{F}^kJ!_{F1}\odot \gamma^A_{k+1} \\
&=J!_{F1}\odot 
\overline{F}[\gamma^{A}_{0},J\pi_{0}\odot g_{B}]\odot c_A & (\text{by definition})\\
&=J!_{F1}\odot
\overline{F}[h_{A},J\pi_{0}\odot g_{B}]\odot c_A & (\text{by definition})\\
&=J!_{F1}\odot 
\overline{F}[h_{A},J!_Z\odot g_{B}]\odot c_A & (\text{$\pi_0=!_Z$})\\
&=J!_{F1}\odot 
\overline{F}[h_{A},\Gamma_{B}(g_{B})]\odot c_A & (\text{by definition})\\
&=
h_A & (\text{$h_{A}\in\mathfrak{H}^{g_{B}}$})\\
&=
\gamma^A_{0} & (\text{by definition})
\,.
\end{align*}
For $k>0$, we have:
%The proof for step cases where $j>0$ are similar,
%except that instead of $h_{A}\in \mathfrak{H}^{g_{B}}$, 
%we use the induction hypothesis.
\begin{align*}
&\overline{F}^kJ!_{F1}\odot \gamma^A_{k+1} \\
&=\overline{F}^kJ!_{F1}\odot
\overline{F}[\gamma^A_{k},J\pi_{k}\odot g_{B}]\odot c_A & (\text{by definition})\\
&=\overline{F}[\:
  \overline{F}^{k-1}J!_{F1}\odot\gamma^A_{k},\:
  \overline{F}^{k-1}J!_{F1}\odot J\pi_{k}\odot g_{B}
\:]\odot c_A \\
&=\overline{F}[\:
  \gamma^A_{k-1},\:
  \overline{F}^{k-1}J!_{F1}\odot J\pi_{k}\odot g_{B}
\:]\odot c_A & (\text{by induction hypothesis})\\
&=\overline{F}[\:
  \gamma^A_{k-1},\:
  J\pi_{k-1}\odot g_{B}
\:]\odot c_A & (\text{$(Z,(\pi_j)_j)$ is a cone})\\
&=\gamma^A_k & (\text{by definition})\,.
\end{align*}
Hence
$(X_A,(\gamma^{A}_k:X_A\kto\overline{F}^k1)_{k\in\omega})$ 
is a cone  over
the sequence
$1\overset{J!_{F1}}{\longkto}
\oF 1\overset{\oF J!_{F1}}{\longkto}
\oF^{2}1\overset{\oF^{2}J!_{F1}}{\longkto}\cdots$\,,
and this implies that there uniquely exists a mediating arrow
$h^{\dagger}_A\colon X_A\kto Z$.

We show that $h^{\dagger}_{A}$ belongs to $\mathfrak{G}^{g_{B}}$,
that is,
$h^{\dagger}_A=J\zeta^{-1}\odot \overline{F}[h^{\dagger}_{A},g_{B}]\odot c_A$.
To this end, by the definition of $h^{\dagger}_A$,
it suffices to show that for each $k\in\omega$
we have
\begin{displaymath}
J\pi_k\odot (J\zeta^{-1}\odot \overline{F}[h^{\dagger}_A,g_{B}]\odot c_A)
=\gamma^A_k\,.
\end{displaymath}
%$J\zeta^{-1}\odot \overline{F}[l^{\dagger}_1,\ldots,l^{\dagger}_i,l_{i+1},\ldots,l_n]\odot c_j$
%is a mediating arrow from a cone $(X,(\gamma^j_a)_a)$ to $(Z,(J\pi_a)_a)$.
%We prove it by induction on $a$.
If $k=0$, then we have:
\begin{align*}
&J\pi_k\odot (J\zeta^{-1}\odot\overline{F}[h^{\dagger}_A,g_B]\odot c_A) \\
&=J!_Z\odot J\zeta^{-1}\odot\overline{F}[h^{\dagger}_A,g_B]\odot c_A
& (\text{$\pi_0=!_Z$}) \\
&=J!_{F1}\odot JF!_Z\odot\overline{F}[h^{\dagger}_A,g_{B}]\odot c_A
&(\text{$!_Z\circ \zeta^{-1}=!_{FZ}=!_{F1}\circ F!_Z$}) \\
&=J!_{F1}\odot\overline{F}[J!_{Z}\odot h^{\dagger}_A,J!_{Z}\odot g_B]\odot c_A
& (\text{$JF=\overline{F}J$}) \\
&=J!_{F1}\odot\overline{F}[J\pi_0\odot h^{\dagger}_A,J!_{Z}\odot g_B]\odot c_A
& (\text{$\pi_0=!_Z$}) \\
&=J!_{F1}\odot\overline{F}[\gamma^A_k,J!_Z\odot g_B]\odot c_A  & 
(\text{each $l^\dagger_j$ is a mediating arrow})\\
&=J!_{F1}\odot\overline{F}[h_{A},\Gamma(g_{B})]\odot c_A
& (\text{by definition})\\
&=h_{A}
& (\text{$h_{A}\in\mathfrak{H}^{g_{B}}$})\\
&=\gamma^A_k & (\text{by definition})\,.
\end{align*}
If $k>0$, then we have:
\begin{align*}
&J\pi_k\odot(J\zeta^{-1}\odot\overline{F}[h^{\dagger}_A,g_B]\odot c_A) \\
&=JF\pi_{k-1}\odot\overline{F}[h^{\dagger}_A,g_{B}]\odot c_A
& (\text{$\zeta$ is a mediating arrow}) \\
&=\overline{F}[J\pi_{k-1}\odot h^{\dagger}_A,J\pi_{a-1}\odot g_B]\odot c_A
& (\text{$JF=\overline{F}J$})\\
&=\overline{F}[\gamma^A_{k-1},J\pi_{k-1}\odot g_{B}]\odot c_A
& (\text{$g^\dagger_A$ is a mediating arrow})\\
&=\gamma^A_k
& (\text{by definition})
\,.
\end{align*}

We shall define
$\Delta^{g_B}\colon\mathfrak{H}^{g_{B}}\to \mathfrak{G}^{g_B}$ by
$\Delta^{g_B}(h_A)\coloneqq h^{\dagger}_A$\,; and
let us show the monotonicity of $\Delta^{g_B}$ here.
Assume that $h_A\sqsubseteq h'_A\colon X_A\kto 1$.
Let $\bigl(X_A,(\gamma^{A}_k\colon X_A\kto\overline{F}^k1)_{k\in\omega}\bigr)$ and
$\bigl(X_A,(\gamma'^{A}_k\colon X_A\kto\overline{F}^k1)_{k\in\omega}\bigr)$ be cones 
%and $l^{1,\dagger}_j:X_j\kto Z$ and $l^{2,\dagger}_j:X_j\kto Z$ be mediating arrows 
that are induced by  $h_A$ and $h'_A$ as above,
respectively.
Then by induction on $k\in\omega$, we can show that
$\gamma^A_k\sqsubseteq\gamma'^A_k$ for each $k\in\omega$.
As $\bigl(Z,(J\pi_k\colon Z\kto \overline{F}^k1)_{k\in\omega}\bigr)$ is a $2$-limit,
it implies that the mediating arrow induced by
$\bigl(X_A,(\gamma^{A}_k\colon X_A\kto \overline{F}^k1)_{k\in\omega}\bigr)$
is less than or equal to the one induced by
$\bigl(X_A,(\gamma^{A}_k\colon X_A\kto \overline{F}^k1)_{k\in\omega}\bigr)$---which means
$\Delta^{g_{B}}(h_{A}) \sqsubseteq \Delta^{g_{B}}(h'_{A})$, by definition.
%Hence $\Delta^{g_{B}}$ is monotone.

%It remains to
To conclude the proof, we show that $\Delta$ and $\Gamma$ indeed
constitute an isomorphism, that is,
\begin{enumerate}
\item\label{item:lem:inverseExists1}
  $\Delta^{g_{B}}\bigl(\Gamma_{A}(g_{A})\bigr)=g_{A}$
  if $g_{A} \in \mathfrak{G}^{g_{B}}$; and
\item\label{item:lem:inverseExists2}
  $\Gamma_{A}\bigl(\Delta^{g_{B}}(h_{A})\bigr)=h_{A}$
  if $h_{A} \in \mathfrak{H}^{g_{B}}$.
\end{enumerate}

\subparagraph{\ref{item:lem:inverseExists1}}
Let $g_A\in\mathfrak{G}^{g_{B}}$.
Let $h_A=\Gamma_A(g_{A})$ and define
a cone $\bigl(X,(\gamma^{A}_k\colon X_A\kto\overline{F}^k1)_{k\in\omega}\bigr)$
as above.
Note that by definition of $\Delta^{g_B}$,
$\Delta^{g_B}\bigl(\Gamma(g_{A})\bigr) = h^{\dagger}_A$ where
$h^\dagger_A\colon X_A\kto Z$ is the unique mediating arrow from
$\bigl(X,(\gamma^{A}_k\colon X_A\kto\overline{F}^k1)_{k\in\omega}\bigr)$ to
$\bigl(Z,(J\pi_k:Z\kto \overline{F}^k1)_{k\in\omega}\bigr)$.
%defined as above where each .

For each $k\in\omega$, we prove $J\pi_k\odot g_{A}=\gamma^A_k$ by induction on $k$.
If $k=0$, then
\begin{align*}
J\pi_k\odot g_{A}
&=J!_Z\odot g_{A}
& (\text{$\pi_0=!_Z$})\\
&=h_{A}
& (\text{by definition}) \\
&=\gamma^A_k
& (\text{by definition})
\,.
\end{align*}
If $k>0$, we have:
%
%By , we have 
%$l_j=J\zeta\odot \overline{F}[l_1,\ldots,l_n]\odot c_j$.
%
%For each $j\in[1,i]$, 
%For a family $(l_j:X_j\kto Z)_{1\leq j\leq i}$, we have:
\begin{align*}
&J\pi_k\odot g_A \\
&=J\pi_k\odot J\zeta^{-1}\odot \overline{F}[g_A,g_B]\odot c_A
& (\text{$g_A\in\mathfrak{G}^{g_B}$})\\
&=JF\pi_{k-1}\odot\overline{F}[g_A,g_B]\odot c_A
& (\text{$\zeta$ is a mediating arrow})\\
&=\overline{F}[J\pi_{a-1}\odot g_A,J\pi_{a-1}\odot g_B]\odot c_A
& (\text{$JF=\overline{F}J$})\\
&=\overline{F}[\gamma^A_{k-1},J\pi_{k-1}\odot g_A]\odot c_A
&(\text{by induction hypothesis}) \\
&=\gamma^A_k & (\text{by definition})
\,.
\end{align*}
Therefore by uniqueness of the mediating arrow,
we have $g_A=h^\dagger_A$, and
this implies Cond.~\ref{item:lem:inverseExists1}.

\subparagraph{\ref{item:lem:inverseExists2}}
By definition, $\Delta^{g_B}(h_A)=h^\dagger_A$
where each $h^\dagger_A$ is the unique mediating arrow from a cone
$\bigl(X,(\gamma^{A}_k\colon X_A\kto\overline{F}^k1)_{k\in\omega}\bigr)$
to the limit
$\bigl(Z,(J\pi_k\colon Z\kto \overline{F}^k1)_{k\in\omega}\bigr)$ where
the former is defined as above.
%Therefore we have $J\pi_a\odot l^\dagger_j=\gamma^j_a$ for each $j$ and $a$.
Letting $k=0$,
%by definition of $\gamma^j_a$, 
we have:
\begin{align*}
\Gamma_A(h^\dagger_A)
&=J!_{Z}\odot h^\dagger_A  & (\text{by definition}) \\
&=J\pi_0\odot h^\dagger_A  & (\text{$\pi_0=!_Z$}) \\
&=\gamma^A_0  & (\text{$h^\dagger_A$ is a mediating arrow}) \\
&=h_A & (\text{by definition})\,.
\end{align*}
This implies Cond.~\ref{item:lem:inverseExists2}
\end{proof}

\subsection{Proof of Lem.~\ref{lem:eqSysCoincidenceTraceAndAccProb}}
\begin{proof}
 It is straightforward  that 
 $\Kl(\giry)(X,1)$  is both a pointed $\omega$-cpo and a pointed
 $\omega^{\op}$-cpo (here restriction to $\omega$ is crucial for
 compatibility with measurable structures). 
 % The homset $\Kl(\giry)(X,1)$ has both the greatest element
 % $\top:X\kto 1$ and $\bot:X\kto 1$.
 % Here, the former is given by $\top(x)(\{*\})=1$ and 
 % the latter is given by $\bot(x)(\{*\})=0$. 
 Moreover,
 Kleisli composition $\odot$ in $\Kl(\giry)$ is seen to be
 $\omega$- and $\omega^{\op}$-continuous, similarly
 to  the proof of \cite[Prop.~4.20]{BrengosMP15bef}---thus the equational
 system $E'$ in~(\ref{eq:eqSysCoincidenceTraceAndAccProb:eqSysTrace})
 indeed has a solution $\seq{l'^{\sol}_{\i}}{n}$, by Lem.~\ref{lem:eqSysSolvableOmegaCont}.

  \begin{displaymath}
    \vcenter{\xymatrix@R=.6em@C+2em{
      {\overline{F}X}
        \kar[r]^-{\oF g}
      &
      \oF Z
        \kar[d]^{J\zeta^{-1}}
        \kar[r]^-{\oF J\bang_{Z}= JF\bang_{Z}}
        \ar@{}[rd]|{=}
      &
      {\overline{F}1}
        \kar[d]^{J\bang_{F1}}
      \\
      {X}
        \kar[u]^{c_{A}}
        \kar[r]_-{g}
      &
      {Z}
        \kar[r]_-{J\bang_{Z}}
      &
      {1}
    }}
  \end{displaymath}
 % Notice that diagrams in $\mathfrak{G}$ and $\mathfrak{H}$, respectively,
 % clearly correspond to $\Phi$ and $\Psi$.
 Recall the similarity between $\Phi_{\X}, \Psi_{\X}$ and the diagrams
 in~(\ref{eq:mathfrakGAndMathfrakH}). 
 We can prove $\Gamma\co\Phi_{\X}=\Psi_{\X}\co\Gamma$ (where $\Gamma$ is from
 Lem.~\ref{lem:inverseExists}),
 as shown in the above diagram; indeed
 $(\Gamma\co\Phi_{\X})(g) = J\bang_{Z} \kco J\zeta^{-1} \kco \oF g \kco c_{A}$, and
 $(\Psi_{\X}\co\Gamma)(g) = J\bang_{F1} \kco \oF J\bang_{Z} \kco \oF g \kco c_{A}$.
 This discharges  Cond.~\ref{item:homOfEqSysCompatibility} of
 Lem.~\ref{lem:eqSysWithFixedPtIso}, where $E$ and $E'$ are taken as in~(\ref{eq:eqSysCoincidenceTraceAndAccProb:eqSysTrace});
 Cond.~\ref{item:homOfEqSysStrongRestriction} is discharged by
 Lem.~\ref{lem:inverseExists}. 
 Therefore by taking $\Gamma$ as $\varphi$ and
 $\Delta^{[\seqby{l_{\i}}{i+1}{n}]}$ as 
 $\psi^{(\seqby{l_{\i}}{i+1}{n})}$ 
 in Lem.~\ref{lem:eqSysWithFixedPtIso}, we conclude existence of a
 solution
 $\seq{l^{\sol}_{\i}}{n}$ of $E$, and that
 $\Gamma([\seq{l^{\sol}_{\i}}{n}])=[\seq{l'^{\sol}_{\i}}{n}]$. 

 % with $E$ and $E'$ in Lem.~\ref{lem:eqSysWithFixedPtIso} being
 % (\ref{eq:eqSysCoincidenceTraceAndAccProb:eqSysTrace}) and
 % (Def.~\ref{def:acceptedLangParityTFSys});
 % and we let $\varphi=\Gamma$ and
 % $\psi^{\seqby{l_{\i}}{i+1}{n}}=\Delta^{(g_{B}=\seqby{l_{\i}}{i+1}{n})}$.

 Finally we realize that 
 $E$ in~(\ref{eq:eqSysCoincidenceTraceAndAccProb:eqSysTrace})
 is the same one as $E_{\X}$ in Def.~\ref{def:acceptedLangParityTFSys}; 
 therefore 
 $\trp(\X)=[\seq{l^{\sol}_{\i}}{n}]$.
 % We also note that the equational
 % system $E$ in~(\ref{eq:eqSysCoincidenceTraceAndAccProb:eqSysTrace})
 % is the same one as $E_{\X}$ (in Def.~\ref{def:acceptedLangParityTFSys});
 % thus we have
 % $\Gamma(\trp(\X))=\Gamma([\seq{l^{\sol}_{\i}}{n}])=[\seq{l'^{\sol}_{\i}}{n}]$.
\end{proof}

\subsection{Proof of Lem.~\ref{lem:fixedPtCharacterizationOfAcceptanceProb}}
\begin{proof}
Without loss of generality, we can assume that $n$ is even.
%For the sake of simplicity we assume that $n$ is even.
We shall append a state $\dlstate$ and a unary letter $\dlletter$, that
represent divergence explicitly, by trapping every divergence
into the non-accepting infinite loop $(\dlletter,\dlstate)(\dlletter,\dlstate)\cdots$.

More concretely, we define a new PPTA
$\X_{\dlstate}=((X_1,\dotsc,X_n,\{\dlstate\}),\Sigma+(\dlletter),\delta_{\dlstate},s)$,
where
\begin{math}
  \delta^{\dlstate}\colon
  (X+\{\dlstate\}) \to \giry\bigl(\textstyle\coprod_{\sigma\in\Sigma+(\dlletter)}X^{|\sigma|}\bigr)
\end{math}
is defined as follows. %by the following.
\begin{align*}
&\delta^{\dlstate}(x)(\sigma,(\seq{x_{\i}}{n})):=  \\
&\qquad
\begin{cases}
\delta(x)\bigl(\sigma,(\seq{x_{\i}}{n})\bigr) & (x,\seq{x_{\i}}{n}\in X, \sigma\in\Sigma) \\
1-\sum\nolimits_{
    (\sigma,(\seq{x_{\i}}{|\sigma|}))
    \in \coprod_{\sigma\in \Sigma} X^{|\sigma|}}
  \delta\bigr(\sigma,(\seq{x_{\i}}{|\sigma|})\bigl) & (n=1, x\in X,  x_1=\dlstate, \sigma=\dlletter) \\
1 &   (n=1, x=x_1=\dlstate, \sigma=\dlletter) \\
0 & (\text{otherwise}) \,.
\end{cases}
\end{align*}
%\begin{align*}
%  &\delta^{\dlstate}(x)\bigl(\sigma,(\seq{x_{\i}}{n})\bigr)
%  &&\hspace{-7em}\;\coloneqq\;
%  \delta(x)\bigl(\sigma,(\seq{x_{\i}}{n})\bigr)
%  \\
%  &\delta^{\dlstate}(x)\bigl(\sigma,(\seq{x_{\i}}{n})\bigr)
%  &&\hspace{-7em}\;\coloneqq\;
%  \textstyle
%  \sum\nolimits_{
%    (\sigma,(\seq{x_{\i}}{|\sigma|}))
%    \in \coprod_{\sigma\in \Sigma} X^{|\sigma|}
%  }
%  \delta\bigr(\sigma,(\seq{x_{\i}}{|\sigma|})\bigl)
%  \\
%  &\delta^{\dlstate}(x)\bigl(\sigma,(\spadesuit)\bigr)
%  &&\hspace{-7em}\;\coloneqq\;
%  0
%  \\
%  &\delta^{\dlstate}(x)\bigl(\dlletter,(\spadesuit)\bigr)
%  &&\hspace{-7em}\;\coloneqq\;
%  1
%\end{align*}
Notice that $\{\dlstate\}$ has an \emph{odd} priority $n+1$ that is \emph{maximum}.
Let $\seq{\tilde{l}^{\sol}_{\i}}{n+1}$ be the solution of 
%We shall consider 
the following equational system
over $[0,1]^{X+\{\dlstate\}}$.
%and let $\seq{\tilde{l}^{\sol}_{\i}}{n+1}$ be its solution.
\begin{equation}\label{eq:eqSysPsiWithDL}
  \begin{array}{rll}
    u'_{1}
    &=_{\mu}&
    \Psi'_{\X_{\dlstate}}([u'_{1},\cdots,u'_{n},u'_{n+1}]) \upharpoonright X_{1}
    \\
    &\;\vdots&
    \\
    u'_{n}
    &=_{\nu}&
    \Psi'_{\X_{\dlstate}}([u'_{1},\cdots,u'_{n},u'_{n+1}]) \upharpoonright X_{n}
    \\
    u'_{n+1}
    &=_{\mu}&
    \Psi'_{\X_{\dlstate}}([u'_{1},\cdots,u'_{n},u'_{n+1}]) \upharpoonright \{\dlstate\}
  \end{array}
\end{equation}
%
%Regarding this $\delta^{\dlstate}$ and
%$X_{1}+\cdots+ X_{n}+\{\dlstate\}$,
%we define $\Psi^{\dlstate}_{\X}:[0,1]^{X+\dlstate}\to[0,1]^{X+\dlstate}$ 
%in a similar manner to $\Psi'_\X$.
%$\Psi^{\dlstate}_{\X}$ is defined obviously.
%
The 
%last part of 
$(n+1)$-th
solution 
$\tilde{l}^{\sol}_{n+1}$ is $[\dlstate \mapsto 0]$,
since it is defined by the least fixed point of the identity function.
%namely, $\lfp[\{\dlstate \mapsto p\} \mapsto \{\dlstate \mapsto p\}]$.
Thus we can ignore the last equation and obtain the following equational system,
without changing the other part of the solution $\seq{\tilde{l}^{\sol}_{\i}}{n}$.
\begin{equation*}%\label{eq:fixedPtCharacterizationOfAcceptanceProb:eqSys}
  \begin{array}{rll}
    u'_{1}
    &=_{\mu}&
    \Psi'_{\X_{\dlstate}}([u'_{1},\cdots,u'_{n},[\dlstate\mapsto 0]]) \upharpoonright X_{1}
    \\
    &\;\vdots&
    \\
    u'_{n}
    &=_{\nu}&
    \Psi'_{\X_{\dlstate}}([u'_{1},\cdots,u'_{n},[\dlstate\mapsto 0]]) \upharpoonright X_{n}
  \end{array}
\end{equation*}
It is easy to see that
\begin{math}
  \Psi'_{\X_{\dlstate}}(\seq{l_{\i}}{n},[\dlstate\mapsto 0])
  =
  \Psi'_{\X}(\seq{l_{\i}}{n})
\end{math}.
Thus the solution $\seq{\tilde{l}^{\sol}_{\i}}{n}$
coincides with  $\seq{l^{\sol}_{\i}}{n}$.

We shall define $\Run_{\X}^{\dlstate}$,
in the similar manner to $\Run_{\X_{\dlstate}}$ (Def.~\ref{def:runOfProbTreeAutom}),
except that
%except the notable difference:
any $\rho\in\Run_{\X_{\dlstate}}$ that contains a
label $(\sigma,\dlstate)$ where $\sigma\in\Sigma$ does \emph{not}
belong to $\Run^{\dlstate}_{\X}$.
(Recall that in the current \emph{probabilistic} setting,
$\Run_{\X}$ is defined to permit \emph{arbitrary} transitions
between the states.)

We augment the equational system~(\ref{eq:eqSysCharacterizationOfAcceptingRuns:eqSys})
(in Lem.~\ref{lem:eqSysCharacterizationOfAcceptingRuns}),
which characterizes the accepting runs, with $\dlstate$.
Though the system~(\ref{eq:eqSysCharacterizationOfAcceptingRuns:eqSys})
is defined with respect to $Run_{\X}$ of an NBTA $\X$,
its definition naturally extends to runs of PBTAs.
The definition of this augmented equational system is as follows.
%We also augment the equational
%system~(\ref{eq:eqSysCharacterizationOfAcceptingRuns:eqSys})
%(in Lem.~\ref{lem:eqSysCharacterizationOfAcceptingRuns})
%with $\dlstate$, as follows. %shown in what follows.
\begin{equation}\label{eq:eqSysBoxWithDL}
  \begin{array}{rll}
    u_{1}
    &=_{\mu}&
    \Diamond_{\X_{\dlstate}}(u_{1} \cup\cdots\cup u_{n}\cup \{\dlstate\}) \cap \Run^{\dlstate}_{\X,X_{1}}
    \\
    &\;\vdots&
    \\
    u_{n}
    &=_{\nu}&
    \Diamond_{\X_{\dlstate}}(u_{1} \cup\cdots\cup u_{n} \cup \{\dlstate\}) \cap \Run^{\dlstate}_{\X,X_{n}}
    \\
    u_{n+1}
    &=_{\mu}&
    \Diamond_{\X_{\dlstate}}(u_{1} \cup\cdots\cup u_{n} \cup \{\dlstate\}) \cap \Run^{\dlstate}_{\mathcal{X},\{\dlstate\}}
  \end{array}
\end{equation}
Much like in the last case of~(\ref{eq:eqSysPsiWithDL}),
we can easily see that the (non-last) solution of
the equational system (\ref{eq:eqSysBoxWithDL}) coincides with
one of~(\ref{eq:eqSysCharacterizationOfAcceptingRuns:eqSys}),
which is $\AccRun_{\X}$.
Note that here the definition of $\Run^{\dlstate}_{\X}$, which
excludes a run with a $(\sigma,\dlstate)$-labeled node, is crucial.

Now we aim to apply Lem.~\ref{lem:eqSysWithOmegaChainMap},
sending the solution of (\ref{eq:eqSysBoxWithDL}) (accepting runs)
to one of (\ref{eq:eqSysPsiWithDL}) (acceptance probabilities), by
\begin{math}
  \mu^{\Run}_{\X_{\dlstate},\place}
% \colon x\mapsto
% \mu^{\Run}_{\X_{\dlstate},x}
\end{math}.
Notice that first: for the equational system (\ref{eq:eqSysBoxWithDL}),
each interim solution can be defined as
either the $\omega$-supremum or the $\omega$-infimum
(as in the proof of Lem.~\ref{lem:eqSysCharacterizationOfAcceptingRuns}),
essentially because $\Diamond_{\X_{\dlstate}}$ is both $\omega$-continuous and $\omega^{\op}$-continuous;
thus (\ref{eq:eqSysBoxWithDL}) can be solved within measurable spaces.
This observation is required, since
\begin{math}
  \mu^{\Run}_{\X_{\dlstate},\place}
\end{math}
is defined only over measurable sets of runs.
Preservation of $\bot$, is almost trivial; and
$\Psi_{\X_{\dlstate}}$ and $\mu^{\Run}_{\X_{\dlstate},\place}$
are both $\omega$-continuous and $\omega^\op$-continuous
by measurability.
%by elementary measure theory.

The other conditions required in
Lem.~\ref{lem:eqSysWithOmegaChainMap} are as follows.
\begin{itemize}
  \item Commutativity:
    \hspace{1em}
    \begin{math}
      \mu^{\Run}_{X_{\dlstate},\place}
      (\Diamond_{\X_{\dlstate}} R)
      \;=\;
      \Psi^{\dlstate}_{\X}
      \left(\mu^{\Run}_{\X_{\dlstate},\place}(R)\right)
    \end{math}
    \enskip
    for
    $R\in\pow(\Run^{\dlstate}_{\X})$
  \item Preservation of $\top$:
    \hspace{0.05em}
    \begin{math}
      \mu^{\Run}_{\X_{\dlstate},\place}
      \left(\Run_{\X_{\dlstate}}\right)
      \;=\;
      1
    \end{math}
\end{itemize}
The commutativity condition is easily seen; and
the preservation of $\top$ is due to the definition of $\delta_{\dlstate}$---in which
the ``missing'' probability is filled by the transitions to $\dlstate$.

Then by applying Lem.~\ref{lem:eqSysWithOmegaChainMap}, we have
\begin{equation*}
  \mu^{\Run}_{\X_{\dlstate},\place}
  (\AccRun_{X,i})
  \;=\;
  l'^{\sol}_{i}
  \,.
\end{equation*}
Since
\begin{math}
  \AccProb(x)
  =
  \mu^{\Run}_{\X,x}(\AccRun_{X})
\end{math}
by definition,
it suffices to show, for any $x\in X$,
\begin{equation*}
  \mu^{\Run}_{\X,x}(\AccRun_{\X})
  \;=\;
  \mu^{\Run}_{\X_{\dlstate},x}(\AccRun_{\X})
  \,.
\end{equation*}

In fact, thanks to measurability, we only need to show that
for any partial run $\xi$ of $\X$:
\begin{equation}\label{eq:muRunDLCoincidesForCylXi}
  \mu^{\Run}_{\X,x}(\Cyl_{\X}(\xi))
  \;=\;
  \mu^{\Run}_{\X_{\dlstate},x}(\Cyl_{\X}(\xi))
  \,.
\end{equation}
We note that $\Cyl_{\X}(\xi)$ does \emph{not} contain any of $\dlletter$ or $\dlstate$,
%just 
because $\xi$ is a run of $\X$ and is not a run of $\X_{\dlstate}$.
Therefore, by the inductive definition of $\mu^\Run_\X$ in Def.~\ref{def:NoDivergence},
(\ref{eq:muRunDLCoincidesForCylXi}) can be straightforwardly confirmed.
This concludes the proof.
%Let us explain of the difference between $\mu^{\Run}_{\X,x}$ and
%$\mu^{\Run}_{(X+\{\dlstate\},\Sigma\cup\{\dlletter\},\delta^{\dlstate},s),x}$:
%the former assigns $\NDL_{\X}$ to a $(x',\ast)$-labeled state,
%and the latter assigns 1 to such a state.
\end{proof}

\subsection{Proof of Thm.~\ref{thm:coincidenceForPPTA}}
\begin{proof}
We identify $\X$ with a $(\giry,\FSigma)$-system
$\bigl((\seq{X_{\i}}{n}),\delta\colon X\kto\overline{\FSigma}X,s\colon 1\kto X\bigr)$, and
let $1=\{\bullet\}$.
We can easily see that
$\Psi_{\X}$ 
(in Lem.~\ref{lem:eqSysCoincidenceTraceAndAccProb})
and $\Psi'_{\X}$
(in Lem.~\ref{lem:fixedPtCharacterizationOfAcceptanceProb})
define exactly the same function.
Therefore, by the claim of these two lemmas, we have
\begin{math}
  \Gamma\bigl(\trp(\X)\bigr) = \AccProb_{\X}
\end{math}.

Now we note the following:
\begin{multline*}
  \Gamma\bigl([x \mapsto \mu^{\myTree}_{\X,x}]\bigr)
  \;=\;
  J\bang_{\myTree}\kco\bigl([x \mapsto \mu^{\myTree}_{\X,x}]\bigr)
  \;=\;
  \mu^{\myTree}_{\X,x}(\myTree_\Sigma)
  \;=\;
  \\
  \mu^{\myTree}_{\X,x}(\Cyl_{\Sigma}(\ast))
  \overset{\text{Def.~\ref{def:NoDivergence}}}{=}
  \mu_{\X,x}^{\Run}\bigl(\,
    \DelSt^{-1}(\Cyl_{\Sigma}(\ast))\cap\AccRun_{\X}
  \,\bigr)
  \,,
\end{multline*}
where $\ast$ denotes the partial tree consisting of one node labeled by
 $\ast$ (``continuation'', Def.~\ref{def:measurableStrOfTreeAndRunBriefly}).
As $\DelSt^{-1}(\Cyl_{\Sigma}(\ast))$ is nothing but the set of all runs $\Run_{\X}$,
we have
\begin{equation*}
  \Gamma\bigl([x \mapsto \mu^{\myTree}_{\X,x}]\bigr)
  \;=\;
  \mu_{\X,x}^{\Run}\bigl(\,
    \AccRun_{\X}
  \,\bigr)
  \;=\;
  \AccProb_{\X}
  \,
\end{equation*}
by the definition of $\AccProb_{\X}$
(in Lem.~\ref{lem:fixedPtCharacterizationOfAcceptanceProb}).

Combining the above two facts 
 we obtain
$\Gamma\bigl(\trp(\X)\bigr)=\Gamma\bigl([x \mapsto \mu^{\myTree}_{\X,x}]\bigr)$.
Recall that there is an inverse of $\Gamma$, namely
$\Delta$ in Lem.~\ref{lem:inverseExists}; this yields
$\trp(x) = \mu^{\myTree}_{\X,x}$.

Now the claim is immediate, as below, where we have only to consider
 cylinder sets $\Cyl_{\Sigma}(\lambda)$ that generate the relevant
 $\sigma$-algebra. 
\begin{equation*}
  \trp(\X)(\bullet)(\Cyl_{\Sigma}(\lambda))
  \;=\;
  \sum_{x\in X}s(x)\cdot\mu^{\myTree}_{\X,x}(\Cyl_{\Sigma}(\lambda))
  \;=\;
  \mu^{\myTree}_{\X}(\Cyl_{\Sigma}(\lambda))
  \;=\;
  \Lang(\X)(\Cyl_{\Sigma}(\lambda))
  \qedhere
\end{equation*}
\end{proof}

\fi
\end{document}